\documentclass{fundam}

\usepackage{amsmath,amssymb}
\usepackage{lineno}
\usepackage{stmaryrd}

\usepackage{xspace}
\usepackage{multicol}
\usepackage{todonotes}
\usepackage{xspace}
\usepackage{xcolor}
\usepackage{amsmath}
\usepackage{amssymb}
\usepackage{mathtools}
\usepackage{mathtools}
\usepackage{mathrsfs}
\usepackage{url}
\usepackage{booktabs}
\usepackage{adjustbox}
\usepackage[utf8]{inputenc}
\usepackage{amsbsy}
\usepackage{stmaryrd}
\usepackage{multicol}

\def \usym2613{\bf{\times}}

\usepackage{lineno}
\usepackage{cleveref}

\usepackage{tikz}
\usetikzlibrary{shapes,snakes}
\usepackage{xcolor}

\newcommand{\Trace}[1]{\ensuremath{\langle{#1}\rangle}\xspace}
\newcommand{\trace}[1][]{\ensuremath{\sigma_{#1}}\xspace}
\newcommand{\sactivity}{{\ensuremath{\blacktriangleright}}\xspace}
\newcommand{\eactivity}{\ensuremath{\mathsmaller{\blacksquare}}\xspace}
\newcommand{\set}[1]{\ensuremath{\{ {#1} \}}\xspace}
\newcommand{\without}[1]{\ensuremath{\backslash{#1}}\xspace} 
\newcommand{\tuple}[1]{\ensuremath{( {#1})}\xspace} 
\newcommand{\sequence}[1]{\ensuremath{\langle {#1}\rangle}\xspace} 
\newcommand{\powerset}[1]{\ensuremath{\mathbb{P}({#1})}\xspace} 
\newcommand{\multisetset}[1]{\ensuremath{\mathbb{M}({#1})}\xspace} 
\newcommand{\sizeof}[1]{\ensuremath{{|{#1}|}}\xspace} 

\newcommand{\place}[1][]{\ensuremath{p_{#1}}\xspace} 

\newcommand{\setofplaces}[1][]{\ensuremath{{P}_{#1}}\xspace} 
\newcommand{\setofactivities}[1][]{\ensuremath{A_{#1}}\xspace} 
\newcommand{\universeoftraces}{\ensuremath{\mathcal{T}}\xspace}
\newcommand{\net}{\ensuremath{N}\xspace}

\newcommand{\inoutpair}[2]{\ensuremath{({#1} |  {#2})}\xspace} 
\newcommand{\fire}[1]{\ensuremath{\xrightarrow{#1}}\xspace}

\renewcommand{\log}[1][]{\ensuremath{L_{#1}}\xspace} 

\newcommand{\underfedsubscript}[3][]{\ensuremath{{\triangledownsubscript^{#1}_{\! \!#2}{(#3)}}\xspace}} 
\newcommand{\underfed}[3][]{\ensuremath{{\triangledown^{#1}_{#2}{(#3)}}\xspace}} 
\newcommand{\overfed}[3][]{\ensuremath{\triangle^{#1}_{#2}{(#3)}}\xspace} 
\newcommand{\fitting}[3][]{\ensuremath{\Box^{#1}_{#2}{(#3)}}\xspace} 

\newtheorem{observation}{Observation}[section]
\newcommand{\Qlength}{Q^\leq}
\newcommand{\proj}[2]{\ensuremath{#1{\upharpoonright}_{#2}}\xspace}

\def\figdist{2mm}

\def\fitnessmetric{\texttt{fm}\xspace}
\def\pFit{\texttt{fitting}\xspace}
\def\pNeg{\texttt{underfed}\xspace}
\def\pPos{\texttt{overfed}\xspace}

\renewcommand{\triangledown}{\begin{sideways}\begin{sideways}$\triangle$\end{sideways}\end{sideways}}
\newcommand{\triangledownsubscript}{\begin{sideways}\begin{sideways}$_\triangle $\end{sideways}\end{sideways}}

\newcommand{\algorithmname}[1]{\texttt{#1}\xspace} 

\def\universeofactivities{\mathcal{A}}

\usepackage{relsize}
\usepackage{rotating}

 \usepackage{hyperref}

\begin{document}


\setcounter{page}{109}
\publyear{24}
\papernumber{2168}
\volume{190}
\issue{2-4}

   \finalVersionForARXIV


\title{Discovering Process Models with Long-Term Dependencies while Providing Guarantees and Filtering Infrequent Behavior Patterns}

\author{Lisa L. Mannel\thanks{Address for correspondence: Process and Data Science (PADS),
                     RWTH Aachen University, Aachen, Germany}, \ Wil M. P. van der Aalst
\\
Process and Data Science (PADS) \\
RWTH Aachen University \\
Aachen, Germany\\
mannel@pads.rwth-aachen.de, \ wvdaalst@pads.rwth-aachen.de}

\maketitle

\runninghead{L.L. Mannel and W. M.P. van der Aalst}{Discovering Process Models while Providing Guarantees...}

\begin{abstract}
In process discovery, the goal is to find, for a given event log, the model describing the underlying process. While process models can be represented in a variety of ways, Petri nets form a  theoretically well-explored description language and are therefore often used. In this paper, we extend the eST-Miner process discovery algorithm. The eST-Miner computes a set of Petri net places which are considered to be fitting with respect to a certain fraction of the behavior described by the given event log as indicated by a given noise threshold. It evaluates all possible candidate places using token-based replay. The set of replayable traces is determined for each place in isolation, i.e., these sets do not need to be consistent. This allows the algorithm to abstract from infrequent behavioral patterns occurring only in some traces. However, when combining places into a Petri net by connecting them to the corresponding uniquely labeled transitions, the resulting net can replay exactly those traces from the event log that are allowed by the combination of all inserted places. Thus, inserting places one-by-one without considering their combined effect may result in deadlocks and low fitness of the Petri net. In this paper, we explore adaptions of the eST-Miner, that aim to select a subset of places such that the resulting Petri net guarantees a definable minimal fitness while maintaining high precision with respect to the input event log. Furthermore, current place evaluation techniques tend to block the execution of infrequent activity labels. Thus, a refined place fitness metric is introduced and thoroughly investigated. In our experiments we use real and artificial event logs to evaluate and compare the impact of the various place selection strategies and place fitness evaluation metrics on the returned Petri net.
\end{abstract}


\section{Introduction and related work} \label{sec:fm:intro}
\noindent
More and more corporations and organizations support their processes using information systems, which record the occurring behavior and represent this data in the form of \emph{event logs}.
Each event in such a log has a name identifying the executed activity (activity name), an identifier mapping the event to some execution instance (case id), a timestamp showing when
the event was observed, and often extended meta-data of the activity or process instance. Based on the timestamp and case id, we can organize the event log as sequences of activities, also called \emph{traces}, which represent example execution sequences of the process.
In the field of \emph{process discovery}, we utilize the event log to identify relations between the activities (e.g., pre-conditions, choices, concurrency), which are then expressed within a process model. While several modeling formalisms exist, this work focuses on modeling processes using Petri nets~\cite{ReisigPN, desel1998petri, freechoicebook, bert87}.

Process discovery is non-trivial for various reasons. We cannot assume that the given event log is complete, as some possible behavior might be yet unobserved. Also, real-life event logs often contain noise in the form of incorrectly recorded data or deviant behavior, which is usually not desired to be reflected in the process model. However, correctly classifying behavior as noise is often not trivial.
Several quality dimensions have been defined to evaluate process models (\cite{PMbook}) which are partially in conflict with each other and thus it is usually not possible to achieve perfect values for all aspects. A process model with good \emph{fitness} can (mostly) reproduce the  behavior contained in the event log, while high \emph{precision} corresponds to not allowing for (much) unobserved behavior. Furthermore, a model with good \emph{generalization} is expected to express behavior possible in the process that generated the event log even if it has not yet been observed. Finally, since process models are often used in contexts that require interpretation by a person, model \emph{simplicity} is of interest.
Additionally, an ideal discovery algorithm should be reasonably time and space efficient.
Since real-life event logs rarely allow for a model that perfectly satisfies all the different quality criteria, different discovery algorithms focus on different aspects, while neglecting others. As a result, the models returned by these discovery algorithms for a given event log can differ significantly. Note that generally there is no clearly definable best process model but instead this choice depends on what purpose it should serve.

Many existing discovery algorithms abstract from the full information given in a log or generate places heuristically, in order to decrease computation time and complexity of the returned process models. While this is convenient in many settings, the resulting models are often underfitting, in particular when processes are complex. Examples are the Alpha Miner variants~\cite{Wen2007}, the Inductive Mining family~\cite{ind}, genetic algorithms or Heuristic Miner~\cite{DBLP:FHM}. In contrast to these approaches, which are not able to (reliably) discover complex model structures, algorithms based on region theory~\cite{bookonregions, lbregions1, lbregions, ilp, ilp2, ilp3, carmonaminregions, darondeau, bergenthum, rozenberg, carmonaboundednets}) discover models whose behavior is the minimal behavior representing the input event log. On the downside, these approaches are known to be rather time-consuming, cannot handle noise, and tend to produce complex, overfitting models which can be hard to interpret. A combination of strategies has been introduced in~\cite{kalenkova20}, which aims to circumvent performance issues by limiting the application of region theory to small fragments of a pre-discovered Petri net.

In~\cite{PN2019} we introduced the discovery algorithm eST-Miner. This approach aims to combine the capability of \emph{finding complex control-flow structures like long-term dependencies} (non-free choice constructs) with an inherent ability to \emph{filter infrequent behavior patterns} while exploiting the token-game to \emph{increase efficiency}.
The basic idea is to construct a Petri net without any places and one uniquely labeled transition for each activity in the event log, and then compute and insert a set of fitting places. A place can be considered fitting even if it disagrees with a part of behavior in the event log, which allows to filter infrequent behavior. Efficiency is significantly increased by skipping uninteresting parts of the search space.
The approach guarantees that all places considered fitting with respect to the event log are discovered, and can thereby provide some kind of precision guarantee. In fact, in the case where fitting places are defined as feasible places, i.e., the discovered Petri net is required to reflect all behavior in the input event log, it guarantees to return a minimal over-approximation.

The eST-Miner evaluates each candidate place in isolation, i.e., for each place it focuses only on the activities in the event log whose transitions are connected to that place. This allows us filter infrequent behavior patterns in the event log, by requiring each place to be able to replay only parts of the traces in the event log. A candidate place will be accepted and inserted into the returned Petri net, if the event log contains sufficient support for the relation between the activities as defined by the place.
Common noise filtering techniques, which are often applied as a preprocessing step before applying a discovery algorithm lose information by simply removing infrequent trace variants or infrequent activities from the event log.  In contrast, the eST-Miner can consider all information in the event log to discover relations between activities. In particular, we aim to accurately represent behavior patterns with high support in the event log, while ignoring infrequent deviations.
\noindent
\begin{figure}[tbh]\label{fig:examplefilter}
\noindent
\parbox{0.3\linewidth}{
\noindent
\begin{align*}
L=[
&\Trace{\sactivity, x_1, x_2, x_3, \textcolor{red}{y_3, y_2, y_1}, \eactivity}^1,\\
&\Trace{\sactivity, x_1, x_2, x_3, \textcolor{red}{y_3, y_1, y_2}, \eactivity}^1,\\
&\Trace{\sactivity, x_1, x_2, x_3,\textcolor{red}{ y_2, y_3, y_1}, \eactivity}^1,\\
&\Trace{\sactivity, x_1, x_2, x_3,\textcolor{red}{ y_2, y_1, y_3}, \eactivity}^1,\\
&\Trace{\sactivity, x_1, x_2, x_3, \textcolor{red}{y_1, y_3, y_2}, \eactivity}^1,
\end{align*}
}
\parbox{0.3\linewidth}{
\noindent
\begin{align*}
&\Trace{\sactivity, \textcolor{red}{x_3, x_2, x_1}, y_1, y_2, y_3, \eactivity}^1,\\
&\Trace{\sactivity, \textcolor{red}{x_3, x_1, x_2}, y_1, y_2, y_3, \eactivity}^1,\\
&\Trace{\sactivity, \textcolor{red}{x_2, x_3, x_1}, y_1, y_2, y_3, \eactivity}^1,\\
&\Trace{\sactivity, \textcolor{red}{x_2, x_1, x_3}, y_1, y_2, y_3, \eactivity}^1,\\
&\Trace{\sactivity, \textcolor{red}{x_1, x_3, x_1}, y_1, y_2, y_3, \eactivity}^1
]
\end{align*}
}
\centering
\includegraphics[width =0.9\linewidth, trim={0cm 8.75cm 0 8.75cm},clip]{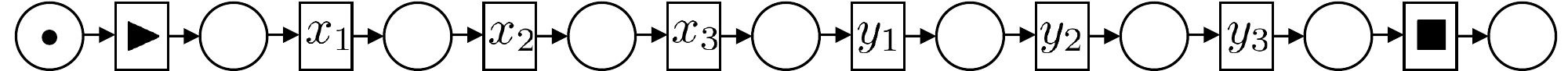}
\caption{The behavior in event log $L$ corresponds in large parts to the sequential Petri net below. However, in all traces some deviations in activity order occur (marked in red). Since all traces and all activities are equally frequent, it is not possible to filter infrequent behavior patterns and discover the underlying main process structure by simply removing infrequent traces or activities in a preprocessing step. This becomes even more challenging for processes that include concurrency, choice or non-free choice constructs.}
\label{fig:exampleNoiseFiltering}
\end{figure}

Consider the simple example event log and Petri net in Figure~\ref{fig:exampleNoiseFiltering}. The event log exhibits frequent behavioral patterns which are reflected in the behavior of the Petri net. However, in all traces some deviations in activity order occur (marked in red). Often, users do not want a discovery algorithm to return a model including all possible behavior in an event log but rather only the main process structure with noise or deviations filtered out. In this example, it is not possible to discover the main process behavior by simply removing infrequent traces or activities in a preprocessing step. In more complex processes, exhibiting concurrency, choice or non-free choice constructs, this becomes even more challenging. However, the place-wise perspective of the eST-Miner allows it to ignore deviating parts of a trace not involved with the place currently evaluated. Each place included in the presented Petri net is able to replay a large part of the presented event log and can therefore be discovered when considered in isolation.

The place-wise perspective of the eST-Miner ensures that all occurrences of activities within a log are considered when discovering a model. However, when the set of discovered fitting places is combined in a Petri net, this Petri net allows only for the behavior in the intersection of the behaviors allowed by all inserted places. This intersection may be small or even empty, thus, the  Petri net may contain deadlocks or dead parts, resulting in a much lower overall fitness than the fitness of each individual place and an overly complicated model. In extreme cases, the constructed net cannot replay any trace at all, as illustrated by the small example in Figure~\ref{fig:example0}. Here, having inserted places $p_6$ and $p_7$ into the Petri net, each of which allows execution of at least $40$ traces from the example event log when considered in isolation, together they cause a deadlock and the Petri net discovered cannot fire any transition after the start transition.

\begin{figure}[tbh]
\centering
\includegraphics[width =0.45\linewidth, trim={1cm 2cm 2.5cm 2.5cm}, clip]{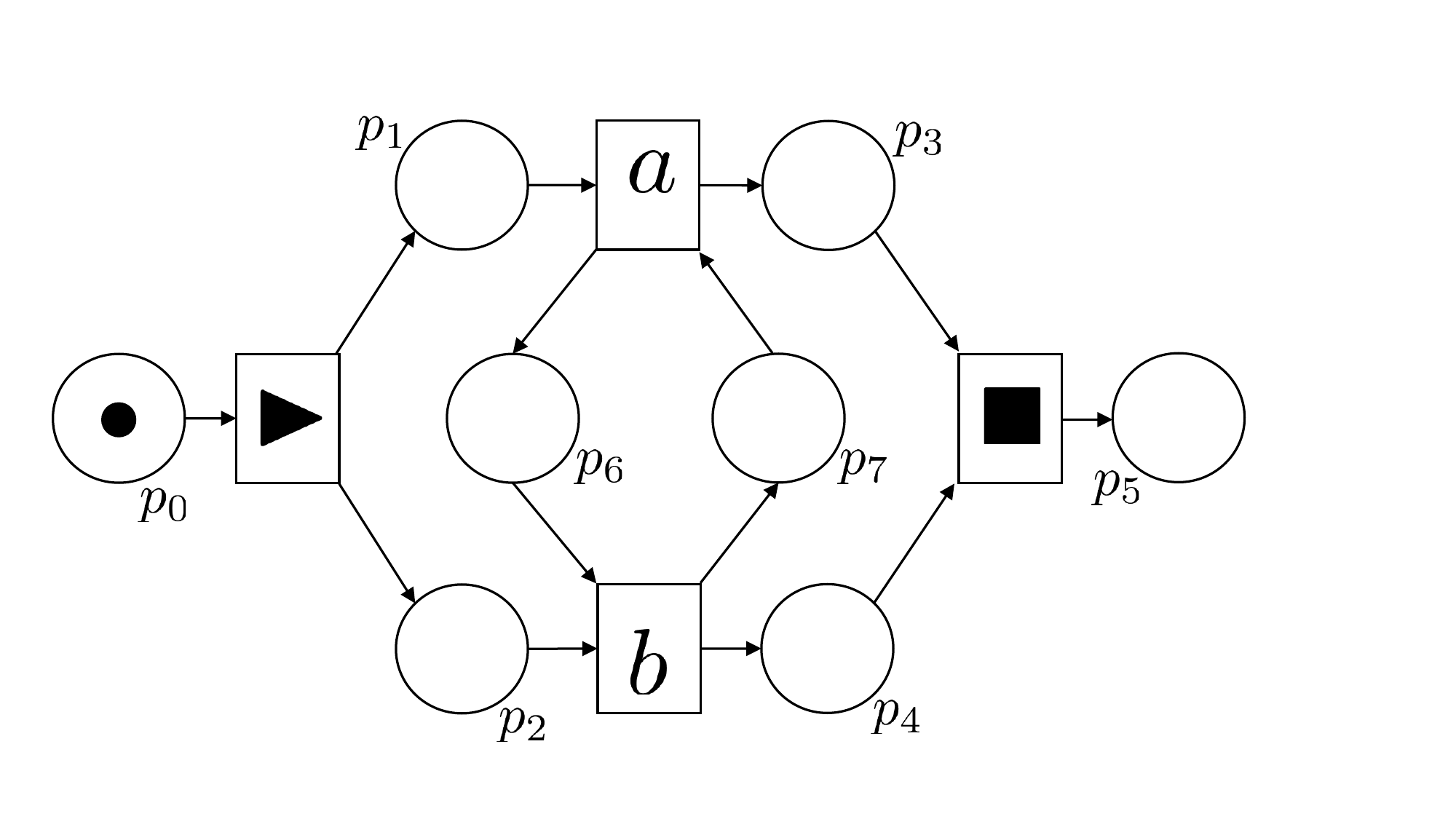}
\caption{{Consider the event log $L=[\langle \sactivity, a, b, \eactivity \rangle^{40}, \langle \sactivity, b, a, \eactivity \rangle^{60}]$, where the first trace variant occurs $40$ times and the second one $60$ times. Considered in isolation, place $p_6$ allows for the first sequence of activities while place $p_7$ allows for the second. However, in combination they cause a deadlock in the Petri net.}}
\label{fig:example0}
\end{figure}

In this paper, we aim to remedy the issue by investigating strategies of selecting a subset of places which can be combined into a deadlock-free Petri net with definable minimal fitness, while simultaneously striving for high precision and simplicity. Additionally, we introduce a new fitness metric used to evaluate candidate places with the goal to better control how candidate place selection affects the inclusion of infrequent activities into the discovered model. We require the algorithm to maintain its ability to discover and model non-free choice constructs and to provide guarantees as indicated above without over- or underfitting. Furthermore, the time and space consumption should remain reasonable, in particular more scalable than classic region theory approaches.

\medskip
This paper is a significantly extended and revised version of our work presented in~\cite{Mannel22}. With respect to this previous work, the main extensions are:
\begin{itemize}

    \item Section~\ref{sec:fm:framework} details the root causes of deadlocks and dead parts in the discovered Petri net and introduces our twofold solution approach.
    \item In Section~\ref{sec:fm:fitness} we introduce a new fitness metric for place evaluation to refine the handling of infrequent activities and thoroughly investigate its properties. In particular, we aim to be able to maintain the eST-Miner's ability to abstract from deviations and avoid overfitting, while at the same time preventing infrequent activities to cause dead parts in the returned Petri net. We explore the properties of the new metric and show that it can be directly incorporated into the eST-Mining framework. This refines the place candidate evaluation and enables significant improvements of the eST-Miners approach to noise filtering.
    \item We significantly extend the introduction to the eST-Miner (Section~\ref{sec:fm:eSToverview}), to better motivate our solution approach and provide a basis for the newly introduced fitness metric for place evaluation.
    \item Section~\ref{sec:fm:selection} has been partially reworked to clarify the goals and properties of the place selection strategy and correct a mistake related to the adaption functions.
    \item We refine and extend the experimental evaluation of our previous work and provide an in-depth discussion of the results, in particular with respect to their interpretation in the context of the applied quality metrics. Furthermore, a brief analysis of the algorithm's running time is added.
    \item We extend and refine the definitions and concepts to make the paper self-contained. Furthermore, Section~\ref{sec:fm:framework} provides a discussion of the context and assumptions we make, together with detailed examples and motivation for the proposed extensions and refinements. Examples and motivation are extended in other sections as well.
\end{itemize}
Section~\ref{sec:fm:prelims} provides basic notation and definitions.
In Section~\ref{sec:fm:eSToverview}, we briefly review the basics of the standard eST-Miner before providing a detailed problem description and method overview in Section~\ref{sec:fm:framework}. The newly introduced fitness metric and its incorporation into the eST-Miner framework are investigated in Section~\ref{sec:fm:fitness}. Section~\ref{sec:fm:selection}     presents our place selection approach, followed by an extensive evaluation in Section~\ref{sec:fm:eval}.  Finally, we discuss design choices, open questions and future work in Section~\ref{sec:fm:discussion} before concluding this work with Section~\ref{sec:fm:conclusion}.

\section{Basic notations, event logs, and process models}\label{sec:fm:prelims}
\noindent
A set, e.g. \set{a_1, a_2, \dots a_n}, does not contain any element more than once. In contrast, a multiset $m \colon A \rightarrow \mathbb{N}_0$ over the set $A$ can contain multiples of the same element. We use square brackets to denote multisets with the multiplicity of an element indicated in its exponent, i.e., we write  $[a_1^{m(a_1)}, a_2^{m(a_2)}, \dots, a_n^{m(a_n)}]$. For readability, elements with multiplicity $0$ and exponents with value $1$ are omitted and we write $a \in m$ if $m(a) \geq 1$ holds. The intersection of two sets contains only elements that occur in both sets, i.e., $\{x,y\} \cap \{y, z\}=\{y\}$, while the intersection of two multisets contains each element with its minimum frequency, i.e, $[x,y^2,z] \nplus [y^5,z^2]=[y^2, z]$. Similarly, the union of two sets contains all elements in both sets once, i.e., $\{x,y\} \cup \{y, z\}=\{x,y,z\}$, while the union of two multisets contains all elements with the sum of their frequencies, i.e, $[x,y^2,z] \uplus [y^5,z^2]=[x, y^7, z^3]$.
By \powerset{X} we refer to the power set of the set $X$, and \multisetset{X} is the set of all multisets over $X$.
We project a multiset $m$ onto a set $A$ by removing all elements not contained in $A$ from $m$ while maintaining other frequencies, e.g., for $A=\{x, y\}$ we have $\proj{[x^3, y, z^2]}{A}=[x^3, y]$. More formally, $\proj{m}{A}(a) = m(a)$  if $a \in A$, and $\proj{m}{A}(a) = 0$ otherwise.
In contrast to sets and multisets, where the order of elements is irrelevant, in sequences the elements are given in a certain order, e.g., $\sequence{x, y, x, y} \neq \sequence{x, x, y, y}$.
We refer to the $i$-th element of a sequence $\sigma$ by $\sigma(i)$.
The size of a set, multiset or sequence $X$, that is \sizeof{X}, is defined to be the number of elements in~$X$.
As short notation for the real numbers between $0$ and $1$ we introduce $\mathbb{R}_0^1 = \{r \in \mathbb{R} \mid 0 \leq r \leq 1\}$.

In the following, we give a standard definition for activities, traces, and logs, except that we require each trace to begin with a designated start activity ($\sactivity$) and end with a designated end activity ($\eactivity$).
Note that it is reasonable to assume a process to have a clear start and end of execution, with each trace describing an end-to-end example run of this process, and that any log can easily be transformed accordingly by adding artificial activities.
\begin{definition}[Activity, Trace, Event Log]\label{def:atl}
Let ${\mathcal{A}}$ be the universe of all possible \emph{activities} (e.g., actions or operations),  let $\sactivity \in \universeofactivities$ be a designated start activity and let $\eactivity \in \universeofactivities$ with $\eactivity \neq \sactivity$ be a designated end activity. A \emph{trace} is a sequence $\trace = \Trace{a_1, a_2, \dots , a_n}$ with $n \geq 2$ such that $a_1=\sactivity, a_n=\eactivity$ and $a_i \in \mathcal{A}\backslash \{\sactivity, \eactivity\}$ for $i \in \{2,3, \dots , n-1\}$. Let \universeoftraces be the universe of all such traces. An \emph{event log} $\log \in \multisetset{\universeoftraces}$ is a multiset of traces.
\end{definition}
In this paper we extend the eST-Miner algorithm, which, given an event log, returns a Petri net without silent or duplicate transition labels, i.e., a Petri net where each transition can be uniquely identified by its activity label. Therefore, in this paper we can refer to transitions as activities.  Furthermore, there is a predefined source place connected to \sactivity marking the start of the process and a predefined sink place connected to \eactivity marking the end of the process. Besides the two predefined places, we only allow for places connecting transitions that are initially unmarked. These places are uniquely identified by their non-empty sets of input activities $I$ and output activities~$O$ (we are not interested in duplicated places). Considering this context, we introduce the following simplified definition for Petri nets.
\begin{definition}[Places and Petri nets]\label{def:PTN}
A \emph{Petri net} is a pair $\net = \tuple{\setofactivities, \setofplaces}$, where $\setofactivities \subseteq \universeofactivities$ is the finite set of activities including start and end, i.e., $\{\sactivity,\eactivity\} \subseteq \setofactivities$ and $\setofplaces' \subseteq \{\inoutpair{I}{O} \in \powerset{A} \times \powerset{A} \mid I \subseteq \setofactivities \wedge I \neq \emptyset \wedge \allowbreak  O \subseteq \setofactivities \wedge O \neq \emptyset \}$ is the set of intermediate \emph{places}. We call $I$ the set of \emph{ingoing} activities of a place and $O$ the  set of \emph{outgoing} activities. The complete set of places of $N$ is $P=P' \cup \{\inoutpair{\emptyset}{\{\sactivity\}},  \inoutpair{\{\eactivity\}}{\emptyset}\}$.
\end{definition}
Note that if $p=\inoutpair{I}{O}$, then $\bullet p=I$ and $p \bullet = O$ using standard notation. To reduce notional overload, we omit set brackets in places, i.e., we write $\inoutpair{\sactivity}{\eactivity}$ instead of $\inoutpair{\{\sactivity\}}{\{\eactivity\}}$. Figure~\ref{fig:examplePN} shows an example Petri net with places represented by circles and transitions represented by squares.
\begin{figure}[hbt]
    \centering
 \scalebox{0.8}{\includegraphics[width=0.45\linewidth, trim={1.2cm 3.5cm 4cm 4cm},clip]{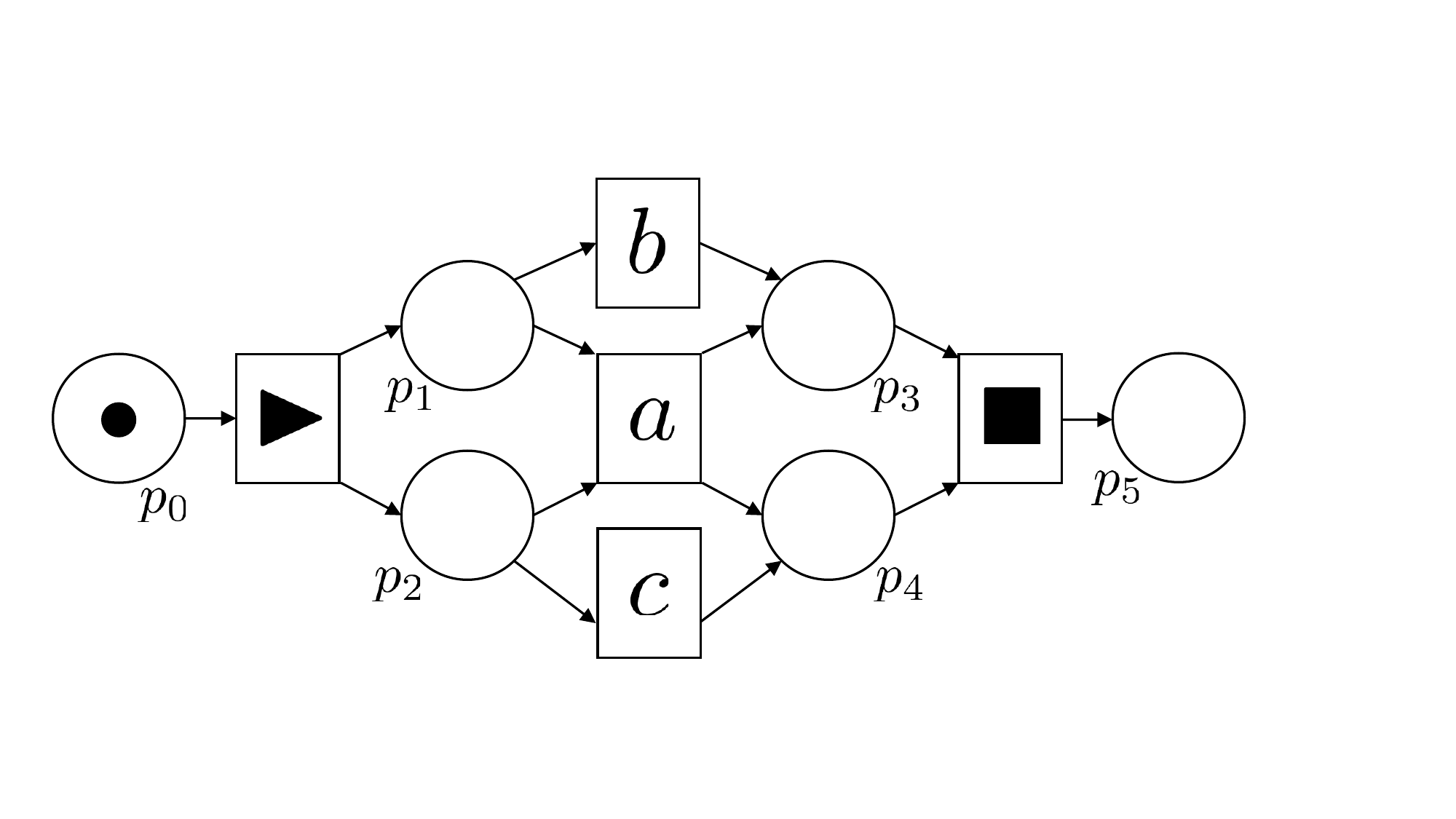}}
    \caption{Petri net with $A = \{\sactivity, a,b,c,\eactivity \}$ and $P=\{\inoutpair{\sactivity}{a,b}, \inoutpair{\sactivity}{a,c}, \inoutpair{a,b}{\eactivity}, \inoutpair{a,c}{\eactivity}, \inoutpair{\emptyset}{\sactivity}, \inoutpair{\eactivity}{\emptyset}\}$.}
    \label{fig:examplePN}
\end{figure}

The state of a Petri net $N=(A,P)$ is defined by its \emph{marking}. A marking  $M$ assigns tokens to places, i.e., $M\colon P \rightarrow \mathbb{N}_0$, and can therefore be defined as a multiset of places. The Petri net in Figure~\ref{fig:examplePN} is in marking $[p_0]$ with the corresponding token being represented as a  black dot.
A transition $a$ is \emph{enabled} in marking $M$ if and only if all places connected to it by an ingoing arc (i.e., all places in the preset of $a$, $\{\inoutpair{I}{O} \subseteq P \mid a \in O\}$) hold at least one token. An enabled transition $a$ can \emph{fire}, thereby  changing the marking $M$ of the Petri net to a new marking $M'$ by removing one token from every ingoing place and producing a token in every outgoing place. We write $M \fire{a} M'$.
In the example Petri net only the transition $\sactivity$ can fire initially, consuming the token from its ingoing place $p_0$ and producing tokens in $p_1$ and $p_2$. This enables transitions $a,b$ and $c$.

\begin{definition}[Replayable Traces]\label{def:replay}
Consider a Petri net $N=(A,P)$ and a trace $\trace = \Trace{\sactivity, a_1, \dots, a_n, \eactivity}$. We say that $N$ can \emph{replay} the trace \trace if and only if the transitions corresponding to the activities can be consecutively fired from the  marking $M_I = [\inoutpair{\emptyset}{\sactivity}]$ to reach the marking $M_F = [\inoutpair{\eactivity}{\emptyset}]$, i.e., there exists a sequence of markings $\Trace{M_1, \dots, M_n, M_{n+1}}$ such that $M_I \fire{\sactivity} M_1  \fire{a_1} M_2 \dots M_n \fire{a_n} M_{n+1} \fire{\eactivity} M_F$.
\end{definition}
Consider the Petri net in Figure~\ref{fig:examplePN}. This Petri net can replay the trace $\Trace{\sactivity, a, \eactivity}$ as follows:  $[p_0] \fire{\sactivity} [p_1, p_2] \fire{a} [p_3, p_4] \fire{\eactivity}[p_5]$. In contrast, the trace $\Trace{\sactivity, a, c, \eactivity}$ is not replayable: after firing  $[p_0] \fire{\sactivity} [p_1, p_2] \fire{a} [p_3, p_4]$ the next transition to be fired would be $c$ but $c$ is not enabled in $[p_3, p_4]$.

\begin{definition}[Behavior of a Petri net]\label{def:behavior}
We define the \emph{behavior} of the Petri net $N=\tuple{A, \setofplaces}$ to be the set of traces replayable by $N$, i.e., $\textit{behavior}(N) = \{\sigma \in \universeoftraces \mid N \textit{ can replay } \sigma\}$.
\end{definition}
Note that Definition~\ref{def:behavior} only allows for traces in the behavior of the form $\langle \sactivity,a_1,a_2, \dots, a_n,\eactivity\rangle$ (compare Definition~\ref{def:atl}) such that all intermediate places are unmarked at the end of replaying a trace and never have a negative number of tokens. The behavior of the Petri net in Figure~\ref{fig:examplePN} is $\{\Trace{\sactivity, a, \eactivity},$ \Trace{\sactivity, b,c,\eactivity}, $\Trace{\sactivity, c,b,\eactivity}\}$.

\section{Introducing the eST-Miner} \label{sec:fm:eSToverview}
\noindent
Several variants and extensions of the eST-Miner have been proposed in the past years. In the following, we introduce the eST-Miner variant used as the basis of this work. For further details, we refer the interested reader to the respective papers.

\subsubsection*{Method Overview}
As input, the algorithm takes a log \log over the set of activities $A \subseteq \universeofactivities$ and a noise threshold $\tau \in \mathbb{R}_0^1$, and returns a Petri net as output. Inspired by language-based regions, the basic strategy of the approach is to begin with a Petri net $N=(A,\{\inoutpair{\emptyset}{\sactivity}, \inoutpair{\eactivity}{\emptyset}\})$ whose transitions correspond exactly to the activities in \log and no intermediate places.
From the finite set of unmarked, intermediate candidate places, $\{\inoutpair{I}{O} \in \powerset{A} \times \powerset{A} \mid I \subseteq \setofactivities \wedge I \neq \emptyset \wedge \allowbreak  O \subseteq \setofactivities \wedge O \neq \emptyset \}$, the subset of all places \emph{fitting} with respect to $L$ and $\tau$ is computed and inserted, where the noise threshold $\tau$ allows us to ignore deviations from the main control-flow relations. A candidate place is considered fitting, if it does not hinder the replay of a significant part of \log as indicated by $\tau$. Details on the exact definition and computation of the set of fitting places are provided below.

To facilitate further computations and human readability,  \emph{implicit} places are identified and removed~\cite{impl1, impl2, colom} from the set of fitting places. A place is implicit if its removal does not  increase the number of traces replayable by the Petri net. We have two approaches at our disposal. The first option is to solve optimization problems to identify implicit places based on the structure of the Petri net, as proposed for the original eST-Miner~\cite{PN2019}. This approach reliably removes implicit places, but is computationally expensive. The second approach replays the event log to compare the markings of the places and uses this information to compute non-minimal regions, corresponding to implicit places \cite{ATAED2020}. While it is much faster, for correct and complete implicit place identification certain requirements must be satisfied, e.g.  sufficient similarity between the event log and behavior of the Petri net as well as guaranteed inclusion of certain places in the Petri net. Adaptation of this approach to the context of this work is out of scope, therefore, we use the first implicit place removal variant.

We continue with defining the exact meaning of fitting places.

\subsubsection*{Fitting Places}
The eST-Miner considers all possible candidate places based on the activities in the event log. Therefore, in the following, we consider only places whose sets of ingoing and outgoing activities occur in the corresponding event log at least once.

A candidate place can be either \emph{fitting} or \emph{unfitting} with respect to \log and $\tau$. We further distinguish \emph{unfitting} places into \emph{underfed} and \emph{overfed} places. These concepts can be used to compute the set of fitting places more efficiently. We introduce them on the level of individual traces first.

\begin{definition}[Fitting, Underfed and Overfed Places, cf.~\cite{monotonicity}] \label{def:fitting}
Let $\place = \inoutpair{I}{O} \in \powerset{\universeofactivities} \times \powerset{\universeofactivities}$  be a place, and let $\trace$ be a trace. With respect to \trace, \place is called
\begin{itemize}
    \item \emph{underfed}, denoted by  \underfed{\trace}{\place},  if and only if
$\exists k \in \set{1, 2, \dots , \sizeof{\trace}} \textit{ such that }\\
\sizeof{\set{i \; | \; i \in \set{1, 2, \dots, k-1} \wedge \trace(i) \in I}} < \sizeof{\set{i \; | \; i \in \set{1, 2, \dots, k} \wedge \trace(i) \in O}}$
    \item \emph{overfed}, denoted by  \overfed{\trace}{\place},  if and only if \\
$\sizeof{\set{i \; | \; i \in \set{1, 2, \dots, \sizeof{\trace}} \wedge \trace(i) \in I}} > \sizeof{\set{i \; | \; i \in \set{1, 2, \dots, \sizeof{\trace}} \wedge\trace(i) \in O}},$
	\item \emph{fitting}, denoted by  \fitting{\trace}{\place},  if and only if not \underfed{\trace}{\place} or \overfed{\trace}{\place}.
\end{itemize}
\end{definition}
A place is {underfed} with respect to a trace \trace if, at some point, replaying \trace requires more outgoing transitions of the place to be fired than ingoing transitions have been fired before, implying a lack of tokens in the place. Such a place hinders replay of $\sigma$ because a transition that should be fired is not enabled. Consider the Petri net in Figure~\ref{fig:examplePN} and the trace $\trace = \Trace{\sactivity, a, b, \eactivity}$. The place $p_1$ is {underfed} with respect to this trace, because when its outgoing transition $b$ occurs in the trace, not enough ingoing transitions have been fired before to provide the necessary token. Thus, we can conclude that in any Petri net that contains $p_1$, $b$ is not enabled in the marking reached after firing \sactivity and $a$.

\medskip
A place whose ingoing transitions occur more often than its outgoing transitions in a trace \trace is called {overfed} with respect to \trace. It hinders replay because it has tokens remaining if the end of replay is reached, which violates the requirement of ending in marking $[\inoutpair{\eactivity}{\emptyset}]$.  The place $p_3$ is {overfed} with respect to \trace: \sactivity is not connected, $a$ and $b$ each produce a token here, and firing \eactivity at the end of the trace consumes only one of the two tokens, thus one token is remaining.

Note that a place can be underfed and overfed with respect to a trace at the same time. If a place is neither overfed or underfed with respect to a trace, it is fitting. The place $p_2$ is {fitting} with respect to \trace. Firing \sactivity produces a token, which is subsequently consumed when $a$ is fired. Neither $b$ or \eactivity are connected to $p_2$. Thus, at the end of $\sigma$, no token is missing or remaining in $p_2$.

\medskip
We can formalize these observations and relate them to the behavior of a Petri net as introduced in Section~\ref{sec:fm:prelims} as follows:

\begin{observation}[Place Fitness and Firing Sequences]\label{cor:fitnessVSfiring}
Given a Petri net $N=(A, P)$ and a trace $\sigma = \Trace{a_1, a_2, ..., a_n} \in \universeoftraces$ that is replayable by $N$, consider a place $p=\inoutpair{I}{O}$ with $p \notin P$.
The following properties hold for $N'=(A, P \cup \{p\})$:
\begin{itemize}
\itemsep=0.95pt
\item $\fitting{\trace}{\place} \Leftrightarrow \sigma \in \textit{behavior}(N')$
\item $\begin{aligned}[t]
\underfed{\trace}{\place} \Leftrightarrow
&\; \exists k, 0 \leq k \leq n, \textit{ such that } [\inoutpair{\emptyset}{\sactivity}] \fire{a_1} M_1 \fire{a_2} ... \fire{a_k} M_k  \wedge  p \notin M_k \wedge a_{k+1} \in O \\
\Rightarrow & \; \sigma \notin \textit{behavior}(N')
\end{aligned}$
\item $\begin{aligned}[t]
\overfed{\trace}{\place} \Leftrightarrow
& \; [\inoutpair{\emptyset}{\sactivity}]  \fire{a_1} M_1 \fire{a_2} ... \fire{a_n} M_n
 \wedge p \in M_n
\Rightarrow   \sigma \notin \textit{behavior}(N')
\end{aligned}$
\end{itemize}
\end{observation}

We define the following functions to succinctly refer to multisets of traces of the event log with respect to which (sets of) places are fitting, underfed or overfed.
\begin{definition}[Multisets of Fitting/Underfed/Overfed Traces (compare \cite{monotonicity})]
Given a set of places $\setofplaces \subseteq \powerset{\universeofactivities} \times \powerset{\universeofactivities}$ and a an event log $L \in \multisetset{\universeoftraces}$, we define the following functions with respect to a place $p \in P$:
\begin{itemize}
\itemsep=0.95pt
    \item $\pFit_L(p) = \proj{L}{\{\trace \in L \mid \fitting{\trace}{\place}\}}$ is the multiset of  log traces for which $p$ is fitting 
    \item $\pNeg_L(p) = \proj{L}{\{\trace \in L \mid \underfedsubscript{\trace}{\place}\}}$ is the multiset of  log traces for which $p$ is underfed 
    \item $\pPos_L(p) = \proj{L}{\{\trace \in L \mid \overfed{\trace}{\place}\}}$ is the multiset of  log traces for which $p$ is overfed 
\end{itemize}
We extend these functions to the set of places \setofplaces as follows:
\begin{itemize}
\itemsep=0.95pt
    \item $\pFit_L(\setofplaces) = \proj{L}{\{\sigma \in L \mid \forall \place \in \setofplaces \colon \fitting{\trace}{\place}\}}.$
    \item $\pNeg_L(\setofplaces) = \proj{L}{\{\sigma \in L \mid \exists \place \in \setofplaces \colon \! \underfedsubscript{\trace}{\place}\}}.$
    \item $\pPos_L(\setofplaces) = \proj{L}{\{\sigma \in L \mid \exists \place \in \setofplaces \colon \overfed{\trace}{\place}\}}.$
\end{itemize}
\end{definition}
Note that, for a Petri net $N=(A,P)$ and event log \log,  $ \pFit_L(\setofplaces)$ corresponds exactly to the log traces replayable by $N$, since none of the places in $N$ hinders the replay of those traces.

A fitness metric is used to assign a value to a place, indicating how well it fits the given event log. We require its range to be $\mathbb{R}_0^1$ to facilitate comparison with thresholds and between different metrics.
\begin{definition}[Measuring Fitness of a Place]
Given an event log $L \in \multisetset{\universeoftraces}$ we define a \emph{fitness metric} to be a function $\fitnessmetric^L \colon \powerset{\universeofactivities} \times \powerset{\universeofactivities} \rightarrow \mathbb{R}_0^1$, which assigns a fitness score to a place based on $L$ such that $0$ represents the lowest achievable fitness, while~$1$ reflects the highest achievable fitness.
\end{definition}
Before introducing concrete fitness metrics, we define the multiset of traces in an event log which contain a given set of activities.
\begin{definition}[Activated Traces, cf. \cite{monotonicity}]
    Given an event log $L \in \multisetset{\universeoftraces}$ and a set of activities $A \subseteq \universeofactivities$, we define the multiset of traces \emph{activated} by $A$ as $\texttt{act}_L(A)= \proj{L}{\{\sigma \in L \mid \exists i \in \mathbb{N} \colon \sigma(i) \in A\}}$.
\end{definition}
Two fitness metrics have been proposed so far, in the following called \emph{absolute fitness} and \emph{relative fitness}.

\begin{definition}[Fitness Metrics, cf. \cite{monotonicity}]\label{def:fit}
Let $L \in \multisetset{\universeoftraces}$ be an event log and let $p=(I \mid O)$ be a place. We define two different fitness metrics of the place $p$ with respect to $L$.
\begin{align*}
    \textit{absolute fitness: } &\fitnessmetric^L_{\textit{abs}}(p)=
            \frac{|\pFit_L(p)|}{|L|}\\
    \textit{relative  fitness: } &\fitnessmetric^L_{\textit{rel}}(p)=
            \frac{|\texttt{act}_L(I \cup O) \bignplus \pFit_L(p)|}{|\texttt{act}_L(I \cup O)|}
\end{align*}
Both metrics return values between $0$ (low fitness) and $1$ (high fitness).
\end{definition}
Based on a fitness metric and the noise threshold $\tau$ we can lift the notions of fitting and unfitting to the complete event log rather than a single trace.
\begin{definition}[Fitness with Respect to a Threshold, compare \cite{monotonicity}]\label{def:fitnessthreshold}
With respect to an event log $\log \in \multisetset{\universeoftraces}$, a threshold $\tau \in \mathbb{R}_0^1$ and a fitness metric \fitnessmetric, we call a place $p$
\begin{itemize}
\itemsep=0.95pt
	\item \emph{fitting} 
	if and only if $\fitnessmetric^{\log}(p) \geq \tau$, and
 	\item \emph{unfitting} 
 	if and only if $\fitnessmetric^{\log}(p) < \tau.$
\end{itemize}
\end{definition}
To increase the efficiency of place candidate evaluation, we need to lift the concepts of underfed and overfed to the log level in a similar way. However, this is done for concrete fitness metrics individually to satisfy certain requirements as explained in the following subsection.

\subsubsection*{Efficient candidate places evaluation}
The eST-Miner uses token-based replay to evaluate each candidate place. The number of possible places is finite but exponential in the number of observed activities.  To avoid replaying the log on the exponential number of candidate places and increase efficiency, the eST-Miner aims to skip uninteresting subsets of candidate places, while still guaranteeing to discover all places that are considered {fitting}. To this end, it utilizes the following \emph{monotonicity properties} to derive information about candidate places based on a place $p$ already evaluated: If $p$ is underfed with respect to a trace, a candidate place generated only by adding outgoing activities to or removing ingoing activities from $p$ will lack at least as many tokens during replay and will therefore also be underfed for that trace. Vice versa, if $p$ is overfed, a candidate place constructed by adding only ingoing activities or removing outgoing activities will have at least the same number of tokens remaining  and is therefore overfed as well. Clearly, this information can be derived without actually evaluating the newly generated places.

\begin{theorem}[Monotonicity Properties, cf. \cite{monotonicity}]\label{theo:monotonicity}
    Consider two places $p_1 =(I_1 \mid O_1)$ and $p_2 =(I_2 \mid O_2)$. Then  the following holds for any event log $L  \in \multisetset{\universeoftraces}$:
    \begin{itemize}
        \item $I_1 \supseteq I_2 \wedge O_1 \subseteq O_2 \implies \pNeg_L(p_1) \subseteq \pNeg_L(p_2)$
        \item $I_1 \subseteq I_2 \wedge O_1 \supseteq O_2 \implies \pPos_L(p_1) \subseteq \pPos_L(p_2)$
    \end{itemize}
\end{theorem}
The eST-Miner organizes the candidate places as a forest, called the \emph{complete candidate tree}. In this forest, we distinguish between two types of edges: blue edges connect a parent candidate place to a child place constructed from it by adding an outgoing activity, while red edges connect it to a child generated by adding an ingoing activity. A visualization of such a tree-structured candidate space for activities $\sactivity, a,b$ and $\eactivity$ is given in Figure~\ref{examplefulltree}. If combined with suitable fitness metrics, the monotonicity properties can be used to skip subtrees in the complete candidate tree.

\begin{figure}[!h]
\vspace*{1mm}
\centering
\includegraphics[width =1\linewidth]{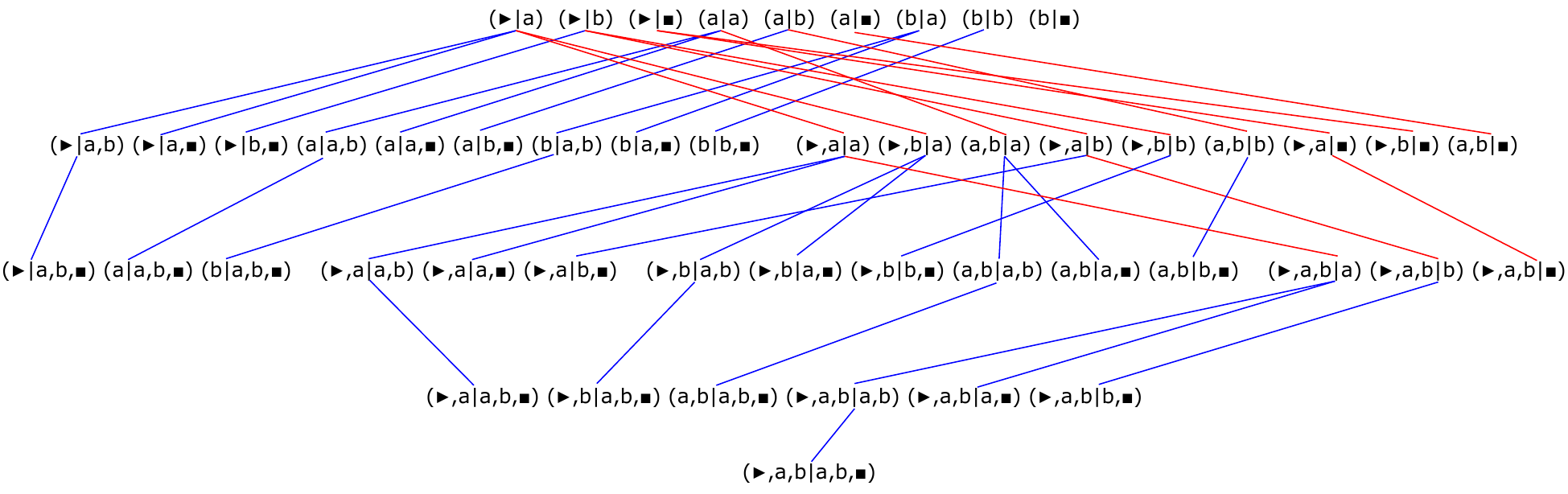}\vspace*{-1mm}
\caption{{Example of a tree-structured candidate space for the set of activities \set{\sactivity, a, b, \eactivity}, with orderings ${\eactivity >_i b >_i a >_i \sactivity}$ and ${\eactivity >_o b >_o a >_o \sactivity}$.}}
\label{examplefulltree}\vspace*{-2mm}
\end{figure}

\begin{definition}[Complete Candidate Tree] \label{def:CT} Let $A \in \universeofactivities$ be a set of activities and let $>_i, >_o$ be two total orderings on this set of activities.
A \emph{complete candidate tree}  is a pair $CT=\tuple{V, F}$ with
$$V = \set{\inoutpair{I}{O} \in \powerset{A} \times \powerset{A} \; | \; I \subseteq \setofactivities\without{\set{\eactivity}} \wedge O \subseteq \setofactivities\without{\set{\sactivity}} \wedge I \neq \emptyset \wedge O \neq \emptyset}.$$
We have that $F= F_\text{red} \cup F_\text{blue}$, with
	\begin{align*}
	F_\text{red} = &\{\tuple{\inoutpair{I_1}{O_1}, \inoutpair{I_2}{O_2}} \in V \times V
		\mid \sizeof{O_2}=1 \wedge O_1 = O_2\\
		&\wedge \exists a \in I_1 \colon \big( I_2 \cup \set{a} =I_1 \wedge \forall a' \in I_2 \colon a >_i a' \big)
	\} \textbf{~(\textcolor{red}{red edges})}\\
	F_\text{blue} = &\{\tuple{\inoutpair{I_1}{O_1}, \inoutpair{I_2}{O_2}} \in V \times V
		\mid I_1=I_2\\
		&\wedge \exists a \in O_1 \colon \big( O_2 \cup \set{a}= O_1 \wedge \forall a' \in O_2 \colon a >_o a' \big)
	\} \textbf{~(\textcolor{blue}{blue edges})}.
	\end{align*}
If $\tuple{\inoutpair{I_1}{O_1}, \inoutpair{I_2}{O_2}} \in F$, we call the candidate \inoutpair{I_1}{O_1} the \emph{child} of its \emph{parent} \inoutpair{I_2}{O_2}.
\end{definition}

\noindent Note the incremental structure of the trees, i.e., the increase in distance from the roots corresponds to the increase of the number of activities connected to the candidate places. Specifically, for a place \inoutpair{I}{O} at depth $k$ we have that $|I|+|O|=k$. Each level of the tree contains all possible places that can be created by connecting the corresponding number of activities (excluding $\sactivity$ as an outgoing activity and $\eactivity$ as an ingoing activity, since such places cannot fit), with the root level consisting of all candidate places that can be build using exactly two transitions.
The asymmetry in the definition of red and blue edges (Definition~\ref{def:CT}) ensures that all candidates are reachable from exactly one root by a unique path (using the same rules for both types of edges would result either in some candidates having both, a blue and a red parent, or no parent at all). Thus, candidate traversal following the tree structure guarantees that all candidates are considered exactly once.

\medskip
Organization of candidates within the same depth and their connections to parent and child candidates is not fixed, but defined by ordering strategies for of ingoing activities~($>_i$) and outgoing activities~($>_o$). With knowledge of these ordering strategies we can deterministically compute the next candidate place to be considered based on the last candidate we evaluated, i.e., based only on the connected activities of the last place. This is very space efficient, since we do not need to keep the whole tree in memory. It also contributes to time efficiency: we do not only avoid evaluation of candidates in skipped subtrees but do not even need to compute and traverse them.

Consider any candidate place $p=\inoutpair{I}{O}$ in the complete candidate tree (Definition~\ref{def:CT}). By construction, the tree structure guarantees that for every descendant $p_\textit{blue}=\inoutpair{I_\textit{blue}}{O_\textit{blue}}$ of $p$ reachable via a path of purely blue edges we have that $I = I_\textit{blue}$ and $O \subseteq O_\textit{blue}$ and thus, by Theorem~\ref{theo:monotonicity}, $\pNeg_L(p) \subseteq \pNeg_L(p_\textit{blue})$ holds. Respectively, for every descendant $p_\textit{red}$ of $p$ reachable via purely red edges, we have that $\pPos_L(p) \subseteq \pPos_L(p_\textit{red})$.
Combining this property with a suitable fitness metric allows us to cut off subtrees consisting of only unfitting candidate places based on the replay result of the parent candidate place.

In Definition~\ref{def:fitnessthreshold} we have used the noise threshold $\tau$ to lift the concepts of fitting and unfitting places to the log level. To enable skipping of uninteresting subtrees, we have to lift the notions of underfed and overfed places to the log level in a way that is suitable for the applied fitness metric.

\begin{samepage}
\begin{definition}[Underfed/Overfed Places with Respect to an Event Log, cf. \cite{monotonicity}] \label{def:fednessFM}
With respect to an event log $\log \in \multisetset{\universeoftraces}$ and a noise threshold $\tau \in \mathbb{R}_0^1$ we define a place $p=\inoutpair{I}{O}$ to be \emph{underfed} or \emph{overfed} based on the fitness metric used.

\noindent
A place {unfitting} with respect to \emph{absolute fitness} can be classified further as
\begin{itemize}
\itemsep=0.95pt
	\item \emph{underfed}, denoted by $\underfed[{\log,\tau}]{\textit{abs}}{\place}$,  if and only if $\frac{|\pNeg_{\log}(p)|}{|\log|} > 1-\tau$,
	\item \emph{overfed}, denoted by $\overfed[{\log,\tau}]{\textit{abs}}{\place}$,  if and only if $\frac{|\pPos_{\log}(p)|}{|\log|} > 1-\tau$.
\end{itemize}
    A place {unfitting} with respect to \emph{relative fitness} can be classified further as
    \begin{itemize}
    \itemsep=0.95pt
	\item \emph{underfed}, denoted by $\underfed[{\log,\tau}]{\textit{rel}}{\place}$,  if and only if $\frac{|\texttt{act}_L(I \cup O) \bignplus \pNeg_L(p)|}{|\texttt{act}_L(I \cup O)|} > 1-\tau$,
	\item \emph{overfed}, denoted by $\overfed[{\log,\tau}]{\textit{rel}}{\place}$,  if and only if $\frac{|\texttt{act}_L(I \cup O) \bignplus \pPos_L(p)|}{|\texttt{act}_L(I \cup O)|} > 1-\tau$.
  \end{itemize}
\end{definition}
\end{samepage}
In Figure~\ref{fig:placestatusrelations} we provide an overview illustrating the fitness status a place can have based on Definitions~\ref{def:fitnessthreshold} and~\ref{def:fednessFM}.  A place can either be fitting or unfitting. A place that is unfitting may be further classifiable to be underfed or overfed. A place can be underfed and overfed at the same time, given that sufficiently many traces validate the necessary requirements. It can also be unfitting without being neither underfed nor overfed.  Considering the event log ${L=[\Trace{\sactivity, a,a,b,d, \eactivity}^{60}, \Trace{\sactivity, a,c,d,d, \eactivity}^{40}]}$ as an example, the following place fitness classifications hold for all fitness metrics discussed in this paper.  With respect to $L$ and a noise threshold of $\tau=0.5$ the place $p_1 = \inoutpair{\sactivity}{b}$ is fitting,  $p_2 = \inoutpair{c}{\eactivity}$ is underfed, $p_3 = \inoutpair{\sactivity}{c}$ is overfed, and $p_4 = \inoutpair{a, d}{a}$ is underfed as well as overfed. With respect to $L$ and a noise threshold of $\tau=0.3$ the place $p_5 = \inoutpair{a}{d}$ is unfitting but neither underfed nor overfed.

\begin{figure}[tbh]
\vspace{2mm}
\centering
\includegraphics[width =0.5\linewidth]{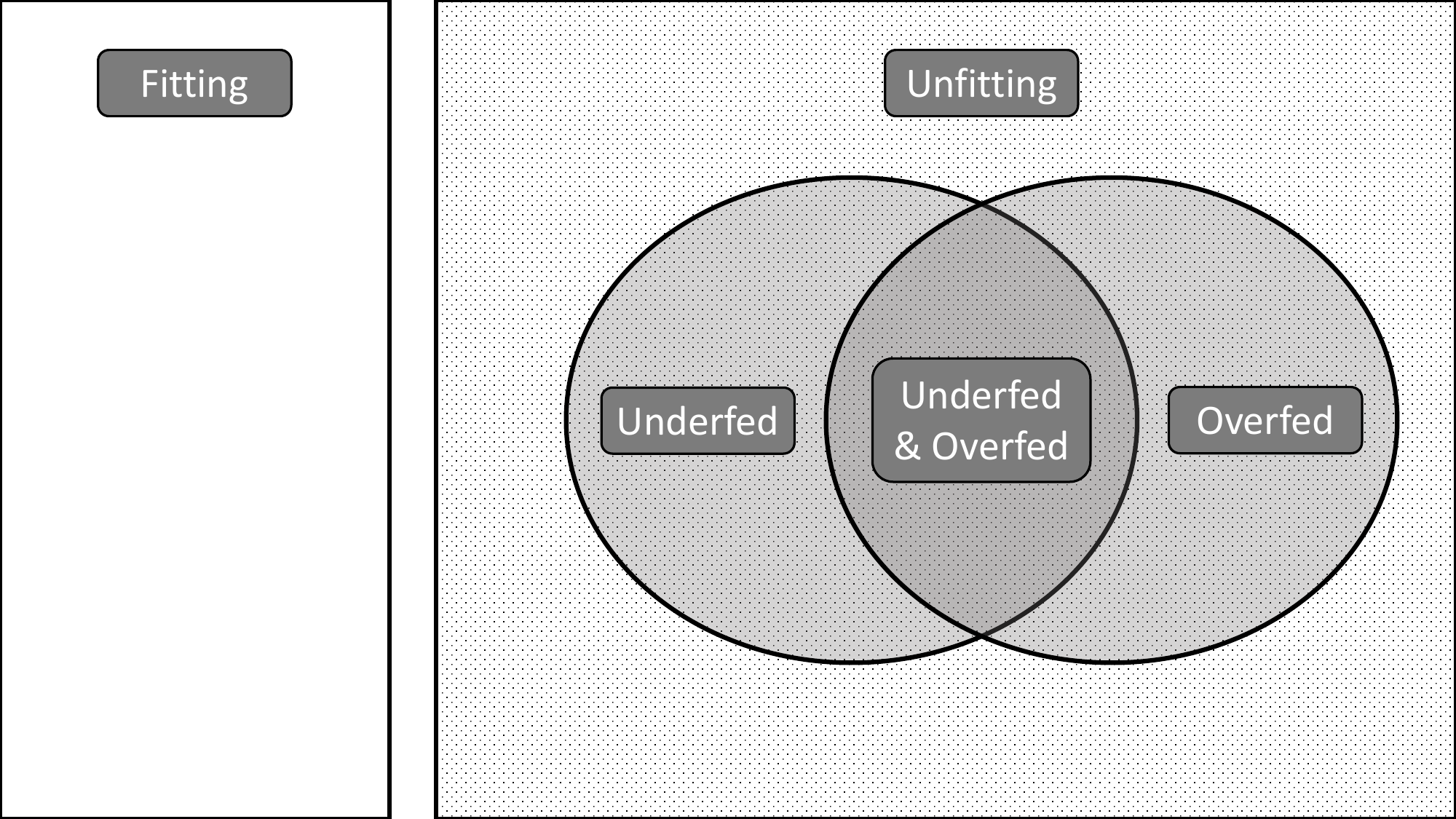}
\caption{Illustrating the fitness status a candidate place can have with respect to an event log and noise threshold. A place can either be fitting or unfitting. If it is unfitting, it may be underfed, overfed or both, enabling the skipping of derived candidate places. It may also be unfitting (Definition~\ref{def:fitnessthreshold}) without satisfying the threshold  to be underfed or overfed (Definition~\ref{def:fednessFM}), in which case no candidates can be skipped based on it.}
\label{fig:placestatusrelations}
\end{figure}

The prerequisite for the eST-Miner to skip parts of the candidate space is the guarantee, that if a place is determined to be underfed (overfed), the same must hold for all its blue (red) descendants in the complete candidate tree. Due to the construction of the tree and the following lemma this requirement is satisfied for Definition~\ref{def:fednessFM}.

\begin{samepage}
\begin{lemma}[Monotonicity of Fitness Metrics, cf.\cite{monotonicity}]\label{lemma:monotonicity}
Consider an event log $L  \in \multisetset{\universeoftraces}$ and a noise threshold $\tau \in \mathbb{R}_0^1$.

For two places $p=\inoutpair{I}{O}$ and $p'=\inoutpair{I}{O'}$ with $O \subseteq O'$ the following implications hold:
\begin{align*}
\underfed[{\log,\tau}]{\textit{abs}}{\place} &\implies \underfed[{\log,\tau}]{\textit{abs}}{\place'},\\
\underfed[{\log,\tau}]{\textit{rel}}{\place} &\implies \underfed[{\log,\tau}]{\textit{rel}}{\place'}.
\end{align*}
For two places $p=\inoutpair{I}{O}$ and $p''=\inoutpair{I''}{O}$ with $I \subseteq I''$ the following implications hold:
\begin{align*}
\overfed[{\log,\tau}]{\textit{abs}}{\place} &\implies \overfed[{\log,\tau}]{\textit{abs}}{\place''},\\
\overfed[{\log,\tau}]{\textit{rel}}{\place} &\implies \overfed[{\log,\tau}]{\textit{rel}}{\place''}.
\end{align*}
\end{lemma}
\end{samepage}

\subsubsection*{Conclusion}

\begin{figure}[!b]
\centering
\includegraphics[width =0.55\linewidth, trim={1.5cm 0cm 1.5cm 0.1cm},clip]{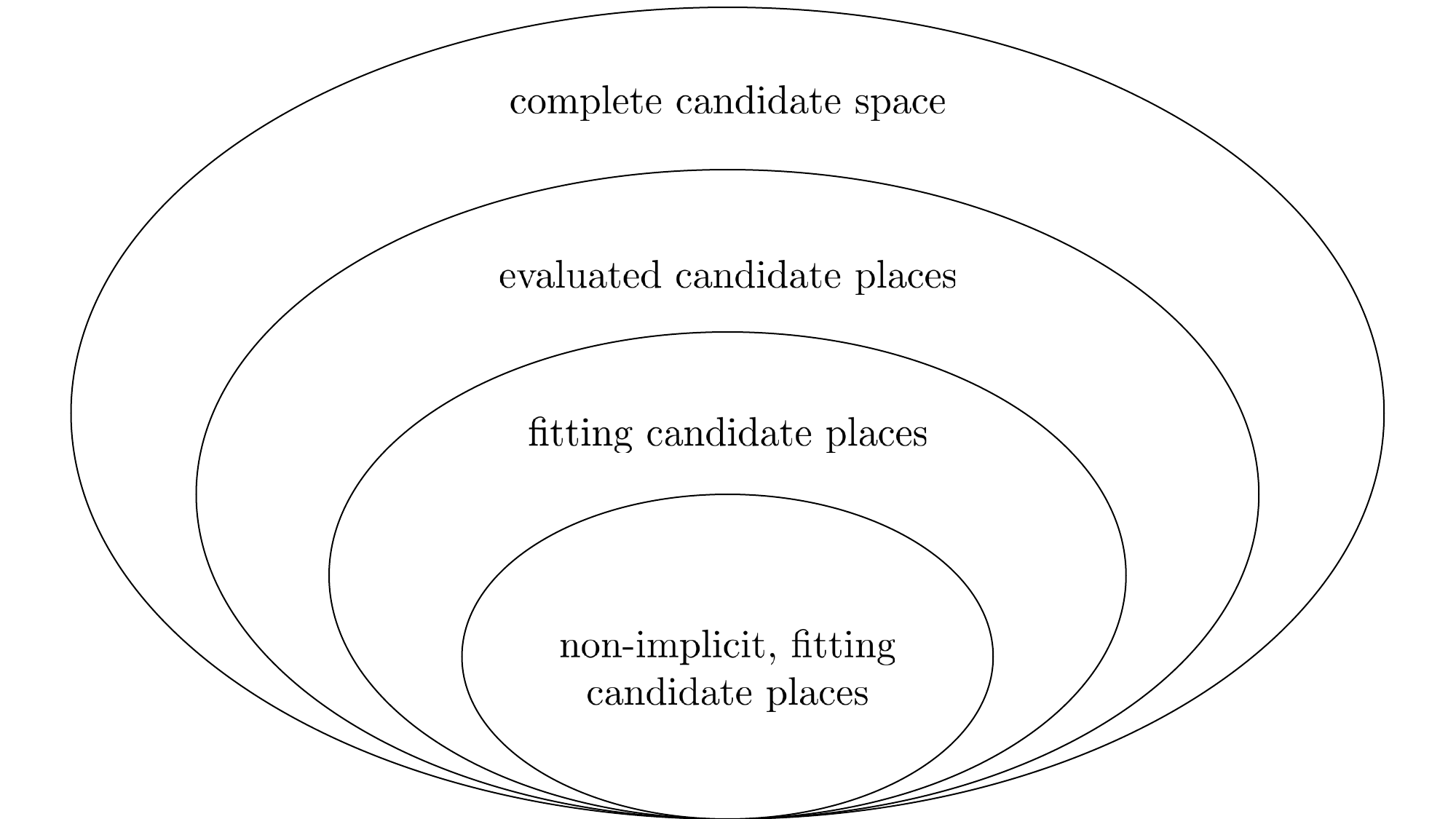}\\
\hfill\\
\scriptsize{
\centering
\begin{tabular}{@{}c|c|cc@{}}
\toprule
\textbf{Set of Places} & \textbf{Teleclaims} & \textbf{RTFM}\\ \midrule
\textbf{complete candidate space}         & 16,769,025  &    16,769,025         \\
\textbf{evaluated candidates}     & 5,119,151 (69 \% skipped)  &  3,066,180  (82 \% skipped) \\
\textbf{fitting places}           & 18,139     &     3,855             \\
\textbf{non-implicit places}      & 8        &        4              \\ \bottomrule
\end{tabular}
}
\caption{Comparison of the sets of places considered within the eST-Miner framework. To exemplify the impact of cutting off unfitting subtrees, we choose the well-known \texttt{Teleclaims} and \texttt{Road Traffic Fine Management} event logs (both have $11+|\{\sactivity, \eactivity\}|=13$  activities) with $\tau=1.0$ (i.e., requiring all places to be perfectly fitting). The table shows the size relations of the computed sets of places to the complete candidate space (for lexicographical activity orderings $>_i$ and $>_o$ in the complete candidate tree).}
\label{fig:placeSetRelations}
\end{figure}

The runtime of the eST-Miner strongly depends on the number of candidate places skipped during the search for fitting places. In Figure~\ref{fig:placeSetRelations}, the subset relations between the sets of places relevant to the eST-Miner framework are visualized. To illustrate the effect of cutting off candidate subtrees, the absolute number of places in these sets is shown for two example applications of the algorithm on two different event logs. While the size of the complete candidate space is the same for both event logs (based on the number of activities), the ordering of the candidates within the tree and properties of the event logs themselves result in different sizes for the subset of candidate places to be evaluated. This example emphasizes the performance gain achieved by skipping subtrees.

\section{Problem explanation and solution framework} \label{sec:fm:framework}
\noindent
One of the major advantages of the eST-Miner is its ability to ignore infrequent behavior patterns during discovery. This is achieved by computing for each place candidate individually the multiset of log traces it allows to replay. Thus, it can consider and combine information contained in all log traces even if not all of these traces are replayable in the constructed Petri net. This is illustrated by the example in  Figure~\ref{fig:examplefilter}: every place in the sequential net allows for a large fraction of the event log to be replayed and the discovered Petri net indeed  represents the main behavior contained in the  log.

\medskip
On the downside, its straightforward insertion of fitting places makes the eST-Miner prone  to introduce deadlocks into the discovered Petri net. Deadlocking constructs are clearly undesirable, since they may result in a model that does not meet the user's expectations with respect to fitness or at least is unnecessary complicated given the behavior it represents. We use the example event log and corresponding Petri net given in Figure~\ref{fig:smallexamplecomplete} to illustrate the issue of deadlocks, discuss its two causes and motivate our twofold solution approach.
\begin{figure}[tbh]
\vspace{2mm}
\centering
\includegraphics[width =0.52\linewidth, trim={0.25cm 0.2cm 0.4cm 4.75cm},clip]{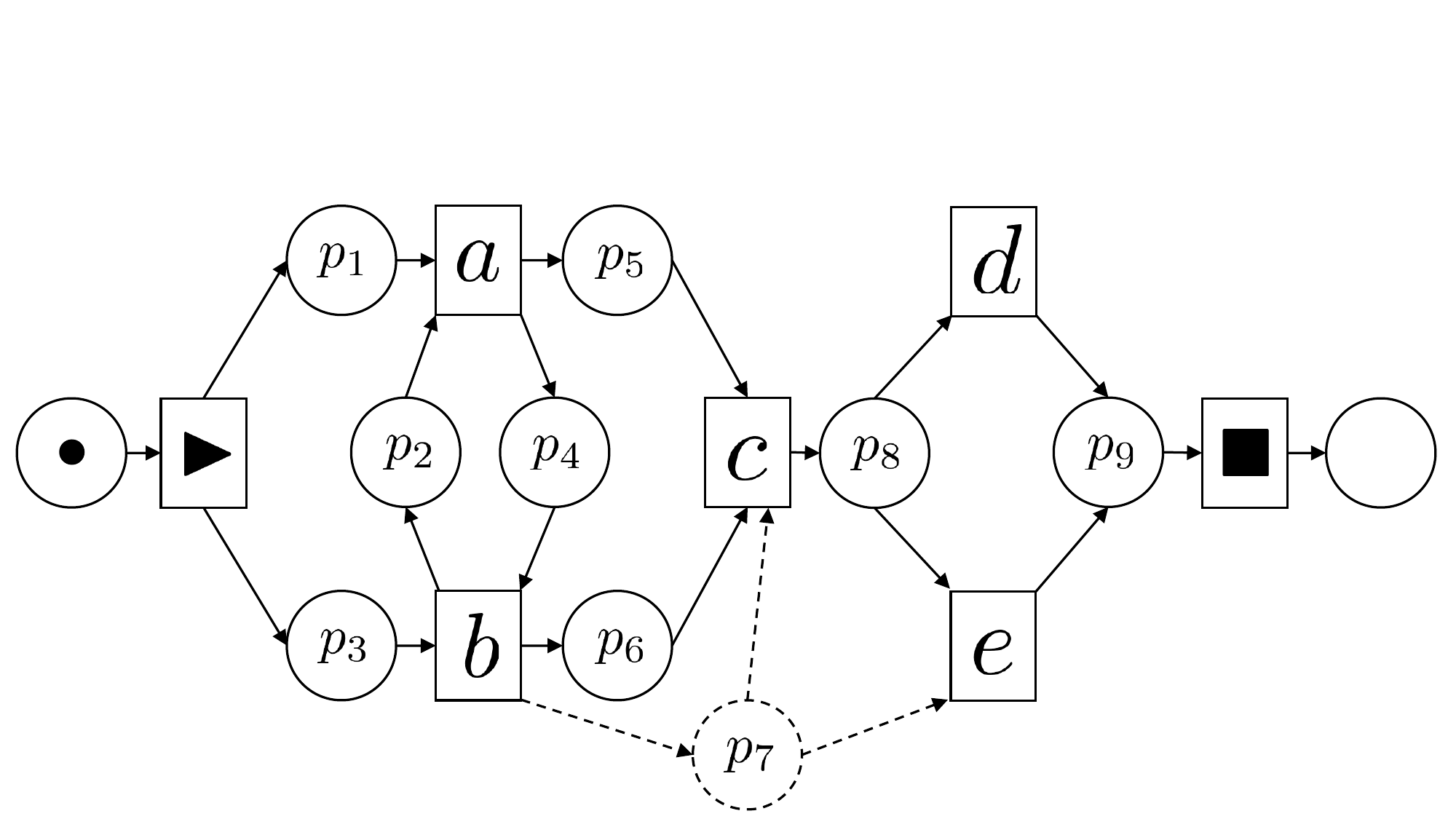}
\caption{To exemplify the problems addressed in this work we use the Petri net above together with the event log $[
\Trace{\sactivity, a,b,c,d,\eactivity}^{35},
\Trace{\sactivity, a,b,c,e,\eactivity}^{5},
\Trace{\sactivity, b,a,c,d,\eactivity}^{55},
\Trace{\sactivity, b,a,c,e,\eactivity}^{5}
]$. {One issue is the insertion of $p_7$ instead of the more desirable place $p_6$. Another problem is posed by the interaction between $p_2$ and~$p_4$.}}
\label{fig:smallexamplecomplete}\vspace*{-2mm}
\end{figure}

\subsubsection*{Deadlocks Related to Place Evaluation}
The first reason for deadlocks in the Petri net discovered by the eST-Miner is caused by how the proposed fitness metrics evaluate places that are connected to infrequent activities.  Both, the absolute and relative fitness metric, evaluate the fitness of a place based on a fraction of log traces it allows to replay. Places that hinder replay of only a few traces are likely to pass the threshold defined by $\tau$. Places that hinder replay of infrequent activities are hardly penalized by this strategy, since blocking of such activities does not significantly lower the fraction of replayable traces. Therefore, it is likely that returned models include places that block infrequent activities.

\medskip
Consider the example event log and the place $p_7 = \inoutpair{b}{c,e}$ in Figure~\ref{fig:smallexamplecomplete}. This place can replay all traces in the event log that do not include the infrequent activity $e$, i.e., a fraction of $0.9$ of all traces. Inserting this place into the Petri net would block $e$ from being executed completely. However, despite occurring only infrequently, activity $e$ has a clear positioning within the control-flow: whenever it occurs in a trace, it replaces execution of activity $d$ after $c$ and before \eactivity. The place $p_6$ can perfectly replay the event log (i.e., replays a fraction of $1.0$ of all traces) but is removed as implicit when~$p_7$ is added.

\medskip
In Section~\ref{sec:fm:fitness}, we aim to tackle this problem by introducing a new fitness metric to evaluate a place. This strategy maintains a high fitness score for places that adequately represent the control-flow patterns of all their connected activities (e.g., $p_6$) while places that block too many traces for any of their connected transitions are assigned a low fitness value (e.g., $p_7$). Additionally, the new metric is defined in such a way that the eST-Miner can still use the evaluation results to improve efficiency by skipping parts of the candidate space.

\subsubsection*{Deadlocks Related to Combination of Places}

The second reason for deadlocks in the Petri net discovered by the eST-Miner is caused by the combination of places that allow for replay of differing parts of the event log.
When a candidate place is evaluated to be fitting based on $L$ and $\tau$, it is simply inserted into the Petri net by connecting it to its uniquely labeled ingoing and outgoing transitions. The resulting Petri net can replay exactly the subset of log traces in the intersection of the traces replayable by all inserted places. Independently of the fitness metric used, this may result in Petri nets that can replay a fraction of traces significantly lower than $\tau$, despite each single place satisfying the threshold.

\medskip
We illustrate the issue using the example in Figure~\ref{fig:smallexamplecomplete}.
Consider the example event log and the candidate places $p_1$ to $p_6$. All of these places can replay at least a fraction of $0.4$ traces in the event log and thus would be inserted into the Petri net for any threshold $\tau \leq 0.4$. However, the resulting Petri net would not be able to replay any trace at all: the places $p2$ and $p3$  will block any execution of activities $a$ and $b$. Two viable solutions exist to obtain a more reasonable model. The first variant, visualized in Figure~\ref{fig:smallexamplecompleteSol1}, would be to include neither $p_2$ nor $p_4$ and allow a behavior where $a$ and $b$ are executed in parallel.

\begin{figure}[tbh]
\vspace*{2mm}
\centering
\includegraphics[width =0.52\linewidth, trim={0.25cm 0.2cm 0.4cm 4.75cm},clip]{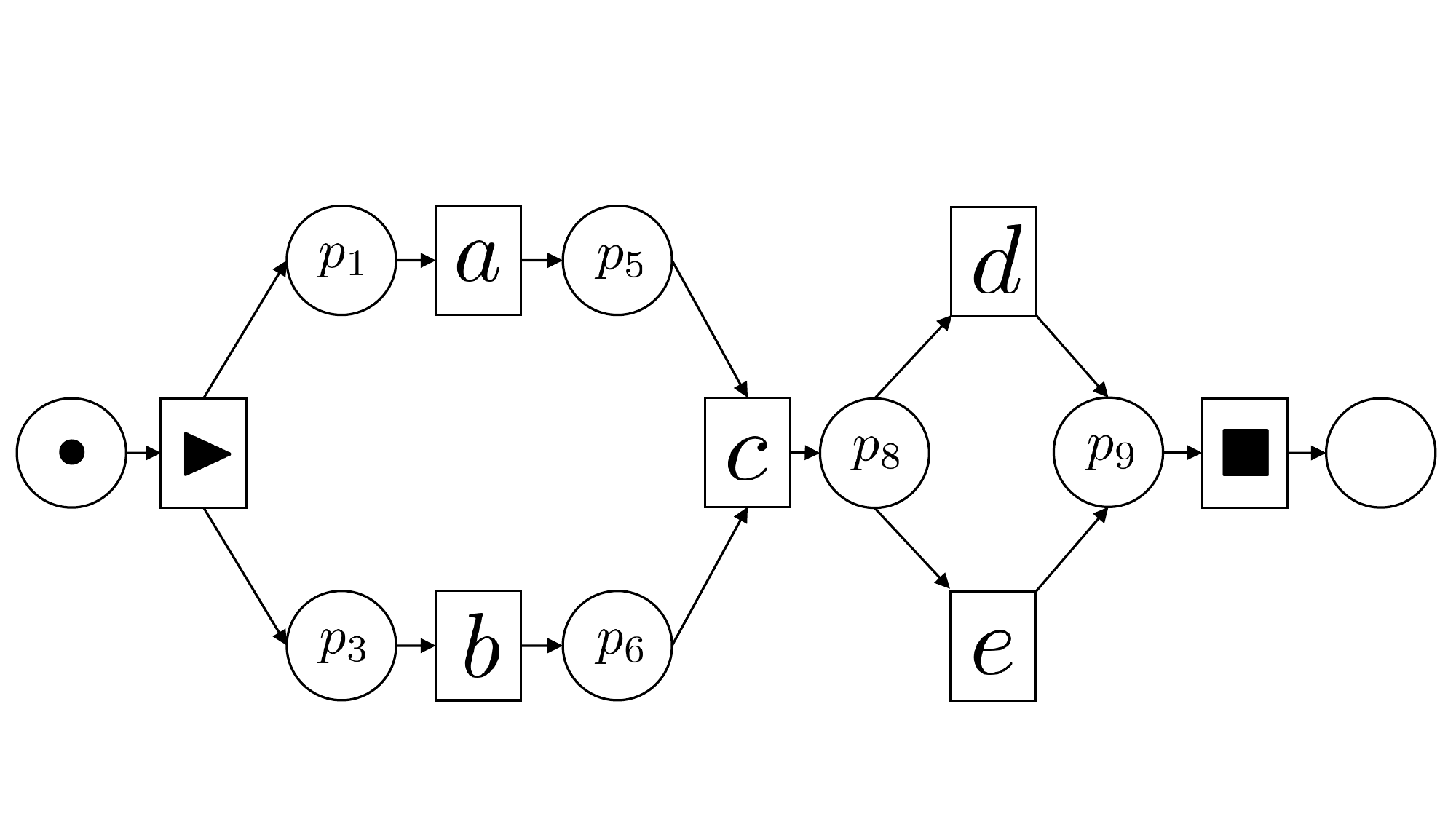}\vspace*{-5mm}
\caption{Consider the event log $[
\Trace{\sactivity, a,b,c,d,\eactivity}^{35},
\Trace{\sactivity, a,b,c,e,\eactivity}^{5},
\Trace{\sactivity, b,a,c,d,\eactivity}^{55},
\Trace{\sactivity, b,a,c,e,\eactivity}^{5}
]$. The Petri net above is able to perfectly replay all traces in the event log.}
\label{fig:smallexamplecompleteSol1}
\end{figure}
Alternatively, one can focus on the most frequent behavior and insert only $p_2$, which would result in the removal of the then implicit places $p_1$ and $p_6$. This second option is illustrated in Figure~\ref{fig:smallexamplecompleteSol2}.

\begin{figure}[tbh]
\centering
\includegraphics[width =0.52\linewidth, trim={0.25cm 0.2cm 0.4cm 4.75cm},clip]{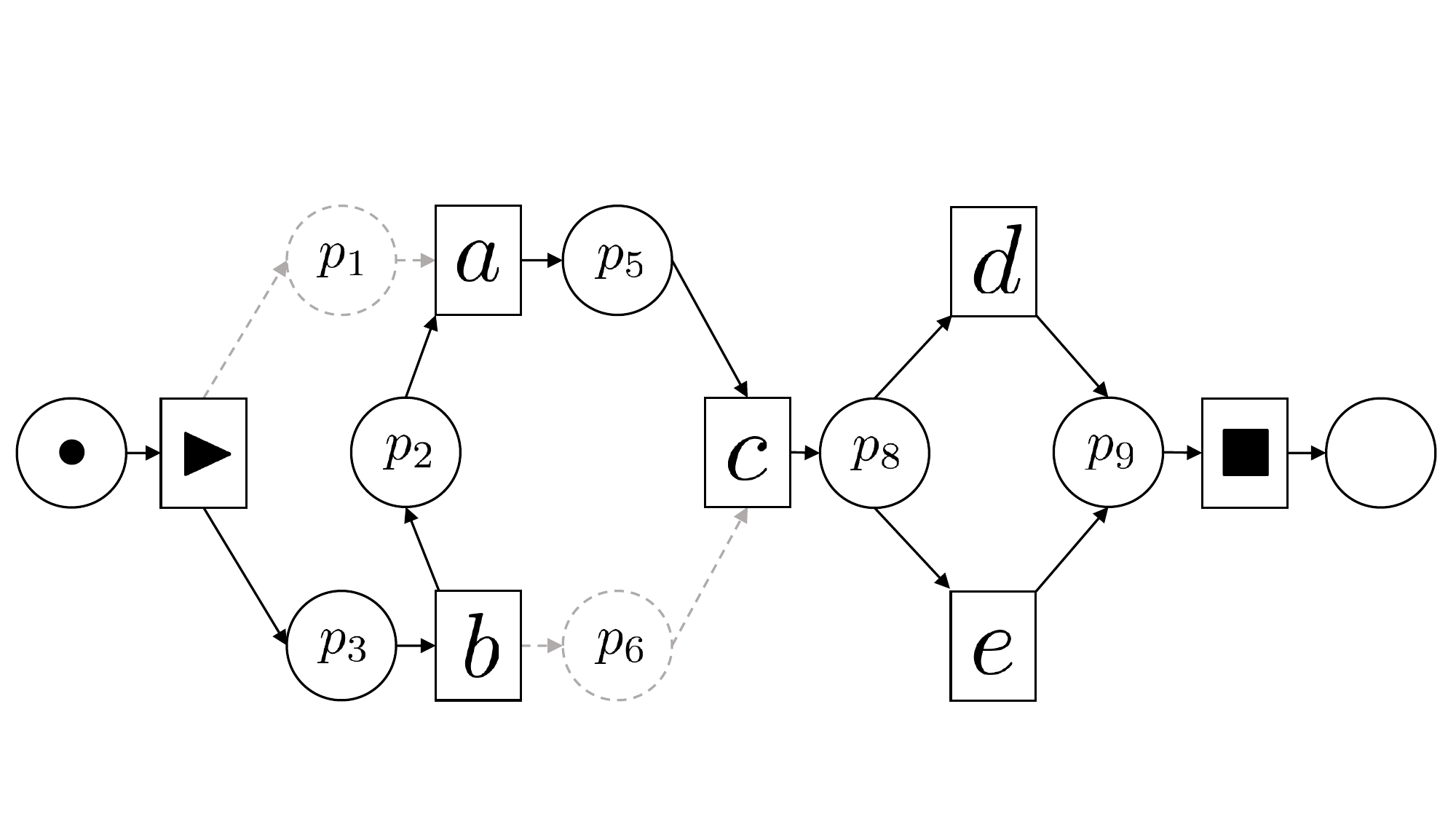}\vspace*{-5mm}
\caption{Consider the event log $[
\Trace{\sactivity, a,b,c,d,\eactivity}^{35},
\Trace{\sactivity, a,b,c,e,\eactivity}^{5},
\Trace{\sactivity, b,a,c,d,\eactivity}^{55},
\Trace{\sactivity, b,a,c,e,\eactivity}^{5}
]$. The Petri net above (implicit places marked with gray, dashed lines) is able to replay the two trace variants $\Trace{\sactivity, b,a,c,d,\eactivity}^{55}$ and $\Trace{\sactivity, b,a,c,e,\eactivity}^{5}$, which constitute a fraction of $0.6$ of the input event log.}
\label{fig:smallexamplecompleteSol2}\vspace*{-2mm}
\end{figure}

To avoid the issue of deadlocks caused by combining places, in Section~\ref{sec:fm:selection} we propose to extend the eST-Mining framework with a place selection strategy that inserts a maximal subset of all fitting places such that the resulting Petri net can replay at least a fraction of $\tau$ traces in the event log and does not contain transitions that can never be fired. With respect to the example above, this corresponds to the second solution (Figure~\ref{fig:smallexamplecompleteSol2}), which maximizes precision while still satisfying the given fitness threshold.

In the following we give an overview of the algorithmic framework including the proposed extensions.

\subsubsection*{Algorithmic Framework}
This work introduces an algorithmic framework extending the eST-Miner with the goal to address the two different causes for deadlocks described previously.
The fitness metric explored in Section~\ref{sec:fm:fitness} refines the evaluation of places.  In particular, it enhances the handling of places connected to rarely occurring activities (e.g., $p_7$ in Figure~\ref{fig:smallexamplecomplete}). To prevent deadlocks stemming from combinations of  places (e.g., $p_2$ and $p_4$ in Figure~\ref{fig:smallexamplecomplete}), we replace the straightforward place insertion of the eST-Miner with a strategy preventing such structures. This new place insertion strategy is detailed in  Section~\ref{sec:fm:selection}. The goal is to select a subset of fitting places such that the returned Petri net is guaranteed to replay at least a fraction of $\tau$ traces in the event log. At the same time, we aim to maximize simplicity and precision of the model.

The proposed algorithm extends the eST-Miner framework by adding and adapting certain submethods. An overview of the framework, including the positioning of our proposed extensions, is given in Figure~\ref{fig:frameworkoverview-highlevel} and described in the following. Details on the extensions will be provided in the corresponding sections.

\begin{figure}[ht]
\centering
\includegraphics[width =0.92\linewidth, trim={1cm 5cm 0 0.5cm},clip]{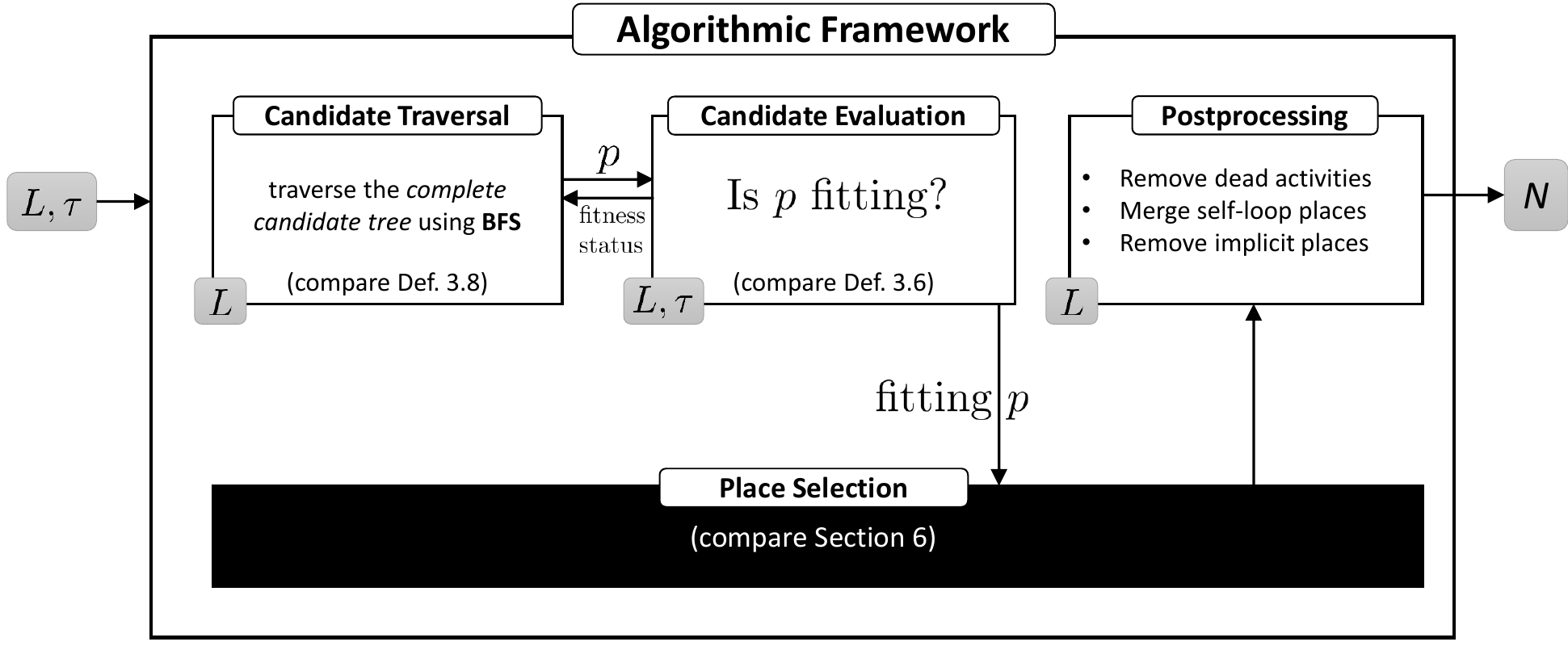}
\caption{{High-level overview of the proposed algorithmic framework. The newly introduced place selection submethod (represented as a black box) is described in detail in Section~\ref{sec:fm:selection}.}}
\label{fig:frameworkoverview-highlevel}
\end{figure}

\medskip
\textbf{Initialization: }
As the only preprocessing steps we add artificial start and end activities (\sactivity and \eactivity) to all traces in the given event log and initialize the output model $N$ as a Petri net with one labeled transition for each activity in the input event log (including \sactivity and \eactivity) and no places except for a marked start place \inoutpair{\emptyset}{\sactivity} and a final place \inoutpair{\eactivity}{\emptyset}.

\medskip
\textbf{Place Insertion: }
The eST-Miner traverses the complete candidate tree (\emph{Candidate Traversal}) to generate place candidates, which are evaluated one by one, possibly allowing to exploit monotonicity properties to skip uninteresting subtrees. Different traversal strategies of the tree are possible, e.g. Depth-First-Search, Breadth-First-Search or guided by heuristics. It is also possible to limit the depth of tree traversal, improving computation time while losing the option to discover any place below that depth. The extension proposed in this paper, with details provided in Section~\ref{sec:fm:selection}, requires a Breath-First-Search traversal strategy and allows limiting the traversal depth.

Every traversed place candidate is evaluated one by one (\emph{Candidate Evaluation}) to measure how fitting they are with respect to the event log $L$ and the noise threshold $\tau$. We can use absolute fitness, relative fitness or the new metric introduced in Section~\ref{sec:fm:fitness}.

The new place selection subroutine (\emph{Place Selection}) following place evaluation is introduced in detail in Section~\ref{sec:fm:selection}. The selected places are added to the Petri net by connecting them to their corresponding uniquely labeled transitions.

\medskip
\textbf{Model Simplification: }
After the selected fitting places have been added, we apply a set of simplification steps (\emph{Postprocessing}). First, the proposed place selection strategy may result in certain activities to be no longer included in the replayable part of the event log. The corresponding transitions are removed from the Petri net (details are provided in Section~\ref{sec:fm:selection}).

Second, all pairs of places of the form $\inoutpair{I \cup A_1}{O \cup A_1}$ and $\inoutpair{I \cup A_2}{O \cup A_2}$ are merged, i.e., replaced by a new place  $\inoutpair{I \cup A_1 \cup A_2}{O \cup A_1 \cup A_2}$. This merging of places with self-loops (with $A_1$ and $A_2$ being the sets of self-looping transitions) simplifies the Petri net without impacting its behavior. Note that this is necessary only when the tree traversal depth is limited, resulting in the two smaller places being discovered while the larger place connected to more transitions is cut off.

As the final postprocessing step we remove implicit places before returning the resulting Petri net.

\medskip
\textbf{Using the Extensions: }The newly introduced place selection strategy focuses on guaranteeing minimal fitness of the complete returned Petri net. However, as illustrated before, places that allow for replay of most of the event log may still have a devastating effect on the fitness of single activities, to the degree of blocking their execution completely. Consequently, the proposed fitness metric and the proposed selection technique complement each other but can also be applied independently.

When applying the place selection strategy with relative or absolute fitness, we expect most infrequent activities to be blocked (dead) and consequently be removed from the returned Petri net. This can be desired by users who want to obtain a model reflecting only the most important behavior or subprocess without applying manual preprocessing. However, since removing activities or traces from the event log that do not satisfy a certain frequency threshold is trivial in commonly used tools, it can be reasonable to expect a user to have applied such preprocessing. For users who have applied filtering such that the event log contains only activities and traces they are actually interested in, an ideal model represents the control-flow of all the activities remaining in the event log. We expect that the new fitness metric allows us to continue to abstract from infrequent behavioral patterns without outright removal of infrequent activities or complete trace variants.

\smallskip
Both scenarios are investigated and discussed in our evaluation.

\section{A new fitness metric}  \label{sec:fm:fitness}
To prevent the introduction of deadlocks as explained in Section~\ref{sec:fm:framework}, in this section we introduce \emph{aggregated fitness} as a new fitness metric. It refines the fitness evaluation of places with the goal to improve the handling of rarely occurring activities: places simply preventing such actvities from being fired will obtain a low fitness score.

\begin{definition}[Aggregated Fitness]\label{def:fitagg}
Let $L \in \multisetset{\universeoftraces}$ be an event log and let $p=\inoutpair{I}{O}$ be a place. We define the fitness metric \emph{aggregated fitness} as 
$$   \fitnessmetric^L_{\textit{agg}}(p)=
            \texttt{min}_{a \in (I \cup O)} \left(
            \frac{|\texttt{act}_L(\{a\}) \bignplus \pFit_L(p)|}{|\texttt{act}_L(\{a\})|}
            \right).$$
\end{definition}
\emph{Aggregated fitness} as introduced in Definition~\ref{def:fitagg} returns the minimal fitness computed for each individual transition connected to the place. To maintain the eST-Miners efficiency we require a fitness metric to allow for cutting off subtrees by exploiting the monotonicity properties of the places (Definition~\ref{theo:monotonicity}), i.e., the metric has to guarantee that there are no fitting places in the skipped subtrees of the complete candidate tree. In the following, we show that the proposed \emph{aggregated fitness} can satisfy this requirement.

\smallskip
To enable exploitation of the monotonicity results for skipping candidate subtrees while employing the noise threshold $\tau$ to filter infrequent behavior patterns, we extend the notions of underfed and overfed (compare Definition~\ref{def:fitting}) from single traces to the whole event log in a suitable way:

\begin{definition}[Aggregated Fitness: Underfed/Overfed Places with Respect to an Event Log] \label{def:fednessFMAgg}
Consider an event log $\log  \in \multisetset{\universeoftraces}$, a noise threshold $\tau \in \mathbb{R}_0^1$ and a place $p=\inoutpair{I}{O}$. If $p$ is unfitting with respect to \log and $\tau$ and \emph{aggregated fitness}, it can be classified further as
    \begin{itemize}
	\item \emph{underfed}, denoted by $\underfed[{\log,\tau}]{\textit{agg}}{\place}$, if and only if
    $\exists a \in I \cup O :
            \frac{|\texttt{act}_L(\{a\}) \bignplus \pNeg_L(p)|}{|\texttt{act}_L(\{a\})|}
            > 1-\tau$,
	\item \emph{overfed}, denoted by $\overfed[{\log,\tau}]{\textit{agg}}{\place}$,  if and only if
    $\exists a \in I \cup O :
            \frac{|\texttt{act}_L(\{a\}) \bignplus \pPos_L(p)|}{|\texttt{act}_L(\{a\})|}
            > 1-\tau$.
    \end{itemize}
\end{definition}
As we have seen before for absolute and relative fitness, Definition~\ref{def:fednessFMAgg} allows for aggregated fitness that a place is underfed and overfed at the same time, given that sufficiently many traces validate the necessary requirements, or is unfitting without being neither underfed nor overfed.

\medskip
To skip parts of the candidate space without missing fitting places, we need to guarantee, that if a place is determined to be underfed (overfed), the same is guaranteed to be true for its blue (red) descendants in the complete candidate tree.

\begin{lemma}[Monotonicity of Aggregated Fitness]\label{lemma:monotonicityagg}
Consider an event log $L  \in \multisetset{\universeoftraces}$ and a noise threshold $\tau \in \mathbb{R}_0^1$.
For three places $p=\inoutpair{I}{O}$, $p'=\inoutpair{I}{O'}$ with $O \subseteq O'$ and $p''=\inoutpair{I''}{O}$ with $I \subseteq I''$the following implications hold:
\begin{align*}
\underfed[{\log,\tau}]{\textit{agg}}{\place} &\implies \underfed[{\log,\tau}]{\textit{agg}}{\place'},\\
\overfed[{\log,\tau}]{\textit{agg}}{\place} &\implies \overfed[{\log,\tau}]{\textit{agg}}{\place''}.
\end{align*}
\end{lemma}

\begin{proof}
First, assume that $\underfed[{\log,\tau}]{\textit{agg}}{\place}$. Then, by definition, there exists some activity $a \in I \cup O$ such that $\frac{|\texttt{act}_L(\{a\}) \bignplus \pNeg_L(p)|}{|\texttt{act}_L(\{a\})|} > 1-\tau$. Because of $O \subseteq O'$ we know that $a \in I \cup O'$. Also, by Theorem~\ref{theo:monotonicity} we have that $\pNeg_L(p) \subseteq \pNeg_L(p')$. Together, this implies that $\frac{|\texttt{act}_L(\{a\}) \bignplus \pNeg_L(p')|}{|\texttt{act}_L(\{a\})|} \geq \frac{|\texttt{act}_L(\{a\}) \bignplus \pNeg_L(p)|}{|\texttt{act}_L(\{a\})|} > 1-\tau$, and thus by definition $\underfed[{\log,\tau}]{\textit{agg}}{\place'}$.

Now, assume that $\overfed[{\log,\tau}]{\textit{agg}}{\place}$. Then, by definition, there exists an activity $a \in I \cup O$ such that $\frac{|\texttt{act}_L(\{a\}) \bignplus \pPos_L(p)|}{|\texttt{act}_L(\{a\})|} > 1-\tau$. Because of $I \subseteq I''$ we know that $a \in I'' \cup O$. By Theorem~\ref{theo:monotonicity} we have that $\pPos_L(p) \subseteq \pPos_L(p'')$. Together, this implies that $\frac{|\texttt{act}_L(\{a\}) \bignplus \pPos_L(p'')|}{|\texttt{act}_L(\{a\})|} \geq \frac{|\texttt{act}_L(\{a\}) \bignplus \pPos_L(p)|}{|\texttt{act}_L(\{a\})|} > 1-\tau$, and thus by definition $\overfed[{\log,\tau}]{\textit{agg}}{\place''}$.
\end{proof}

We have shown that the monotonicity properties of the newly introduced \emph{aggregated fitness} enable the same efficiency optimizations by skipping candidate subtrees as the absolute and relative fitness. Thus, this metric can be easily integrated into the eST-Miners framework. To compare the impact on the place fitness evaluation, we continue to investigate relationships between these fitness metrics.

\begin{lemma}[Relationships Between Fitness Metrics]\label{lemma:relationships}
Let $L  \in \multisetset{\universeoftraces}$ be an event log and let $p=\inoutpair{I}{O}$ be a place. Then the following holds:
$$\fitnessmetric^L_{\textit{rel}}(p) \leq \fitnessmetric^L_{\textit{abs}}(p).$$
\end{lemma}

\begin{proof}
With respect to the place $p$, we can partition the event log $L$ into the following disjunct multisets of traces:
\begin{align*}
    L^\boxempty_\checkmark &= \pFit_L(p) \nplus \texttt{act}_L(I \cup O)
        \textit{, the multiset of fitting, activated traces.}\\
    L^\boxtimes_\checkmark &= (\pNeg_L(p) \uplus \pPos_L(p)) \nplus \texttt{act}_L(I \cup O)
        \textit{, the multiset of unfitting, activated traces.}\\
    L_\times &= L \backslash \texttt{act}_L(I \cup O)
        \textit{, the multiset of non-activated traces.}
\end{align*}
Note that $|L^\boxempty_\checkmark|+|L^\boxtimes_\checkmark|>0$, since we consider only activities contained in $L$.
In the following, we use these sets to rewrite absolute and relative fitness and show the claimed relation between them:
\begin{align*}
& & \fitnessmetric^L_{\textit{rel}}(p) &\leq \fitnessmetric^L_{\textit{abs}}(p)\\
&\Leftrightarrow & \frac{|L^\boxempty_\checkmark|}{|L^\boxempty_\checkmark|+|L^\boxtimes_\checkmark|}
                &\leq
                \frac{|L^\boxempty_\checkmark|+|L_\times|}{|L^\boxempty_\checkmark|+|L^\boxtimes_\checkmark|+|L_\times|} \\
&\Leftrightarrow& {|L^\boxempty_\checkmark|} \cdot \left({|L^\boxempty_\checkmark|+|L^\boxtimes_\checkmark|+|L_\times|}\right)
                &\leq
                \left({|L^\boxempty_\checkmark|+|L_\times|}\right) \cdot \left({|L^\boxempty_\checkmark|+|L^\boxtimes_\checkmark|}\right) \\
&\Leftrightarrow& |L^\boxempty_\checkmark|^2 + |L^\boxempty_\checkmark| \cdot |L^\boxtimes_\checkmark|+ |L^\boxempty_\checkmark| \cdot |L_\times|
                &\leq
                |L^\boxempty_\checkmark|^2 +|L_\times| \cdot |L^\boxempty_\checkmark|+|L^\boxempty_\checkmark| \cdot L^\boxtimes_\checkmark|+|L_\times|\cdot |L^\boxtimes_\checkmark| \\
&\Leftrightarrow& 0
                &\leq
               |L_\times|\cdot |L^\boxtimes_\checkmark|
\end{align*}

\vspace*{-7mm}
\end{proof}

We have shown that relative fitness is at least as strict as absolute fitness. Therefore, using relative fitness instead of absolute fitness during candidate place evaluation guarantees the discovered places to be at least as fitting.
Unfortunately, no similar relation holds between $\fitnessmetric^L_{\textit{agg}}(p)$ and $\fitnessmetric^L_{\textit{rel}}(p)$ as illustrated by the example in Figure~\ref{fig:counterexample}. Consider the given place $p$ and event log $L_1$: we have that  $\fitnessmetric_\textit{agg}^{L_1}(p) = \texttt{min}(\frac{90}{90}, \frac{0}{10}) = 0$, which is strictly smaller than $\fitnessmetric_\textit{rel}^{L_1}(p) = \frac{90}{100}$. However, considering the same place $p$ and the event log $L_2$, we have that $\fitnessmetric_\textit{agg}^{L_2}(p) = \texttt{min}(\frac{33}{33}, \frac{33}{66}, \frac{33}{66}) = \frac{1}{2}$ is strictly larger than $\fitnessmetric_\textit{rel}^{L_2}(p) = \frac{33}{99}=\frac{1}{3}$.
\begin{figure}[tbh]
\vspace*{-1mm}
    \centering
    \begin{multicols}{2}
    \scalebox{0.8}{\includegraphics[width =0.4\linewidth, trim={0 10.5cm 23.7cm 0},clip]{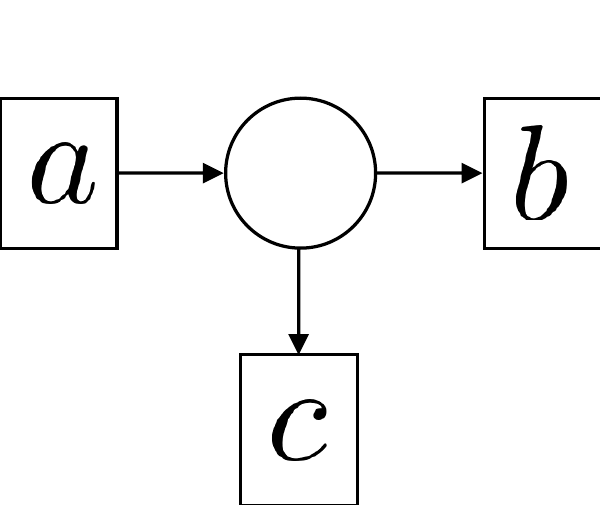} }\\
   \begin{align*}
   &L_1=[{\Trace{a,b}^{90}}, \textcolor{gray}{\Trace{x,y}^{20}}, \textcolor{red}{\Trace{c}^{10}}]\\
   &L_2=[{\Trace{a,b,a,c}^{33}}, \textcolor{gray}{\Trace{x}^{1}}, \textcolor{red}{\Trace{b}^{33}}, \textcolor{red}{\Trace{c}^{33}}]
   \end{align*}
    \end{multicols}\vspace*{-3mm}
    \caption{Traces in $L_1$ and $L_2$ which are unfitting with respect to the place $p=\inoutpair{a}{b,c}$ are marked in red, while traces not activated are colored gray. For $L_1$ we have that $\fitnessmetric_\textit{agg}^{L_1}(p) < \fitnessmetric_\textit{rel}^{L_1}(p)$, while for $L_2$ we have that $\fitnessmetric_\textit{agg}^{L_2}(p) > \fitnessmetric_\textit{rel}^{L_2}(p)$.}
    \label{fig:counterexample}
\end{figure}

To obtain a metric that guarantees to satisfy a given fitness threshold $\tau$ with respect to all three fitness metrics discussed, we define \emph{combined fitness} as a final fitness metric. Combined fitness simply takes the minimum of all three fitness metrics and therefore inherits their monotonicity properties, i.e., when using it to skip parts of the candidate space we maintain the guarantee to not skip any fitting place candidate. More specifically, if a place is fitting with respect to combined fitness, then it is also fitting with respect to each individual fitness metric, and if a place is underfed (overfed) with respect to at least one of the individual metrics then the descendants in the corresponding subtree are also underfed (overfed) with respect to that fitness metric.

\begin{samepage}
\begin{definition}[Combined Fitness of a Place]
Let $L  \in \multisetset{\universeoftraces}$ be an event log and let $p=\inoutpair{I}{O}$ be a place. We define the \emph{combined fitness} of the place $p$ with respect to $L$ as

{\centering
$\fitnessmetric_\textit{comb}^L(p) = \texttt{min}(\fitnessmetric^L_{\textit{abs}}(p),\fitnessmetric^L_{\textit{rel}}(p), \fitnessmetric^L_{\textit{agg}}(p)).$
\par
}
\end{definition}
\end{samepage}

\eject\noindent
{\emph{Combined fitness} is rather restrictive, mostly due to the newly introduced aggregated fitness. A more forgiving aggregation function than the minimum, e.g., such as average, median or harmonic mean, would be interesting to investigate as an alternative. Unfortunately, such aggregations would not allow for the skipping of candidate subtrees in the eST-Miners complete candidate tree, since they do not provide the same monotoncity properties. For example, consider the place $p=\inoutpair{a}{b,c}$ and the event log $L=[\Trace{\sactivity, a,b,b,\eactivity}^1, \Trace{\sactivity, a,c, \eactivity}^1,\Trace{\sactivity,a,d,\eactivity}^1]$. The aggregated fitness of $p$ as defined is $\fitnessmetric^L_\textit{agg}(p)=\texttt{min}(\frac{2}{3}, \frac{0}{1}, \frac{1}{1}) = 0$ and $p$ would be considered underfed because of $b$, allowing us to skip its blue subtree. When taking the average instead of the minimum, the score would be $\texttt{average}(\frac{2}{3}, \frac{0}{1}, \frac{1}{1}) \approx 0.55$. Consider now the place $p'=\inoutpair{a}{b,c,d}$ constructed from $p$ by adding the outgoing activity $d$. The aggregated fitness of $p'$ is $\fitnessmetric^L_\textit{agg}=\texttt{min}(\frac{2}{3}, \frac{0}{1}, \frac{1}{1}, \frac{1}{1}) = 0$ and $p'$ would also be underfed, as expected. However, when taking the average instead of minimum, the score would be $\texttt{average}(\frac{2}{3}, \frac{0}{1}, \frac{1}{1}, \frac{1}{1}) \approx 0.66$. This example clearly illustrates that the monotonicity properties necessary to skip subtrees without losing guarantees do not hold for such aggregation functions. }

The \emph{combined fitness} is the strictest fitness metric introduced in this work and therefore guarantees that any (set of) places satisfying a threshold $\tau$ with respect to combined fitness also satisfies $\tau$ with respect to absolute, relative and aggregated fitness. We can still choose to be less restrictive by lowering the threshold $\tau$. Note that due to Lemma~\ref{lemma:relationships} we do not need to compute absolute fitness in the implementation of combined fitness.
In Section~\ref{sec:fm:eval} we apply the proposed eST-Miner variant with \emph{combined fitness} as well as \emph{relative fitness} as a fitness metric to various event lots to evaluate their impact on real-life data.

\section{Place selection}  \label{sec:fm:selection}
\noindent
The eST-Miner evaluates all candidate places and discovers a set of places fitting the input event log based on the noise threshold $\tau$ and the chosen fitness metric. As illustrated by the simple example in Figure~\ref{fig:smallexamplecomplete}, simply adding all fitting places to a Petri net indiscriminately may result in deadlocks and unnecessary complexity. In this section, we propose an approach aiming to mitigate this problem by selecting a suitable subset of places.

\medskip
To motivate our strategy, we first discuss a more complex example. Consider the event log $L$ and the set of places $p_1$ to $p_8$ in Figure~\ref{fig:motivationalSelection}. For each of the three fitness metrics introduced in Definitions~\ref{def:fit} and~\ref{def:fitagg}, all places in this (incomplete) subset of candidate places are \emph{fitting} with respect to $L$ and $\tau = 0.75$. Inserting all of these places results in the given Petri net $N$, which can replay only the first trace variant in $L$, corresponding to $60$ \% of the traces. The introductory example in Figure~\ref{fig:smallexamplecomplete} illustrates, that the fraction of replayable traces may even decrease to $0$. Such a result is undesirable, since it is unnecessarily complex with respect to the behavior it represents, not free of dead parts and likely to disappoint user expectations with respect to fitness. In the following, we explore strategies to return a deadlock free Petri net that is guaranteed to replay at least a fraction of $\tau$ traces in the event log by inserting only a selection of the discovered fitting places.
\begin{figure}[!ht]
\centering
\begin{adjustbox}{width={\textwidth}}
\begin{tabular}{c|c|c|c|c|c|c|c|c|c|c}
\textbf{{ID}} &\textbf{{Traces in $L$}} &\textbf{{$p_1$}} &\textbf{{$p_2$}} &\textbf{{$p_3$}} &\textbf{{$p_4$}} &\textbf{{$p_5$}} &\textbf{{$p_6$}} &\textbf{{$p_7$}} &\textbf{{$p_8$}} & {$N$} \\
 &  &\textbf{{$(\sactivity|a)$}} &\textbf{{$(a|c)$}} &\textbf{{$(a|b)$}} &\textbf{{$(c|e)$}} &\textbf{{$(b|e)$}} &\textbf{{$(e|\eactivity)$}} &\textbf{{$(b|c,d)$}} &\textbf{{$(d,e|\eactivity)$ }} &\\
\hline
1 & {$\langle \sactivity, a,b,c,e, \eactivity \rangle^{60}$} & $\checkmark $ & $\checkmark$  &$\checkmark$ & $\checkmark$&
         $\checkmark$ & $\checkmark$ &$\checkmark$ &$\checkmark$ & $\checkmark $ \\
2 & {$\langle \sactivity, a,b,d, \eactivity \rangle^{20}$} & $\checkmark $ &
      $ \bf{\times} $  &$\checkmark$ & $\checkmark$ & $ \bf{\times} $ & $ \bf{\times}$ &$\checkmark$ &$\checkmark$& $ {\bf \times}$\\
3 & {$\langle \sactivity, a,c,b,e, \eactivity \rangle^{15}$} & $\checkmark $ & $\checkmark$  &$\checkmark$ & $\checkmark$&$\checkmark$ &$\checkmark$ &$ {\bf \times}$ &$\checkmark$ &$ \usym2613$\\
4 & {$\langle \sactivity, a,b,d,e, \eactivity \rangle^{5}$} & $\checkmark $ & $ {\bf \times}$  &$\checkmark$ & $\usym2613$&$\checkmark$ &$\checkmark$ &$\checkmark$ &$ {\bf \times}$ &$ {\bf \times}$\\
\end{tabular}
\end{adjustbox}\\
\vspace*{4mm}
\includegraphics[width =0.68\linewidth]{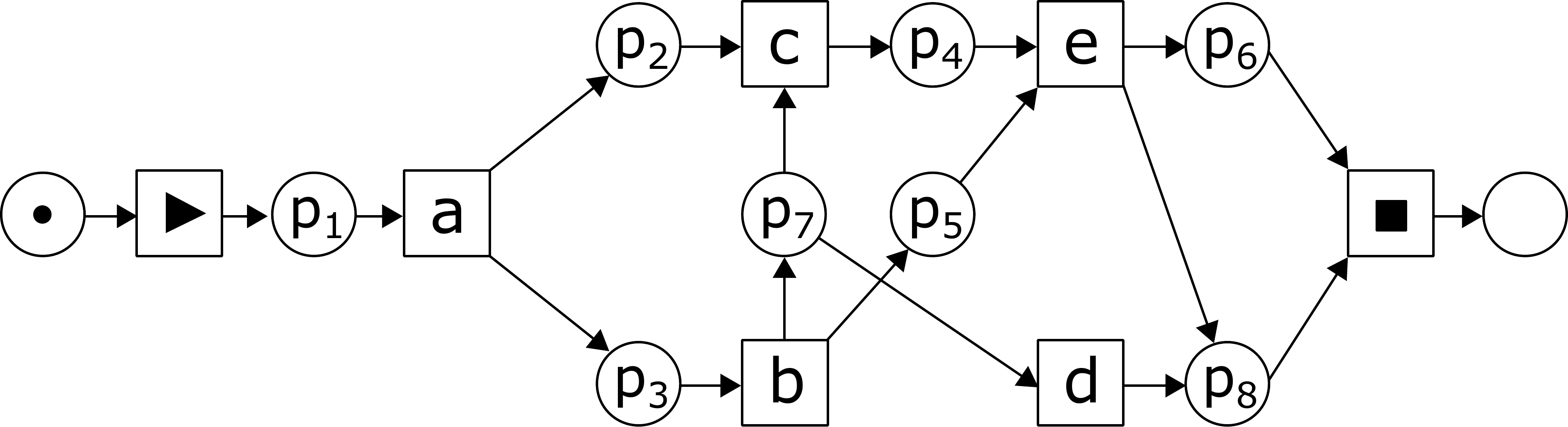}
\caption{{The table indicates for each of the given trace variants and candidate places whether the place can replay that trace variant. The Petri net $N$ is created by inserting all these places and can replay only the first trace variant, i.e., $0.6 \cdot |L| =60$ traces.}}
\label{fig:motivationalSelection}\vspace*{10mm}
%

\centering
\includegraphics[width =0.9\linewidth]{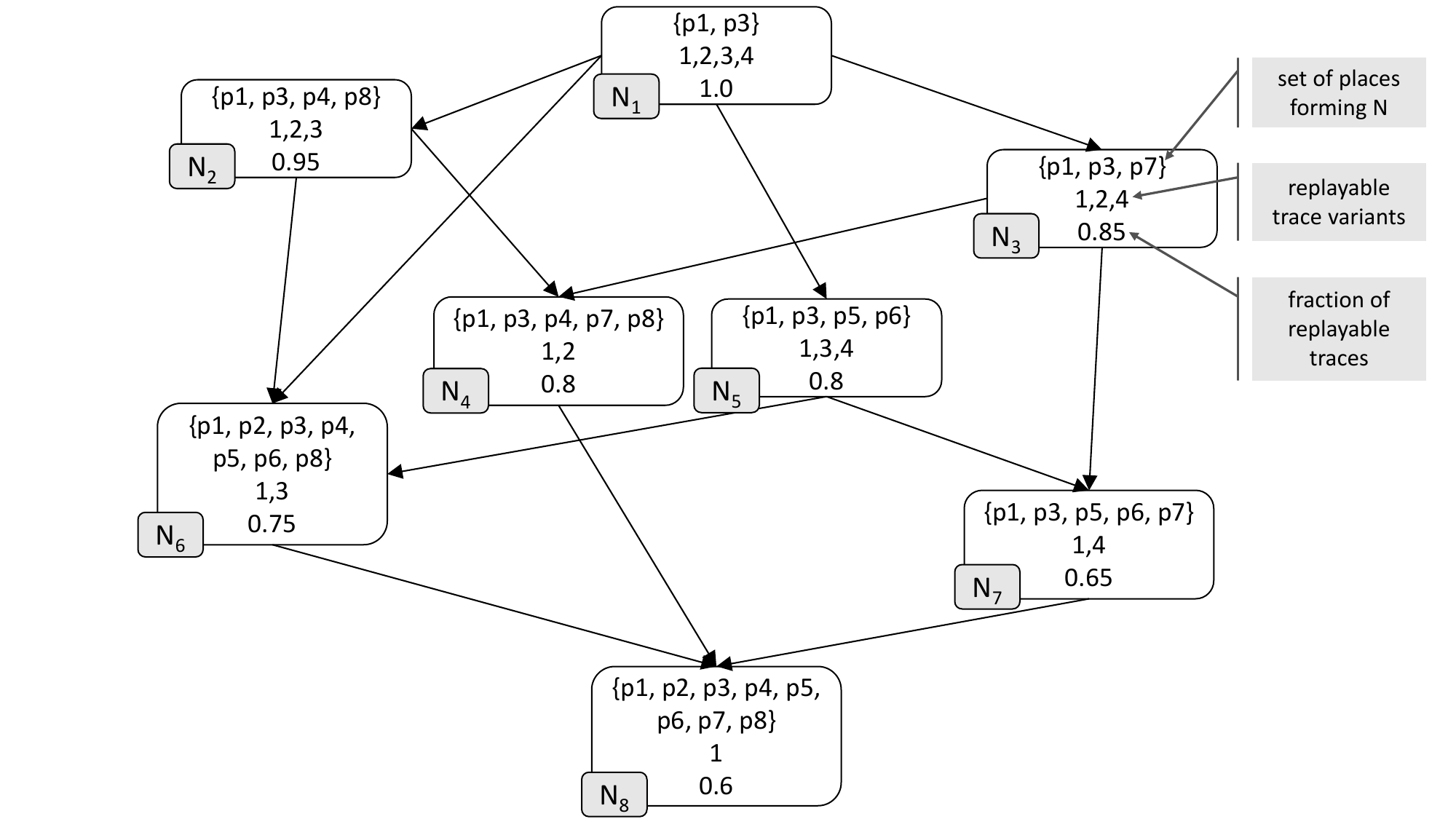}
\caption{{Consider the set of places given in Figure~\ref{fig:motivationalSelection}. This figure shows all possible combinations of these places such that adding any other place to the corresponding Petri net would decrease the number of replayable log traces. Each set of places, i.e., Petri net, is annotated with the list of trace variants it can replay and the corresponding fraction of log traces. Note that $N_8$ corresponds to the Petri net shown in Figure~\ref{fig:motivationalSelection}. }}
\label{fig:combinations}\vspace*{-2mm}
\end{figure}

Consider the set of all fitting places discovered during the place candidate evaluation of the eST-Miner. Selecting an adequate subset of these places, such that also the resulting Petri net as a whole satisfies the noise threshold $\tau$,  is challenging for a variety of reasons. While trivial solutions exist, such as not inserting any place at all, defining what the best solution would be is not straightforward, since several maximal subsets of places satisfying this requirement may exist. These subsets may differ, for example, in size, fraction of replayable traces, place complexity (number of connected activities) or subjective 'interestingness' measures. Figure~\ref{fig:combinations} illustrates all maximal sets of places that can be built from the example places given in Figure~\ref{fig:motivationalSelection}. These sets are maximal in the sense that adding any of the other places would decrease the number of replayable log traces. Depending on the choice of the minimal fitness threshold~$\tau$, the solution considered optimal by the user is unclear.

\medskip
Furthermore, even if we have somehow obtained a notion of optimality, first collecting all fitting places and then computing an optimal solution can quickly become unfeasible, both in terms of time complexity and memory requirements. This is due to the huge number of fitting but potentially implicit places discovered by the eST-Miner. Unfortunately, knowledge of which places are contained in the Petri net is required to identify implicit places reliably.

To circumvent the issue of time and space complexity, we combine the eST-Miners sequential place evaluation procedure with a guided greedy place selection approach, which is described in detail in Subsections~\ref{sec:placeclass} and~\ref{sec:framework}. In the absence of a clear notion of optimality, we propose and investigate several heuristic selection strategies and evaluate their impact on different quality aspects of the returned Petri net. In this paper, we consider fitness, precision, and simplicity as desirable properties.

A model with high fitness can express most of the behavior seen in the event log. High precision means that the model does not allow for a lot of behavior not seen in the event log. Simplicity refers to the general readability and understandability of the model and is therefore inherently subjective. In this work, we approximate simplicity by assuming that fewer arcs indicate a simpler model.  While generalization (avoiding overfitting with respect  to the sample process executions given in the event log) is desirable, additional information would be required to evaluate it, which is why we consider it  outside the scope of this work.  The concrete metrics used in this work to evaluate quality will be briefly discussed in Section~\ref{sec:fm:eval}. For further details, we refer the reader to~\cite{PMbook, bookconformance}.

\subsection{Place classification}\label{sec:placeclass}

When making the decision to insert a place into the model, this reduces the possible choices we can make later on: the place constrains the behavior of the model and only places with a sufficiently large intersection of replayable traces can be added to the model at a later point. Consider the example place combinations in Figure~\ref{fig:combinations} with a fitness threshold of $\tau = 0.75$ and assume that the model already contains the places $p_1$ and $p_3$. If the next fitting place we discover is $p_7$, and we immediately insert it into the Petri net, we can no longer discover a Petri net including, for example, $p_6$ without violating our fitness constraint. Such choices may prevent us from discovering a more desirable solution. Therefore, we aim to capture the main behavior of the log by using heuristics to postpone, or even disallow, the addition of very restrictive places.

\medskip
To this end, we introduce a new parameter $\delta$ which is our main tool to guide the choice of places while balancing fitness, precision, and simplicity. This $\delta$ specifies the largest acceptable reduction in replayable traces when adding a place to the model. Optionally, $\delta$ can be adapted for each place individually using an adaption function \texttt{adapt} to favor certain places over others, according to the user's preferences. Favored places can be added earlier, despite being rather restrictive, while other places will be added only if they do not constrain the behavior too much. Such adaption strategies are discussed in Section~\ref{sec:fm:strategies}.

Definition~\ref{def:pclassify} formalizes the use of $\tau$, $\delta$ and \texttt{adapt} to decide for a newly discovered fitting place $p$, whether the algorithm should \emph{add} it to the selected set of places $P$, \emph{keep} it for later re-evaluation or \emph{discard} it forever.
\begin{definition}[Place Classification Using $\tau, \delta$ and \texttt{adapt}] \label{def:pclassify}
Consider an event log $L  \in \multisetset{\universeoftraces}$ over the set of activities $A \in \universeofactivities$, a set of places $P \subseteq \powerset{A} \times \powerset{A}$, and a place $p \in \powerset{A} \times \powerset{A}$. We use parameters $\tau \in \mathbb{R}_0^1$ and $\delta \in \mathbb{R}_0^1$, and a function $\texttt{adapt} \colon \mathbb{R}_0^1 \times (\powerset{A} \times \powerset{A}) \rightarrow \mathbb{R}_0^1$ to categorize $p$ as follows:
\begin{align*}
	\texttt{keep}_{L, \tau}(P, p) &= |\pFit_L(P) \nplus \pFit_L(p)| \geq \tau \cdot |L|\\
	\texttt{add}_{L, \tau, \delta}(P, p) &= \texttt{keep}_{L, \tau}(P, p) \\
		&\wedge |\pFit_L(P)| - |\pFit_L(P) \nplus \pFit_L(p)| \leq \texttt{adapt}(\delta, p) \cdot |L|
\end{align*}
If $\texttt{keep}_{L, \tau}(P, p)$ does not hold, $p$ will be discarded.
\end{definition}
In the following subsection, we give an overview of the complete approach.

\subsection{Selection framework}\label{sec:framework}

\begin{figure}[!b]
\centering
\includegraphics[width =0.98\linewidth, trim={0 0.4cm 0 0.5cm},clip]{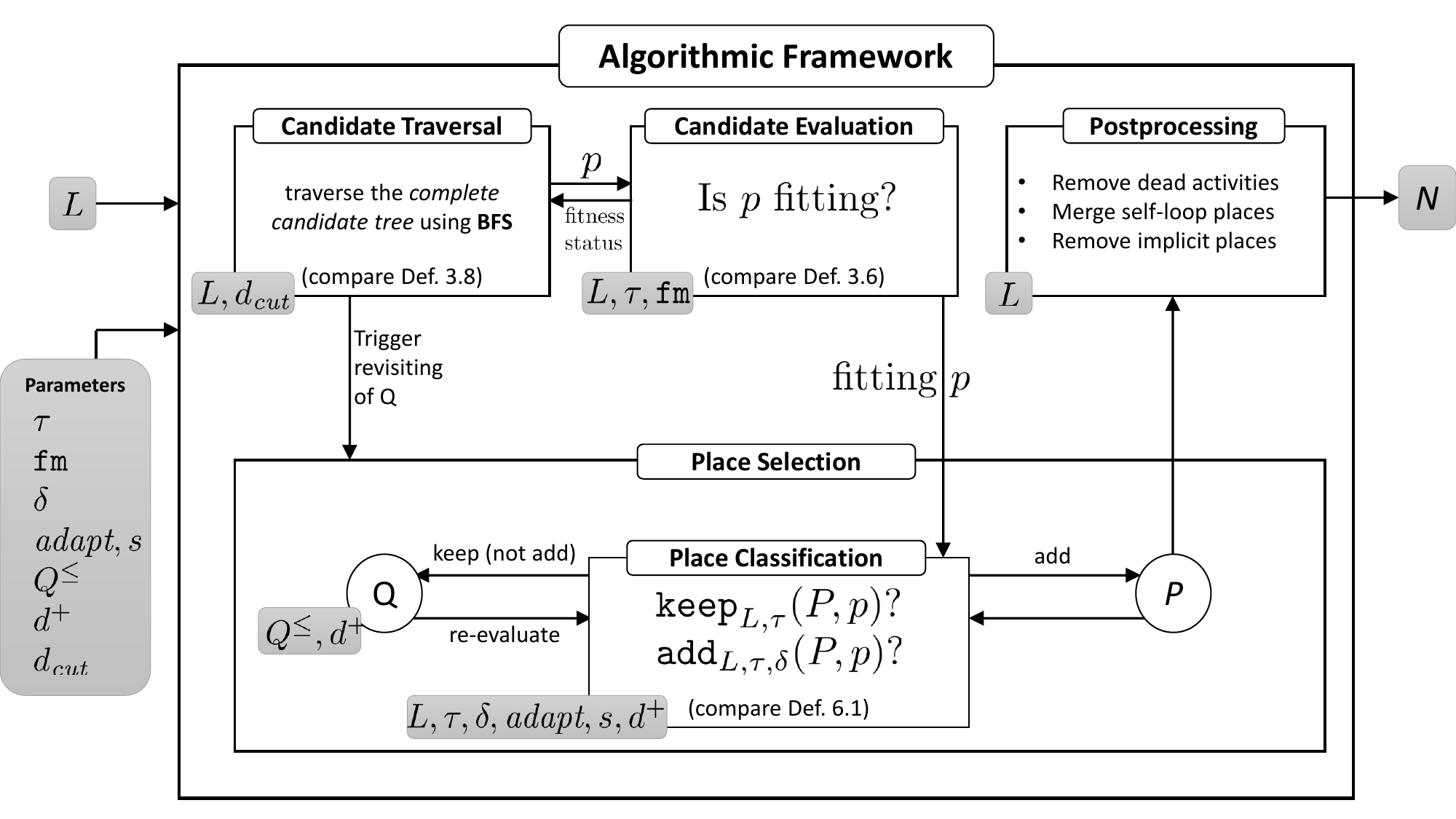}
\caption{Overview of the presented approach, including input, output, and parameter use.}
\label{fig:frameworkoverview}
\end{figure}
An overview of our approach, indicating inputs, outputs and use of parameters, is given in Figure~\ref{fig:frameworkoverview}.
Since we consider the simplicity of the model to be a desirable property, we set the eST-Miner to traverse the \emph{complete candidate tree} using BFS rather than DFS. Thus, places with few connected activities are evaluated first and can therefore be inserted into the model at an earlier stage. Furthermore, we limit the traversal depth to places with $d_\textit{cut}$ activities. This can not only significantly improve the running time but also balances precision and fitness on the one hand and simplicity and generalization on the other hand by preventing the algorithm to discover complex places located at deeper levels of the complete candidate tree. While some complex behaviors may need such places to be expressed exactly, inserting them into the Petri net is usually devastating to readability and in practical applications their constraints can often be sufficiently approximated by much simpler places.

\medskip
After the eST-Miner framework evaluates a candidate place $p$ to be fitting with respect to a threshold $\tau$, we use the \emph{classification functions} given in Definition~\ref{def:pclassify} to decide whether the place should immediately be added to the output Petri net, discarded forever or kept for re-evaluation. In the latter case it is added to a queue $Q$ of potential places which is sorted according to how interesting a place is.
In our case, we sort first by place simplicity (few transitions are better) and second by the number of replayable log traces. Optionally, the length of $Q$ can be limited to $\Qlength$ with the least interesting places being dropped if necessary, thus trading an improvement in time and space complexity for potentially lowered model quality.

Whenever the BFS candidate traversal reaches a new level in the \emph{complete candidate tree}, we revisit the potential places queue $Q$ and re-evaluate its places using the classification functions before proceeding with the traversal {of more complex places. This makes sense to promote simplicity in particular together with the sigmoid delta adaption function proposed in Section}~\ref{sec:fm:strategies}, which gives preference to places less complex than indicated by the current tree level. After reaching the lowest tree level, the approach continues to iterate over the potential places queue repeatedly {$d^+$} times. This can be relevant for delta adaption functions {depending on place complexity and current tree depth}, as exemplified in Section~\ref{sec:fm:strategies}: with each iteration the current tree depth parameter of the adaption functions is incremented (hence the term \emph{artificial tree depth}) allowing for gradually increased leniency also for the most complex places evaluated.

\medskip
Finally, the resulting Petri net $N=(A, P)$ may contain dead parts: activities which occur only in the subset of log traces that are no longer replayable by $N$ are not guaranteed to be executable at all. Therefore, as a final step, we detect and remove all activities that do not occur in $\pFit_L(P)$  together with their connected arcs. Before returning this Petri net as final output, the eST-Miner framework removes implicit places, and merges self-looping places when applicable (see Section~\ref{sec:fm:framework}).

{The approach returns a Petri net $N$ satisfying the following guarantees.}

\begin{theorem}[Guarantees]
Given an event log $L  \in \multisetset{\universeoftraces}$ over activities $A$, parameters $\tau \in \mathbb{R}_0^1, \delta \in \mathbb{R}_0^1, s \in \mathbb{N}, \Qlength \in \mathbb{N}, d^{+} \in \mathbb{N}, d_\textit{cut} \in \mathbb{N}$ and an adaption function
$\texttt{adapt} \colon \mathbb{R}_0^1 \times (\powerset{A} \times \powerset{A}) \rightarrow \mathbb{R}_0^1$, the eST-Miner extended with the place selection strategy given in Definition~\ref{def:pclassify} computes a Petri net {${N\!=\!(A', P)}$} with $A'\! \subseteq\! A$, such that $N$ can replay at least $\tau \cdot |L|$ traces from $L$ and every transition in $A'$ can be fired at least~once.
\end{theorem}

\begin{proof} The algorithm initializes the Petri net $N_0 = (A, \{\inoutpair{\emptyset}{\sactivity}, \inoutpair{\eactivity}{\emptyset}\})$ with one transition for each activity in \log. There is no place constraining the behavior of $N_0$ except for \{\inoutpair{\emptyset}{\sactivity} and \inoutpair{\eactivity}{\emptyset}\}, which allow for $\sactivity$ and \eactivity to be fired exactly once each. According to our trace definition (Definition~\ref{def:atl}), these activities do occur exactly once in each trace. Thus, $N_0$ can replay at least $\tau \cdot |L|$.
The method then iteratively adds places. According to Definition~\ref{def:pclassify} a place $p$ can be added to a Petri net $N_{1}=(A_{1}, P_{1})$ only if $\texttt{add}_{L, \tau, \delta}(P_{1}, p))$ holds, which requires  $\texttt{keep}_{L, \tau}(P_{1}, p)$ to hold. This requirement ensures that $|\pFit_L(P_{1}) \nplus \pFit_L(p)| \geq \tau \cdot |L|$, i.e., the Petri net with the place $p$ added, $N_{2}=(A_{1}, P_{1} \cup \{p\})$ can replay at least $\tau \cdot |L|$ traces from \log.

Since the requirement must hold for all added places, no further transitions are added and every transition that is not part of the replayable traces is removed, the final returned Petri net $N=(A', P)$ can replay at least $\tau \cdot |L|$ traces from \log, $A' \subseteq A$ holds and every transition can be fired at least once (when replaying a trace including the corresponding activity).
\end{proof}

 {Furthermore, if the length of $Q$ is not limited, and thus a place $p$ is discarded only if it does not satisfy $\texttt{keep}_{L, A, \tau}(P,p)$, the set of places $P$ is maximal in the sense that no place from the set of evaluated candidate places can be added without violating the fitness constraints imposed by the chosen heuristics.}

\subsection{Selection strategies} \label{sec:fm:strategies}

As  illustrated by the example place combinations in Figure~\ref{fig:combinations}, the order of places added can have a significant impact on the selected subset of places and, thus, the behavior of the returned Petri net. The presented framework allows for a wide range of heuristic functions, optimizing the place selection individually towards a variety of possible user interests. Thus, obviously, the examples presented in the following are by far not exhaustive and entirely different choices are possible, but they can serve as a starting point for an investigation of the impact and suitability of our approach.

\medskip
The \emph{sigmoid} delta adaption function aims to promote fitness and simplicity. The \emph{constant} and \emph{no delta} delta adaption functions are introduced to be used as a baseline in our experiments, towards which the effect of the sigmoid delta adaption function can be compared.

\subsubsection*{No Delta}
As a baseline to compare to, we introduce a function that ignores the parameter $\delta$ and simply adds every fitting place to the Petri net as soon as it is discovered. Within the framework, this can be formalized to
$$\texttt{adapt}_\textit{noDelta}(\delta, p) = 1.$$

\subsubsection*{Constant Delta}
Trivially, we can choose not to adapt delta at all. We simply add every fitting, non-discarded place that does not reduce the replayable traces from the log by a fraction of more than delta. Formally, this resembles the identity function:
$$\texttt{adapt}_\textit{constant}(\delta, p) = \delta$$

\subsubsection*{Sigmoid Delta Adaption}
While optimizing towards fitness as well as simplicity, we can balance the two forces in different ways based on the details of the adaption function. In this work investigate the \emph{sigmoid} delta adaption as defined in the following and visualized in Figure~\ref{fig:linearsigmoid}.

\begin{figure}[tbh]
\centering 
\includegraphics[width =0.47\linewidth, trim={0 0 12.5cm 0}, clip]{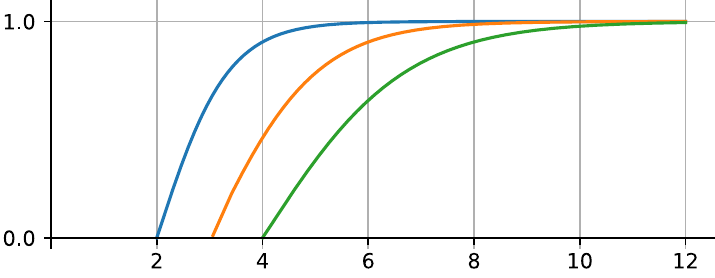}\vspace*{-1mm}
\caption{{Example behavior of the delta adaption modifier $\textit{mod}^{s,d}_\textit{sigmoid}(\delta)$ for three places with $2, 3$, and $4$ activities, respectively. The x-axis indicates the current tree depth $d$, with $d_\textit{max}=12$, while the y-axis indicates the modifier to be multiplied with $\delta$.}}
\label{fig:linearsigmoid}
\end{figure}

Given a set of activities $A$, the maximum depth of the complete candidate tree is $d_\textit{max}=2|A|$. Furthermore, let $d \in [2, 3, \dots , d_\textit{max}]$ be the current depth of the candidate tree traversal. We call $s \in \mathbb{N}\backslash \{0\}$ the \emph{steepness modifier}.

\medskip
Consider a place $p = \inoutpair{I}{O}$.
We define the \emph{sigmoid delta adaption function} as follows:
\begin{align*}
\texttt{adapt}^{s, d}_\textit{sigmoid}(\delta, \inoutpair{I}{O}) &= \delta \cdot \textit{mod}^{s, d}_\textit{sigmoid}(\inoutpair{I}{O})\\
&=\delta \cdot \left( \frac{2}{1+e^{\left( (-1) \cdot \frac{s}{(|I|+|O|)} \cdot (d-(|I|+|O|)) \right)}}-1 \right)
\end{align*}

The adaption function multiplies a modifier with the parameter $\delta$. Figure~\ref{fig:linearsigmoid} illustrates the behavior of this modifier for three example places of varying complexity. The \emph {sigmoid} delta adaption is designed to prefer simple places. When a place originates from the currently traversed level of the \emph{complete candidate tree}, i.e., it is among the most complex places currently available, the function will evaluate to $0$, meaning that only a perfectly fitting place can be added. The simpler the evaluated place is compared to the current tree level, the larger the result of the function and the more unfitting traces are allowed, with $\delta$ marking the maximal returnable value. The modifier grows fast in the beginning, but stagnates towards the end, preferring the simpler places more strongly, while the more complex places are (roughly) equally undesirable. The steepness modifier $s$ controls the intensity of the growth.

\section{Experimentation and evaluation} \label{sec:fm:eval}
We performed several experiments where we run the proposed algorithm with a wide variation of combinations of possible parameter settings on several event logs with different properties. The focus of this paper is on avoiding deadlocks and dead parts in the returned models, which is why the main focus of this evaluation is on the quality of the discovered models. In the end of this section we will briefly discuss the performance, in particular showing that the extension proposed in this work does not add significantly to the running time of the eST-Miner.

 \begin{table}[!b]
 \vspace*{-1mm}
\caption{{List of logs used for the evaluation. The upper part lists real-life logs while the lower part shows artificial logs. Logs are referred to by their abbreviations.\vspace*{-1mm}}}
\centering
\scalebox{0.9}{
 \begin{tabular}{l|l|r|r|l}
\textbf{\small {Log Name}} &\textbf{\small {Abbreviation}} &	\textbf{\small {Activities}} &	\textbf{\small {Trace Variants}} &	\textbf{\small {Reference}}\\
\hline
Sepsis &	\algorithmname{Sepsis} &		16 & 846 &~\cite{log:sepsis}\\
Road Traffic Fine Management &	\algorithmname{RTFM}	& 11	& 231 &~\cite{log:rtfm}\\
\hline
Teleclaims &	\algorithmname{Teleclaims}	& 	11 & 12  &~\cite{PMbook}\\ 
Order-Handling & \algorithmname{Orders} & 8 & 9 &~\cite{log:oderhandling} \\ 
\end{tabular} }
\label{tab:logs}

\vspace*{5mm}
\caption{Overview of the parameter settings used in our experimentation. The combinations result in $6300$ runs for each event log. The value ranges were chosen based on a smaller set of preliminary experiments, aiming to investigate a wide range of parameter settings on the one hand, while on the other hand avoiding unnecessary complexity resulting from variation without notable impact. For example, for our inputs no places were discarded for $\Qlength \geq 10000$. For $d^+$ we chose a very low and a very high value to evaluate whether it had any impact at all. Finally, for the chosen event logs $d_\textit{cut} = 5$ has shown to be sufficient to find complex structures with the standard eST-Miner, i.e., increasing the traversed tree depth increases computation time but has no strong impact on model quality.\vspace*{-1mm}}
\centering
\scalebox{0.9}{ \begin{tabular}{p{1.6cm}|p{2.9cm}|p{7.35cm}}
\textbf{Parameter} &\textbf{Used Values} &	\textbf{Purpose}\\
\hline
$\tau$ &	$0.3$, $0.4$, $0.5$, $0.6$, $0.7$, $0.8$, $0.9$ &	\small{Defines the minimal fraction of log traces that every place, as well as the final Petri net, must be able to replay.}\\
\hline
$\delta$ &	$0.05$, $0.1$, $0.15$, $0.2$, $0.25$ &	\small{Used to define the allowed reduction in log traces replayable by $N$ when adding a place.}\\
\hline
\fitnessmetric &	$\fitnessmetric_\textit{rel}, \fitnessmetric_\textit{comb}$ &	\small{Defines the fitness metric to be used.}\\
\hline
$\texttt{adapt}$ &	$\texttt{adapt}_\textit{noDelta}$, $\texttt{adapt}_\textit{constant}$, $\texttt{adapt}_\textit{sigmoid}$ & \small{The delta adaption function used to guide the heuristics.}\\
\hline
$s$ &	$1$, $2$, $3$, $4$, $5$ &	\small{The steepness of the increase of the adaption function (relevant for $\texttt{adapt}_\textit{sigmoid}$ only).}\\
\hline
$\Qlength$ &	$100$, $1000$, $10000$ &	\small{The maximal number of places stored in $Q$.}\\
\hline
$d^+$ &	$0$, $10$ &	\small{Artificial tree depth to re-evaluate places in $Q$ after end of tree traversal (relevant for $\texttt{adapt}_\textit{sigmoid}$ only).}\\
\hline
$d_\textit{cut}$ &	$5$ &	\small{Stop candidate traversal after the specified tree level.}\\
\end{tabular} }
\label{tab:params}
\end{table}

\subsection{Experimental setup}
Table~\ref{tab:logs} provides an overview of the event logs used in our experimentation. \texttt{Sepsis} has a relatively high number of different trace variants, all of which have comparable frequencies, with the most frequent trace variant making up only $3.33$~\% of the event log. Activities are repeated often within a trace, which must lead to looping behavior within a Petri net with uniquely labeled transitions.  \texttt{RTFM} is rather large, with a moderate variety of trace variants and activities. Both for variants and activities some are very frequent while others are quite infrequent. \texttt{Teleclaims} is an established artificial log useful for testing discovery of various control-flow structures. With \texttt{Orders} we can demonstrate the algorithm's ability to discover complex control flow structures, as well as the option to abstract from rare behavior.

\medskip
For each event log we perform $6300$ runs of the algorithm with varying combinations of the different parameters, as specified in Table~\ref{tab:params}. Note that we keep the order of place candidate traversal fixed for all runs.

To investigate the qualitative impact of the proposed heuristics, we need to fix the order of candidate evaluation to prevent effects due to different evaluation orders. For easy reproducibility, we use a lexicographical ordering  based on activity names.
The purpose is to focus on the effect of the different parameters, and possibly derive which of them are the most relevant for the discovery of certain models and whether certain (combinations of) settings are preferable.

\medskip
Our experiments with combined fitness have shown that only for the \texttt{Sepsis} log there are place candidates that score lower on relative fitness than aggregate fitness. However, this happened for very few parameter combinations and then only for at most $2$ candidate places. Thus, for the logs evaluated, we can conclude that the choice between using $\fitnessmetric_\textit{agg}$ and $\fitnessmetric_\textit{comb}$ does not have a significant impact.

\subsubsection*{Evaluation Metrics}

Several approaches exist to measure model quality with respect to an event log. In the following, we give a brief overview of the techniques applied in the context of this evaluation.

\medskip
We use alignment-based fitness to measure how well the behavior in the log is represented by the model. Alignments take an event log and compute the minimal number of insertions and deletions needed to make the traces fitting with respect to a given model, then normalize this value using the worst-case edit distance.
Consider the example Petri net in Figure~\ref{fig:qmetricsexample} and the trace \Trace{\sactivity, a, b, c,\eactivity}. This trace can be aligned to the Petri net by, for example, removing $b$ and $c$ or by removing $a$. Removing $a$ needs $1$ edit operation, which is the optimum in this case. The worst-case edit distance would require us to remove every activity from the trace ($5$) and insert an activity for every transition that needs to be fired to obtain the shortest path through the model ($3$). This results in an alignment-based fitness of $1-\frac{1}{5+3}$ for the example trace. For details, we refer the reader to \cite{alignments}.
\begin{figure}[hbt]
    \centering  
  \includegraphics[width=0.4\textwidth, trim={1.2cm 3.7cm 4.75cm 4cm}, clip]{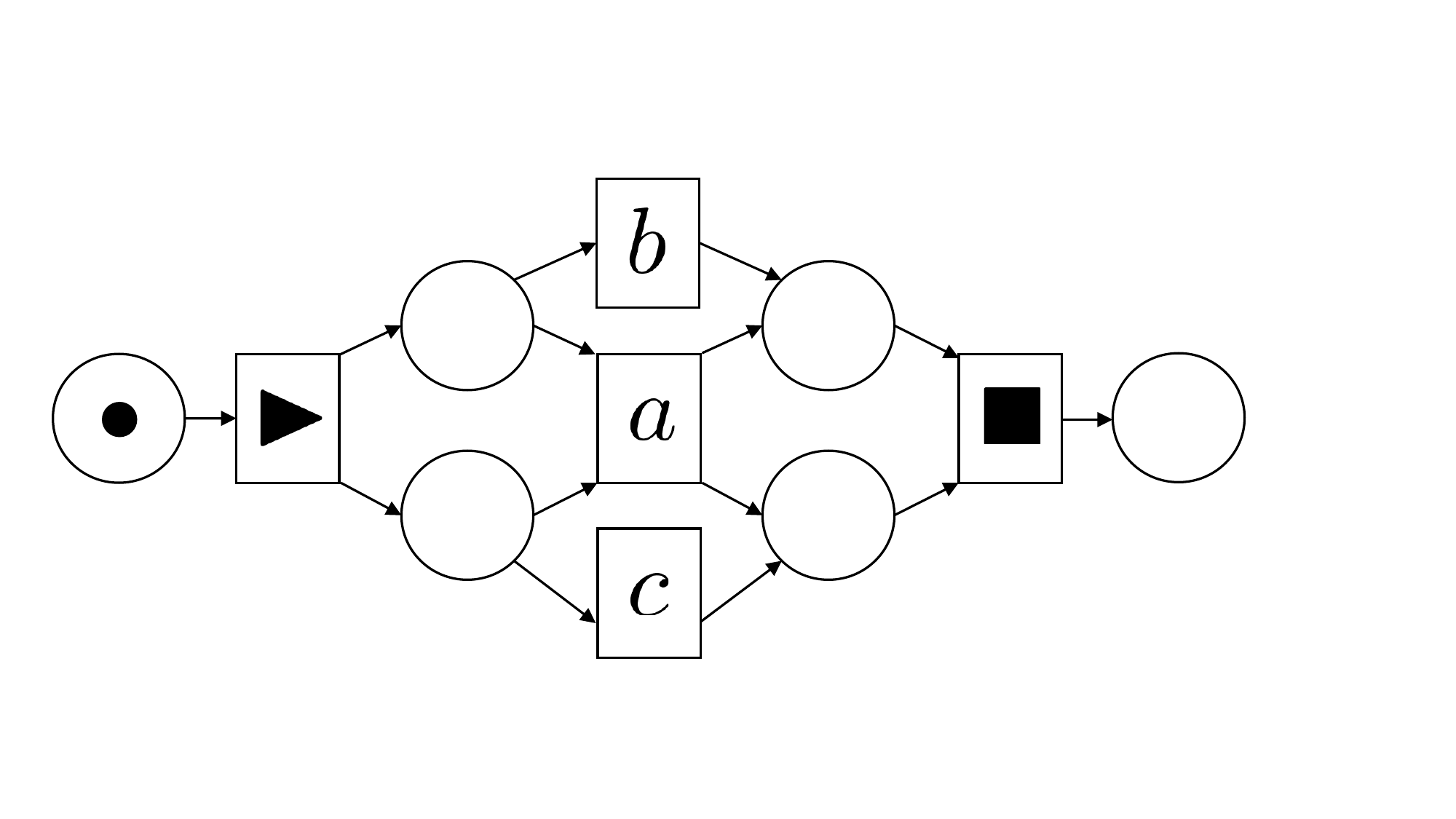}
    \caption{An example model used to illustrate the applied quality metrics.}
    \label{fig:qmetricsexample}
\end{figure}

\noindent Evaluating whether the model is sufficiently restrictive, i.e., does not allow for (much) behavior outside of the process behavior is more challenging. The event log represents a set of example process executions and cannot be expected to be complete, making it hard to reliably evaluate process-related precision. Another problem is the potentially infinite language of the model which can make a straightforward comparison infeasible. Therefore, metrics that approximate precision are commonly used. We choose a precision metric based on escaping edges. This metric compares the number of activities enabled in the model with the number of activities actually executed during the replay of the event log.
Consider the example Petri net in Figure~\ref{fig:qmetricsexample} and the traces \Trace{\sactivity, a,\eactivity} and \Trace{\sactivity, c,b,\eactivity}. We replay all prefixes of traces in the event log and take note of the frequency of that prefix, the number of enabled transitions at the end of the prefix and the number of transitions fired after the prefix. In our example, for both traces in the beginning only $\sactivity$ is enabled and consequently fired ($2 \cdot 1$ enabled, $2 \cdot 1$ fired). This enables three activities $a, b$ and $c$, of which we fire $a$ in the first trace and $c$ in the second trace ($2 \cdot 3$ enabled, $2 \cdot 2$ executed). Now, for the first trace \eactivity is the only one executed and fired ($1 \cdot 1$ enabled, $1 \cdot 1$ fired). For the second trace, after firing $c$ only $b$ is enabled and fired ($1 \cdot 1$ enabled, $1 \cdot 1$ fired) followed by the final activity, \eactivity, being enabled and fired ($1 \cdot 1$ enabled, $1 \cdot 1$ fired). This results in a precision of $\frac{2 \cdot 1+2\cdot 2+1\cdot 1+1 \cdot 1+1 \cdot 1}{2 \cdot 1+2 \cdot 3+1 \cdot 1+1 \cdot 1+1 \cdot 1}=\frac{9}{11}$, i.e., the number of transitions actually fired during replay divided by the number of enabled transitions. For details on escaping edges precision we refer the reader to \cite{ETConformance}.

\medskip
Recall, that for most real-life logs not all quality aspects can be perfectly satisfied at the same time (e.g., high fitness often entails low precision and the other way around). They should be balanced according to the user's needs, i.e., the choice of the best model depends on its purpose.
In our evaluation, we consider two target audiences. Some users prefer to directly apply a discovery algorithm to the unmodified log, expecting the algorithm to filter infrequent activities automatically and return a model that focuses on the main behavior in the event log. In the absence of a clear use-case, it is common to score process models based on the harmonic mean of fitness and precision ($F_1$-Score), which ensures that the resulting aggregated quality score reflects both metrics, i.e., a low value in one of the metrics cannot be obscured by a high value in the other. Without further information, we assume that such a user to be interested in obtaining a model with a high $F_1$-Score.

\begin{definition}[$F_1$-Score]
    Given an event log $L \in \multisetset{\universeoftraces}$ and a Petri net $N$ with a fitness score of $\texttt{fitness}_L(N)$ and a precision score of $\texttt{precision}_L(N)$, we define the $F_1$-score as

           $$F_1(L,N)=  \left\{ \begin{matrix*}[l]
    		0, \; \;  \textit{ if } \texttt{fitness}_L(N) = 0 \vee \texttt{precision}_L(N) = 0\\[0.5em]
  		 \cfrac{2}{\cfrac{1}{\texttt{fitness}_L(N)}+\cfrac{1}{\texttt{precision}_L(N)}}, \; \; \textit{ otherwise} \end{matrix*} \color{white} \right)$$
\end{definition}

\noindent
In general, users applying our algorithm have the option to perform some basic preprocessing beforehand. In particular, they may use available functionality \cite{Prom} to remove infrequent activities and infrequent trace variants they are not interested in. In such cases, all infrequent activities and traces remaining in the event log can be considered to be of interest and should be reflected in the discovered model, while still abstracting from exceptional behavior patters. Without further information, we assume that such a user is interested in obtaining a model that scores high in fitness and precision but ideally also contains all activities from the event log. Therefore, we introduce the the metric of \emph{activity-coverage} and define the \texttt{HM}-Score as the harmonic mean of fitness, precision and activity-coverage.

\begin{definition}[Activity-Coverage]
    Let $L \in \multisetset{\universeoftraces}$ be an event log over the set of activities $A$ and let $N=(A', P)$ be a Petri net with $A' \subseteq A$, {i.e., $A'$ is the subset of log activities included in the Petri net $N$.} We define \emph{activity-coverage} of $N$ with respect to $L$ as

   $$\texttt{activity-coverage}_L(N)=\frac{|A'|}{|A|}$$
\end{definition}
A value of $1$ indicates that all activities in the event log are also part of the Petri net, while a value of $0$ indicates that no activity in the event log is part of the Petri net. With respect to the event log $L=[\Trace{\sactivity, a, \eactivity}, \Trace{\sactivity, c,b, \eactivity}, \Trace{\sactivity, d, e, d, e, a,\eactivity}]$, the example Petri net in Figure~\ref{fig:qmetricsexample} achieves an activity-coverage of $\frac{3}{5}.$

\begin{definition}[HM-Score]
    Given an event log $L \in \multisetset{\universeoftraces}$ and a Petri net $N$ with a fitness score of $\texttt{fitness}_L(N)$ and a precision score of $\texttt{precision}_L(N)$, we define the \texttt{HM}-score as\vspace*{-2mm}

           $$\texttt{HM}(L,N)= \left\{ \begin{matrix*}[l]
    		0, \; \; \textit{ if } \texttt{fitness}_L(N) = 0 \vee \texttt{precision}_L(N) = 0 \vee \texttt{activity-coverage}=0\\[0.5em]
  		\cfrac{3}{\cfrac{1}{\texttt{fitness}_L(N)}+\cfrac{1}{\texttt{precision}_L(N)}+\cfrac{1}{\texttt{activity-coverage}_L(N)}}, \; \; \textit{ otherwise} \end{matrix*} \color{white} \right)\vspace*{1mm}$$
\end{definition}
The simplicity of a model is highly subjective and a variety of factors may contribute to the readability of a model. Not all of these can be easily represented by numbers, e.g., the way a Petri net is plotted. Even though several metrics have been suggested for it, they are restricted to only some aspects contributing to model understandability and no generally agreed upon solution has become the standard yet. Therefore, we forgo an extensive evaluation of simplicity and focus on a straightforward metric based on the average number of arcs per transition. Not only is this metric closely related to the strategy of the presented approach to prefer places that have few arcs, but also existing research confirms the general relevance of this aspect \cite{simplicity}.
\begin{definition}[Simplicity] \label{def:simplicity}
Given a Petri net $N=(A, P)$, we define simplicity as the average number of arcs per transitions, i.e.,

$$\texttt{simplicity}(N)=\frac{\sum_{\inoutpair{I}{O} \in P} |I|+|O|}{|A|}.$$
\end{definition}
The example Petri net in Figure~\ref{fig:qmetricsexample} achieves a simplicity of $\frac{14}{5}\approx 2.8$. On the downside, this measure of simplicity does not map to the interval between $0$ and $1$ and is therefore not directly comparable with the other quality metrics (and thus not included in the aggregated score). On the upside, it returns an objective value that allows for further subjective interpretation as the reader sees fit.

\subsection{Qualitative analysis}
Some interdependencies between the model quality aspects are to be expected and confirmed by our results. Removing a transition from a Petri net reduces the behavior of the net and therefore has a negative effect on fitness and a positive effect on precision. It is important to keep in mind the exact metrics used to measure fitness and precision to avoid misinterpretation of the results: while alignments are quite forgiving with respect to missing infrequent transitions (they simply assign a penalty whenever the corresponding activity occurs in the event log), escaping edges based precision is sensitive with respect to transitions that are frequently enabled without being fired (e.g. in the case of parallelism).

\medskip
One of our major goals in this work is to be able to avoid the removal of transitions based on infrequency. By achieving this goal, we obtain models that contain transitions which are far more often enabled than fired. Such models score significantly worse with respect to precision than models without those infrequent transitions, while fitness remains comparable. Therefore, in addition to evaluating models using the \texttt{HM}-Score, we focus on models with perfect activity-coverage by evaluating them separately from the complete set of results. Finally, we include some representative models to support the interpretation of the number-based evaluation. However, one needs to keep in mind that the choice of the best model always depends on the user's needs and models scoring high in the context of this general evaluation do not necessarily  represent the best model for every application.

\subsubsection*{Overview of Quality Results}

In Figure~\ref{fig:resultsComplete}, an overview of the quality results of the $6300$ models generated for each log is given. Fitness and simplicity remain rather stable, with  fitness being generally high and simplicity values clustering between $2$ and $3$ arcs per transition on average, which we consider a good value. On the other hand, precision and activity-coverage, and by extension the harmonic means \texttt{HM} and $F_1$, vary a lot for the discovered models. This clearly indicates that the choice of parameters has a strong impact on these quality aspects.

\begin{figure}[!ht]
\vspace*{-1mm}
    \centering 
    \includegraphics[width=0.45\linewidth]{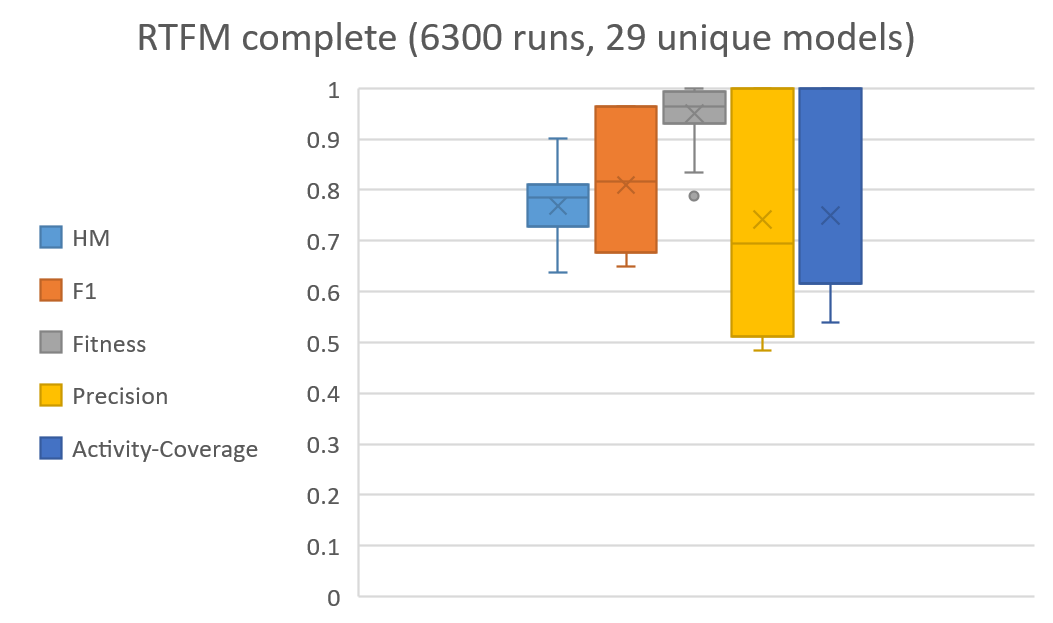}
    \includegraphics[width=0.45\linewidth]{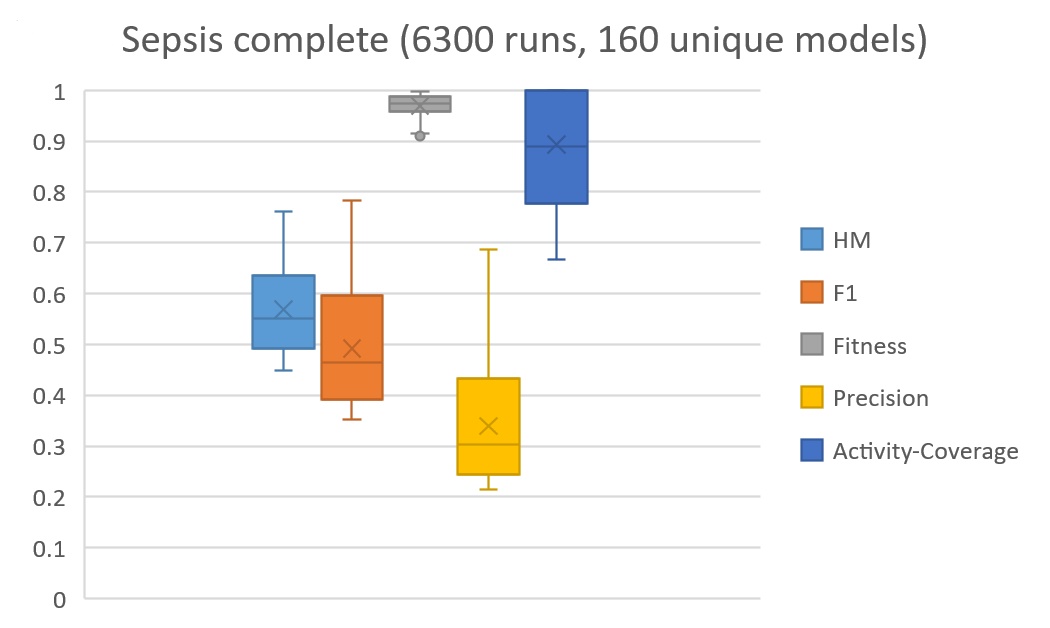}\\
    \includegraphics[width=0.45\linewidth]{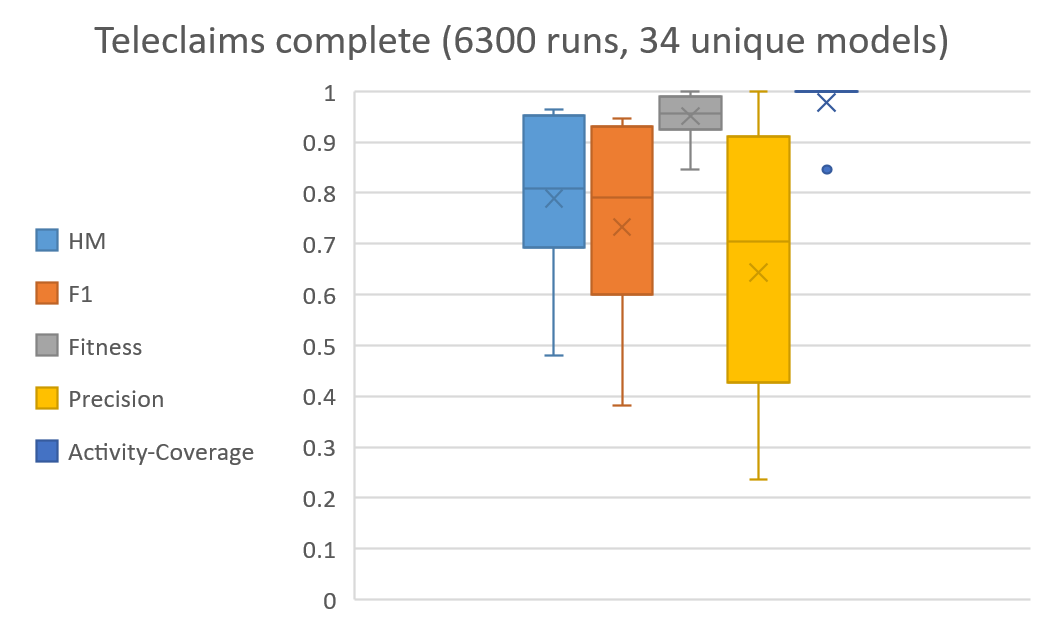}
    \includegraphics[width=0.45\linewidth]{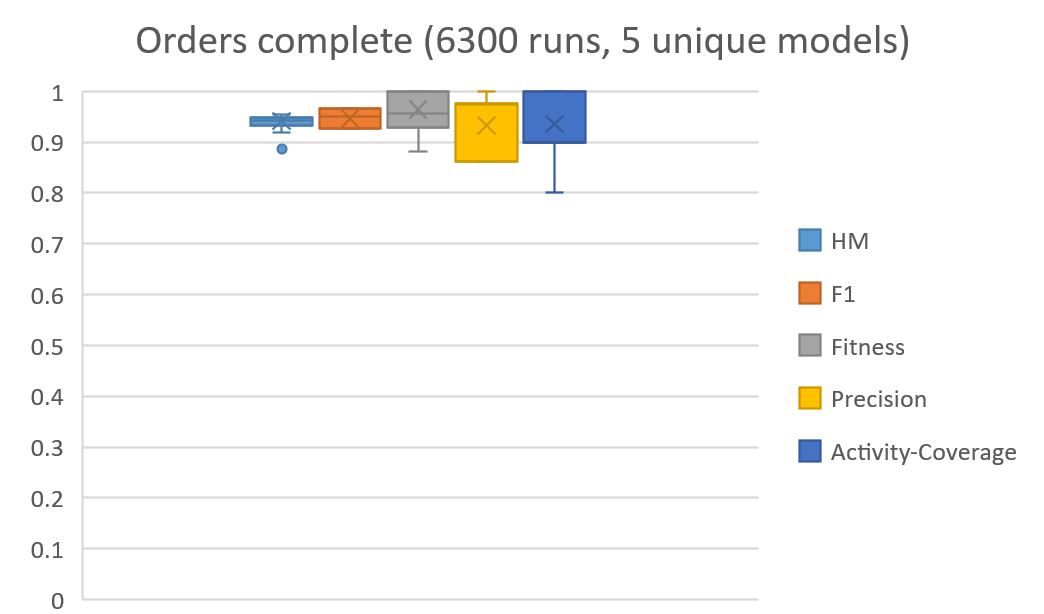}\\
    \centering
    \includegraphics[width=0.65\linewidth]{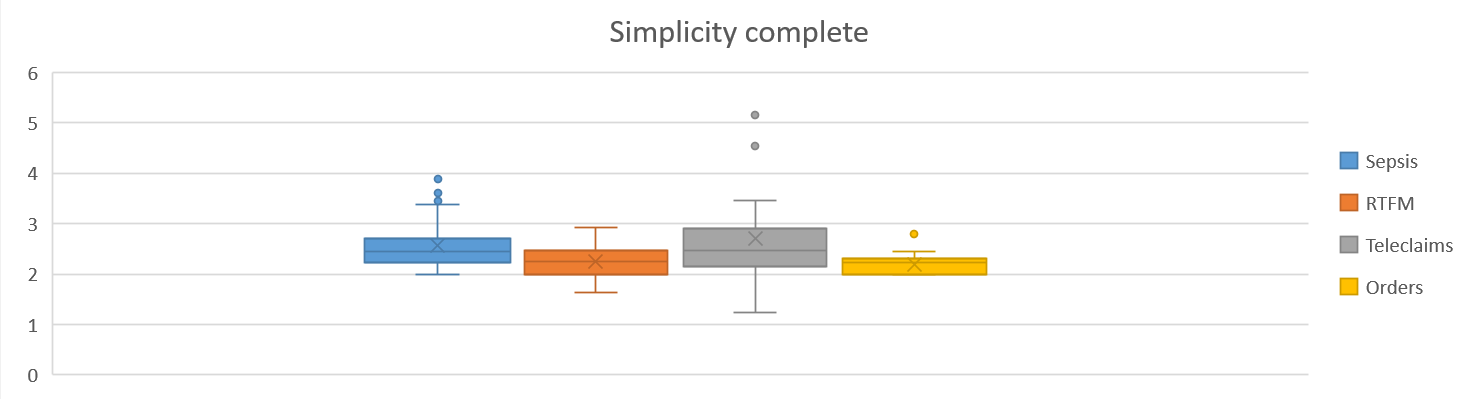}\vspace*{-1mm}
    \caption{Overview of model quality results for all $6300$ runs (complete) with varying parameters but fixed candidate traversal order.}
    \label{fig:resultsComplete}\vspace*{1mm}
\end{figure}

\begin{table}[!h]
\centering
\caption{{Overview of the qualitative results of selected models discovered during our experimentation. For each log we show the quality results of the model with the maximal \texttt{HM} value, the maximal $F_1$ value, and the model with the maximal \texttt{HM} out of the subset of models with perfect activity-coverage (aCov), as well as the frequency with which this model was discovered. Additionally, we provide scores for the Inductive Miner infrequent (IMf, default settings) and the eST-Miner with $\tau=1.0$. Green background marks comparatively large values. All of these models are also visualized in Figures}~\ref{fig:Ordersmodels} to \ref{fig:teleclaimsmodels} at the end of the section.}
\includegraphics[width =1.0\linewidth, trim={0.64cm 13cm 0.65cm 1.9cm},clip]{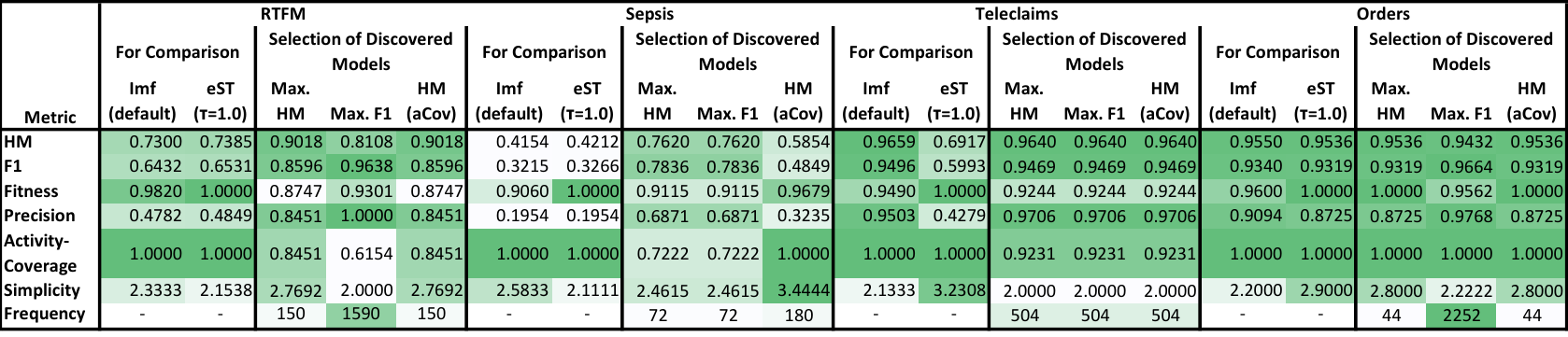} 
\label{tab:fm:modelsOverview}\vspace*{-2mm}
\end{table}


\medskip
While we discovered only $5$ unique models, i.e., models with different behavior, for \texttt{Orders}, there were $29$ unique models found for \texttt{RTFM}, $34$ for \texttt{Teleclaims} and $160$ for \texttt{Sepsis}. The quality results and frequencies of a selected subset of the discovered Petri nets are given in Table~\ref{tab:fm:modelsOverview}. Additionally, we provide the same results for the models discovered by the Inductive Miner infrequent (IMf) with default settings as implemented in ProM~\cite{Prom} and the models discovered by the eST-Miner with $\tau=1.0$ (comparable to region theory results). Our approach can discover models with \texttt{HM} and $F_1$ scores that clearly outperform IMf with default settings as well as eST-Miner with $\tau=1.0$ on the two real-life event logs. For the two artificial event logs results are comparable. A detailed comparison of these models follows at the end of this section.

\subsubsection*{Discussion of Dead Transitions}
Since we are interested in abstracting from infrequent behavioral patterns without outright removal of infrequent activities or complete trace variants, we take a closer look at the discovery of models that include all activities from the event log. During our experimentation, we discovered $2$ different such models for \texttt{Orders}, $31$ for \texttt{Teleclaims}, $29$ for \texttt{Sepsis} and $9$ for \texttt{RTFM}.

\begin{figure}[!h]
 \vspace{2mm}
    \centering
    \includegraphics[width=0.45\linewidth]{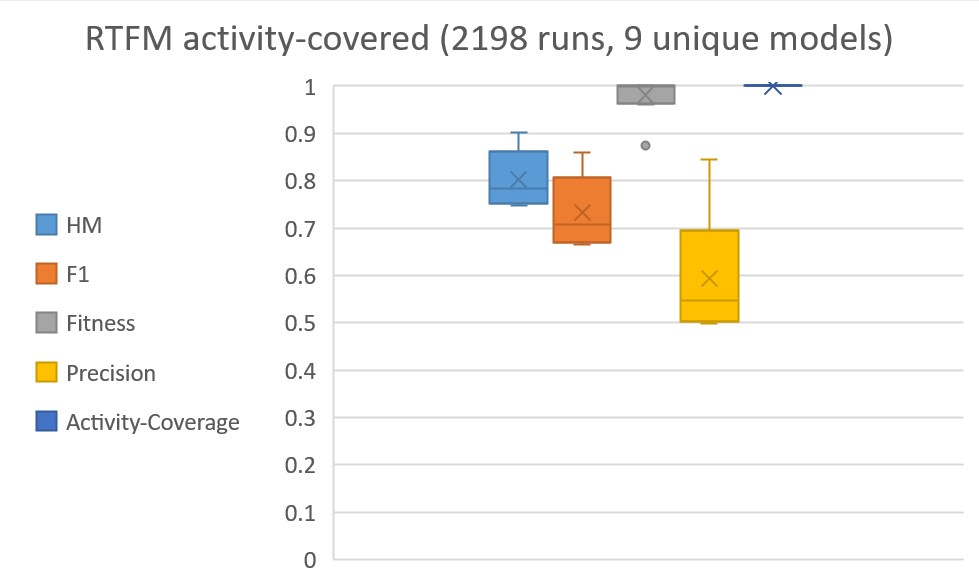}
    \includegraphics[width=0.45\linewidth]{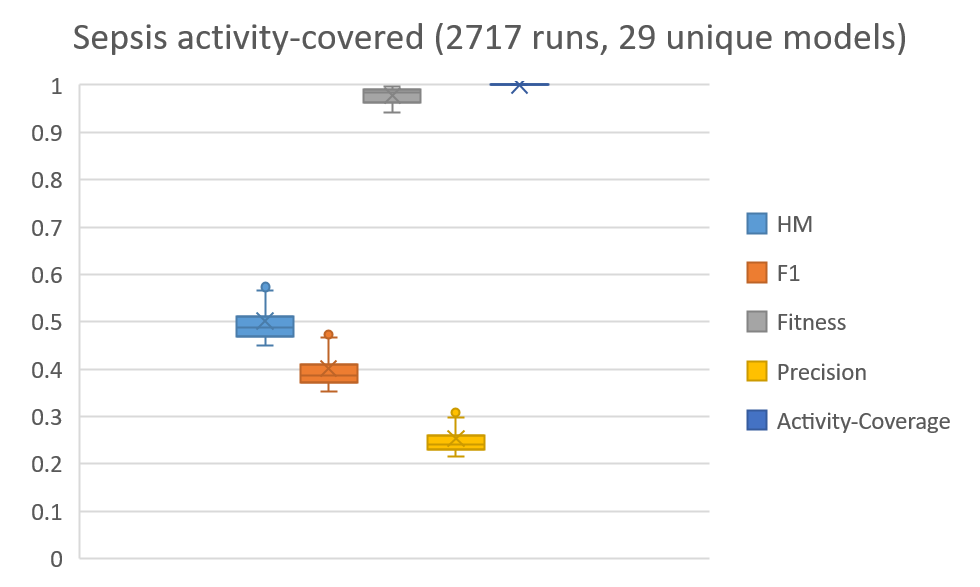}\\
    \includegraphics[width=0.45\linewidth]{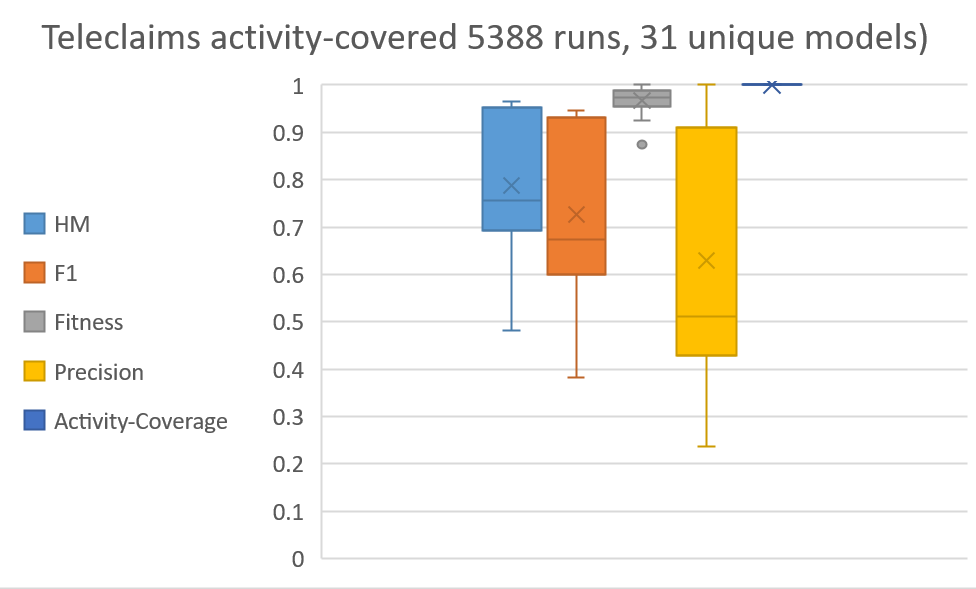}
    \includegraphics[width=0.45\linewidth]{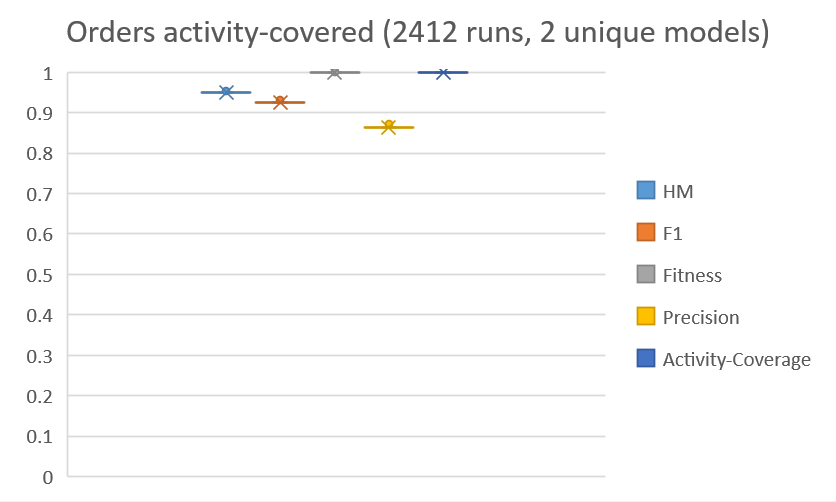}\\
    \includegraphics[width=0.65\linewidth]{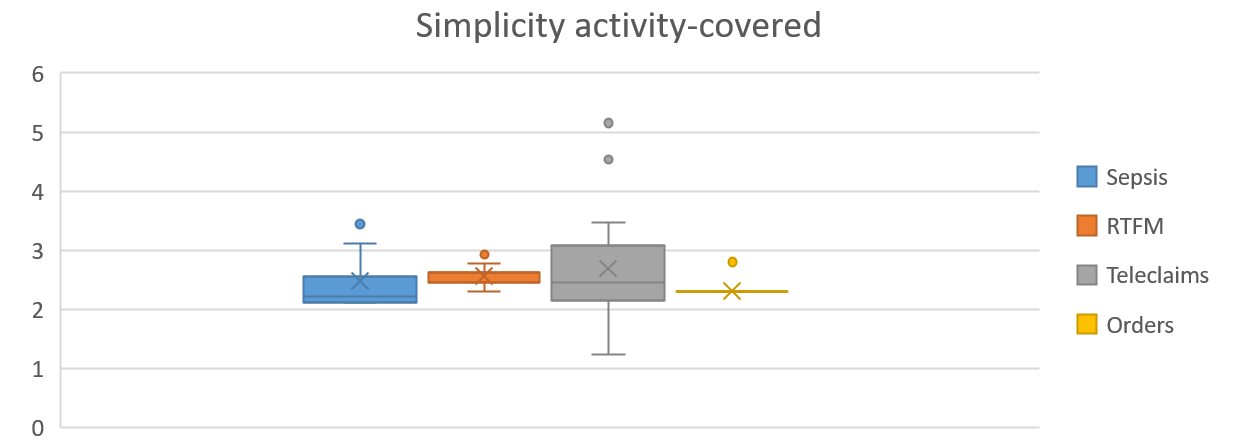}
    \caption{Overview of quality results for all models discovered in the $6300$ runs that include all activities observed in the corresponding event log. }
    \label{fig:resultsAlive}
\end{figure}

In Figure~\ref{fig:resultsAlive}, we show an overview of the quality metrics restricted to the runs that resulted in models that include all activities. The general tendencies remain similar to the results shown for the models discovered for the complete set of runs. Fitness is consistently high, with less lower scoring outliers than we have seen for the complete set of models. This is to be expected, since these models do contain all log activities and our proposed approach guarantees that they can all be fired at least once. Consequently, precision is comparable for the artificial logs which do not include a lot of noise or diverse behavior, or lower for the real-life logs, which exhibit a significantly higher variance in behavior. Notably, the number of models and variance in precision have decreased for \texttt{Sepsis} and \texttt{Orders} - apparently, the parameter combinations allowing for the discovery of models without dead transitions do result in models scoring similarly in quality.

\begin{figure}[!ht]
    \centering
    \includegraphics[width=0.45\linewidth]{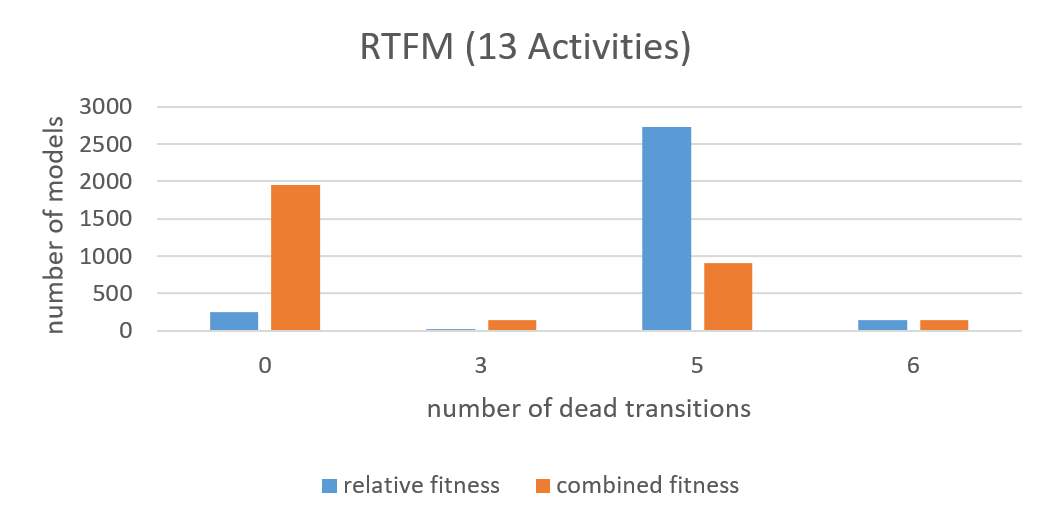}
    \includegraphics[width=0.45\linewidth]{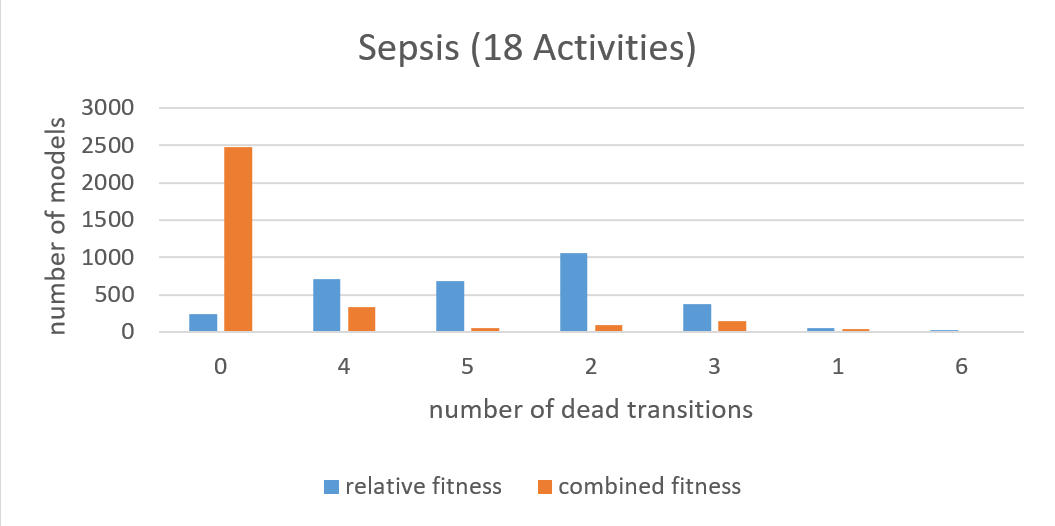}\\
    \includegraphics[width=0.45\linewidth]{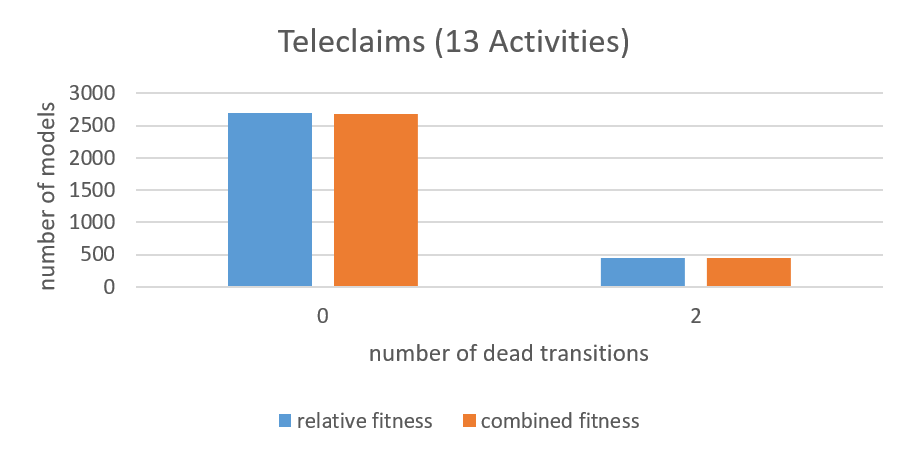}
    \includegraphics[width=0.45\linewidth]{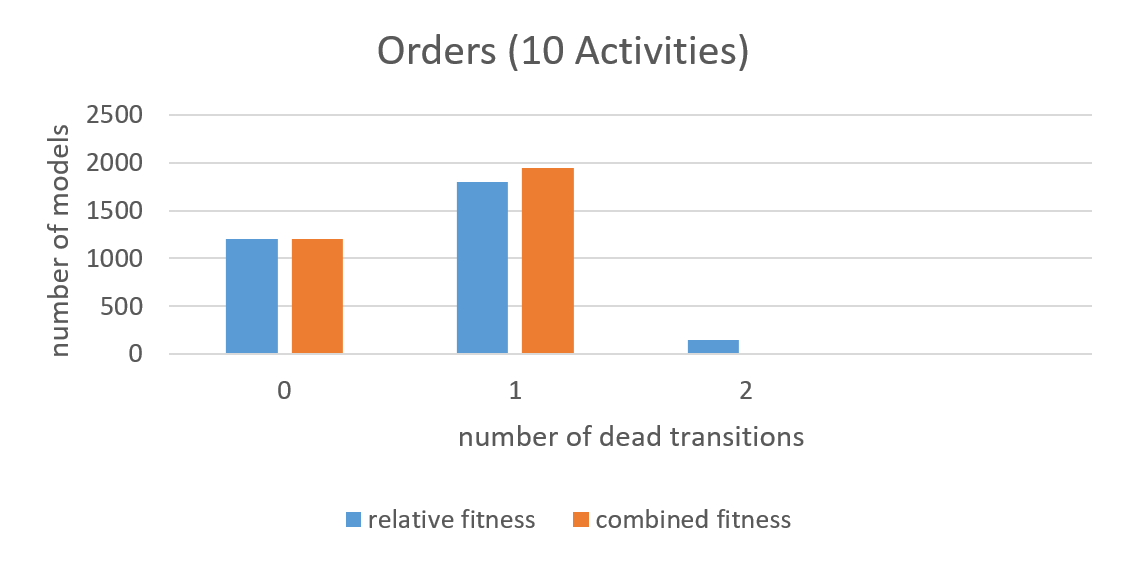}  \vspace*{-1mm}
\caption{Comparison of the number of dead transitions of the models discovered using \emph{relative fitness} and using \emph{combined fitness}
       with the proposed algorithm.}
    \label{fig:comvsrel}
\end{figure}

\begin{table}[!b]
\vspace*{-2mm}
\centering
{\scriptsize
\caption{Out of the 6300 runs we performed in our experiments, $3150$ were performed using \emph{combined fitness} and $3150$ using \emph{relative fitness}. For each event log this table gives an overview about the number of runs that resulted in models without dead transition, as well as how many different such models were discovered.}
\label{tab:aliveModelsRuns}
\begin{tabular}{@{}c|cc|cc@{}}
\toprule
 &
  \multicolumn{2}{c|}{\textbf{\begin{tabular}[c]{@{}c@{}}combined fitness\\ (3150 runs)\end{tabular}}} &
  \multicolumn{2}{c}{\textbf{\begin{tabular}[c]{@{}c@{}}relative fitness\\ (3150 runs)\end{tabular}}} \\
                    & \textbf{\#runs} & \textbf{\#unique models} & \textbf{\#runs} & \textbf{\#unique models} \\ \midrule
\textbf{\texttt{RTFM}}       & 1946            & 8                        & 252             & 1                        \\
\textbf{\texttt{Sepsis}}     & 2479            & 26                       & 238             & 3                        \\
\textbf{\texttt{Teleclaims}} & 2688            & 18                       & 2700            & 13                       \\
\textbf{\texttt{Orders}}     & 1206            & 1                        & 1206              & 1                        \\ \bottomrule
\end{tabular}
}
\end{table}

In Section~\ref{sec:fm:fitness}, we introduced \emph{combined fitness} with the goal of preventing the uncontrolled deletion of infrequent activities due to their execution being blocked accidentally. In Figure~\ref{fig:comvsrel}, we compare the impact of using combined fitness and of using relative fitness on the number of dead transitions in the discovered models. In particular for the two real-life event logs, \texttt{Sepsis} and \texttt{RTFM}, a clear effect is visible: when using combined fitness is it much more likely to discover models which include all or most of the observed activities.
This observation is confirmed by the data in Table~\ref{tab:aliveModelsRuns} where we compare the number of runs that result in models without dead transitions for the runs using combined fitness and the runs using relative fitness. While the trend is clearly visible for the two real-life event logs, for \texttt{Teleclaims} and \texttt{Orders} the majority of parameter combinations result in models without dead transitions, independently of the fitness metric chosen. Similar tendencies can be observed for the number of different models without dead transitions; especially for the two real-life logs, most such models are discovered using combined fitness.

\medskip
Another interesting aspect can be observed for the \texttt{RTFM} and \texttt{Teleclaims} event logs: models with certain numbers of dead transitions ($1,2,4$ for \texttt{RTFM}, $1$ for \texttt{Teleclaims}) are never discovered. This can be explained by groups of transitions which are closely coupled in behavior and (almost) always occur together in the same traces. They are removed from the model together once these traces are no longer replayable. 

\begin{figure}[!b]
\footnotesize{Inductive Miner infrequent (default settings):}

\includegraphics[width =0.99\linewidth]{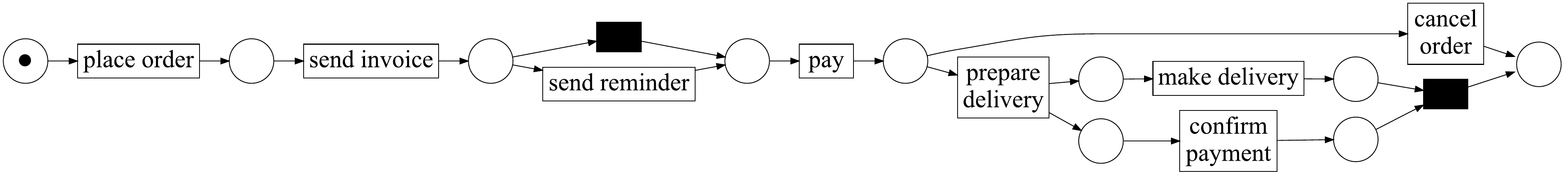}\hfill

\footnotesize{eST-Miner (${\tau=1.0}$):}

\includegraphics[width =0.99\linewidth]{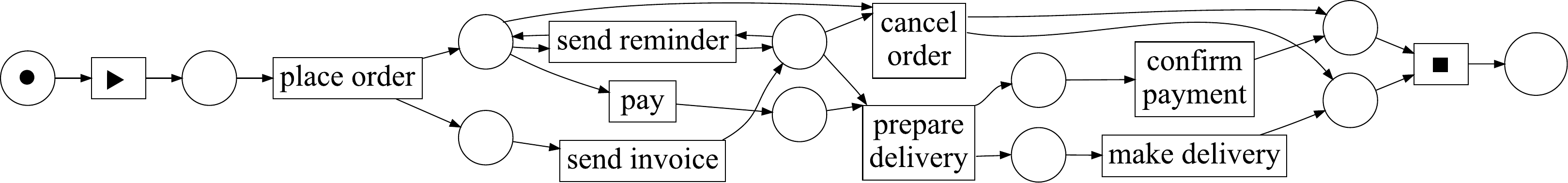}\hfill

\vspace{\figdist}
\footnotesize{Presented Approach: highest \texttt{HM} value out of the complete set of results, as well as the set of models without dead transitions. }

\includegraphics[width =0.99\linewidth]{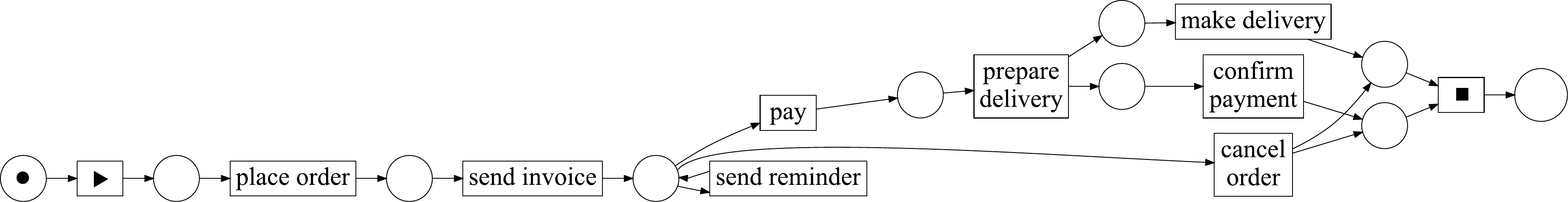}\hfill

\vspace{\figdist}
\footnotesize{Presented Approach: highest $F_1$-score out of the complete set of results. }

\vspace{\figdist}
\includegraphics[width =0.99\linewidth]{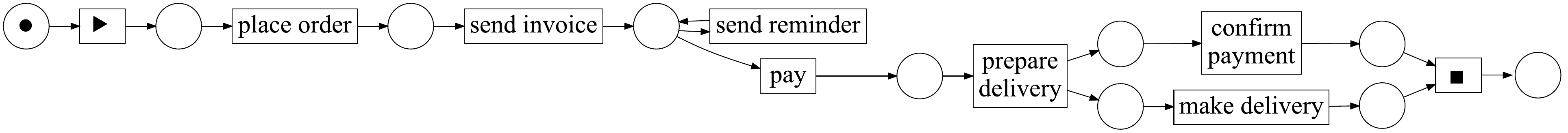}
\caption{{The Petri nets discovered based on the \texttt{Orders} log using the Inductive Miner infrequent (default settings), the eST-Miner with $\tau=1.0$, and a subset of interesting models discovered using the presented approach. }}
\label{fig:Ordersmodels}
\end{figure}

\begin{figure}[!htbp]
\footnotesize{Inductive Miner infrequent (default settings):}

\vspace{\figdist}
\includegraphics[width =0.95\linewidth]{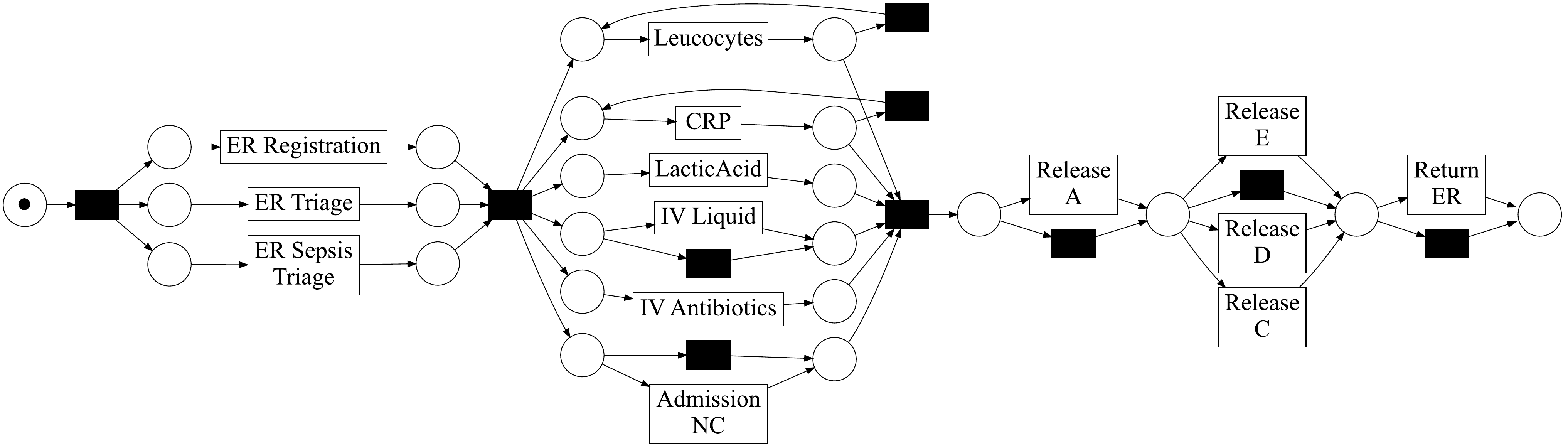}\hfill

\footnotesize{eST-Miner (${\tau=1.0}$):}

\vspace{\figdist}
\includegraphics[width =0.52\linewidth]{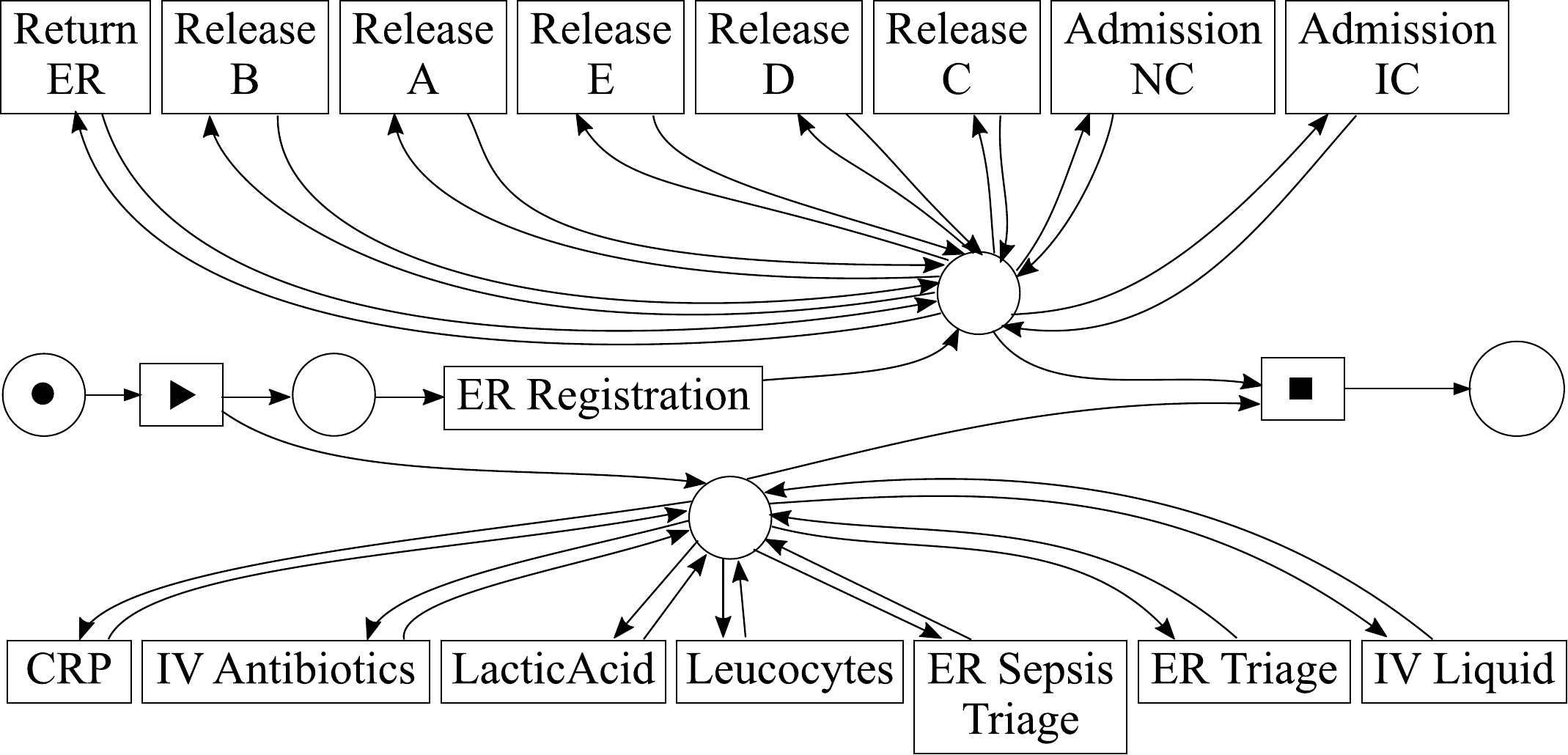}\hfill

\vspace{\figdist}
\footnotesize{Presented Approach: highest \texttt{HM} value and highest $F_1$-score out of the complete set of results.}

\vspace{\figdist}
\includegraphics[width =1.0\linewidth]{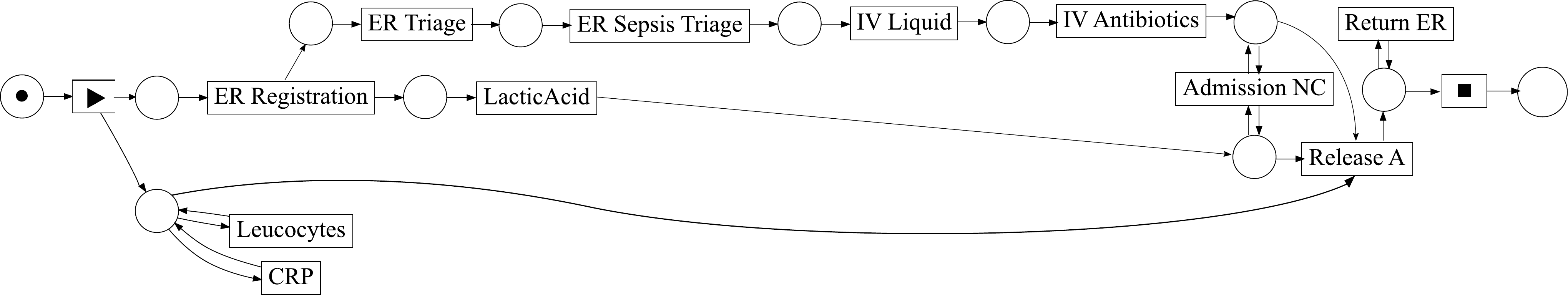}\hfill

\vspace{\figdist}
\footnotesize{Presented Approach: highest \texttt{HM} value out of the set of models without dead transitions.}

\vspace{\figdist}
\includegraphics[width =0.85\linewidth]{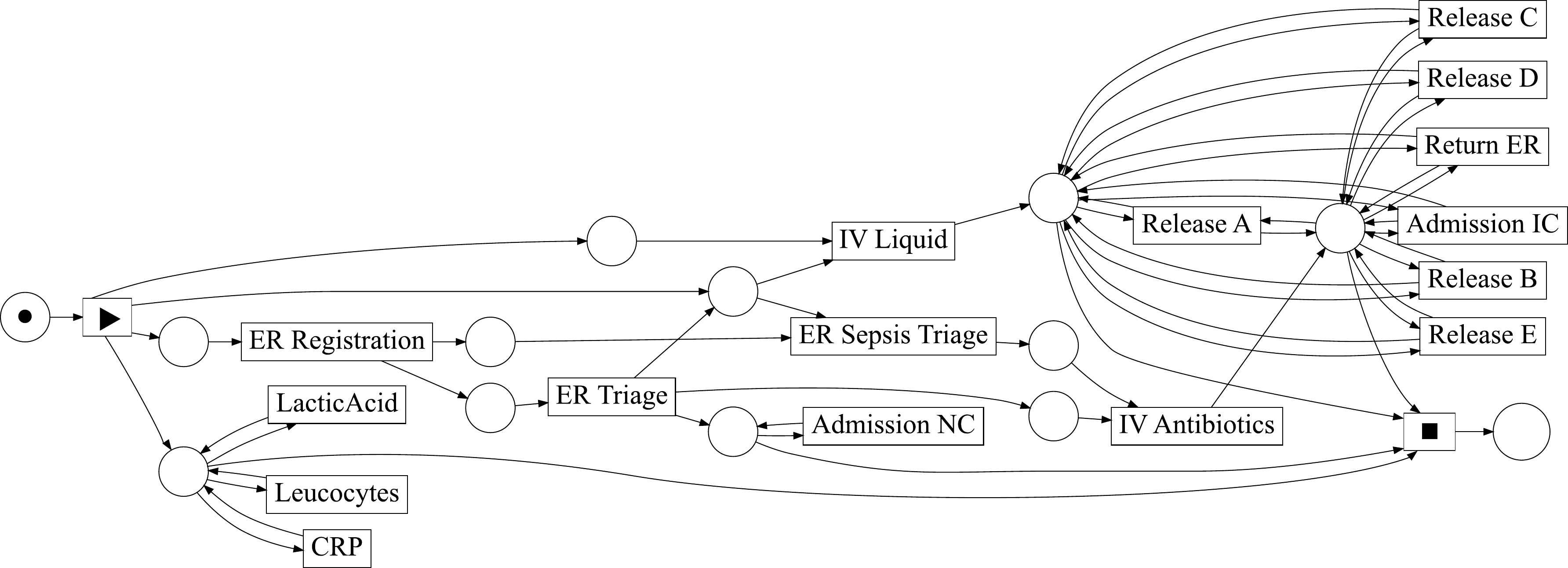}\hfill
\caption{{The Petri nets discovered based on the \texttt{Sepsis} log using the Inductive Miner infrequent (default settings), the eST-Miner with $\tau=1.0$, and a subset of interesting models discovered using the presented approach.}}
\label{fig:sepsismodels}
\end{figure}

\subsubsection*{Discussion of Selected Models}\label{subsec:models}
In Figures~\ref{fig:Ordersmodels},~\ref{fig:sepsismodels},~\ref{fig:rtfmmodels} and~\ref{fig:teleclaimsmodels} we present a selection of models for each log. For comparison, we present the models discovered by IMf (default settings) and the model discovered by the eST-Miner with $\tau=1.0$. From the many models discovered during the experimentation with our approach, we show the models scoring highest with respect to \texttt{HM} (the harmonic mean of fitness, precision, \mbox{activity-coverage}) as well as with respect to $F_1$ (harmonic mean of fitness and precision only). Furthermore, we present the models with the highest \texttt{HM} and $F_1$ values out of the set of models without dead transitions.  Of course, while these models score highest with respect to the general quality metrics used in this evaluation, other models may still be considered to be more suitable for specific contexts or applications.

\medskip
All models shown in Figure~\ref{fig:Ordersmodels} were discovered for the \texttt{Orders} log. For this rather simple event log, all models achieve relatively high scores with respect to the quality metrics. However, some notable differences in the expressed behavior can be observed in particular with respect to the activities \emph{send invoice, send reminder, pay} and \emph{cancel order}. According to the event log, in most cases the execution of \emph{send invoice} is eventually followed either by \emph{pay} (and then delivery) or by \emph{cancel order}, but never both. In rare cases, payment occurs before sending the invoice. After sending the invoice, reminders can be sent repeatedly, until payment is received or the order is canceled. This behavior is fully expressed only by the model discovered using the eST-Miner with $\tau=1.0$, which is comparable to results produced by region-based approaches. Since payment before sending the invoice is rare, users may prefer the other models which focus on behavior where payment arrives after sending the invoice. The model discovered by IMf further deviates from the log by not allowing for repeated reminders (occurring in 25 \% of the traces), and enabling the cancellation of orders after payment. {The  model with the highest \texttt{HM} value is the same for both, the complete set of models as well as the set without dead transitions. It includes all activities observed in the event log, in contrast to the model with the highest $F_1$-Score, which does not contain the activity \emph{cancel order} (occurring in $13.03$ \% of traces) at all, resulting in slightly lower fitness but increased precision.}

\medskip
The  \texttt{Sepsis} event log exhibits many repetitions of activities and a comparatively high control-flow variance, with $846$ trace variants in $1050$ traces, the most frequent of which occurs only $35$ times. Thus, the discovery of a model with simultaneously high fitness and precision is challenging. Figure~\ref{fig:sepsismodels} presents a selection of discovered Petri nets. The IMf manages to discover groups of activities that occur in sequence, however, within these groups the activities are in parallel and mostly skippable, resulting in a very low precision. The eST-Miner with $\tau=1.0$ illustrates a disadvantage of requiring perfect fitness: the resulting model allows for nearly all possible behaviors. For the model with highest \texttt{HM} value out of all models without dead transitions, this problem becomes less severe. Finally, the model with highest \texttt{HM} and $F_1$ values out of all discovered models manages to capture the main behavior hidden in the traces while ignoring infrequent activity behavior, achieving comparatively high precision at the cost of not representing all activities.

\begin{figure}[!h]
\footnotesize{Inductive Miner infrequent (default settings):}

\vspace{\figdist}
\includegraphics[width =1.0\linewidth]{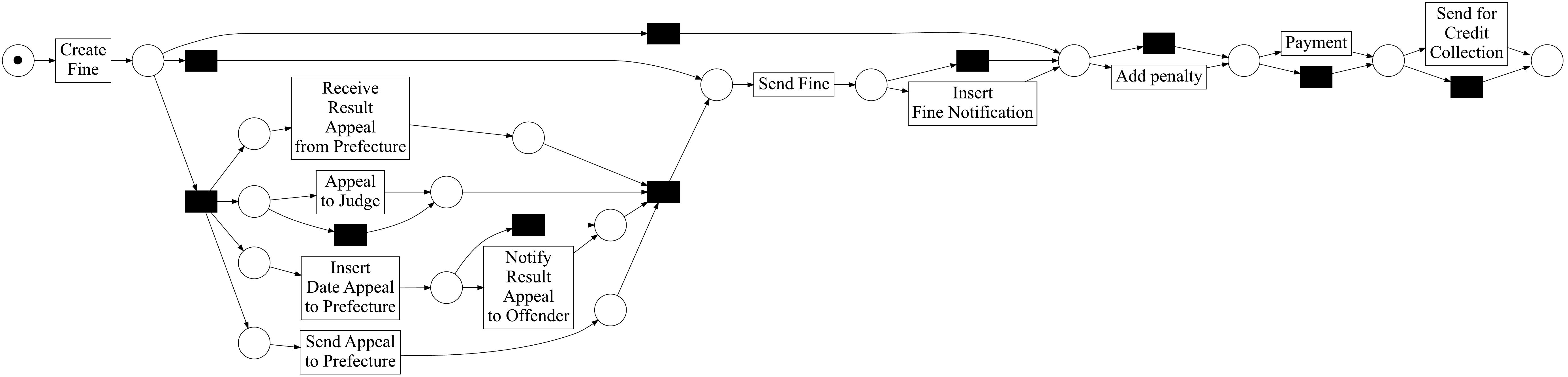}\hfill

\vspace{\figdist}
\footnotesize{eST-Miner (${\tau=1.0}$):}

\vspace{\figdist}
\includegraphics[width =0.56\linewidth]{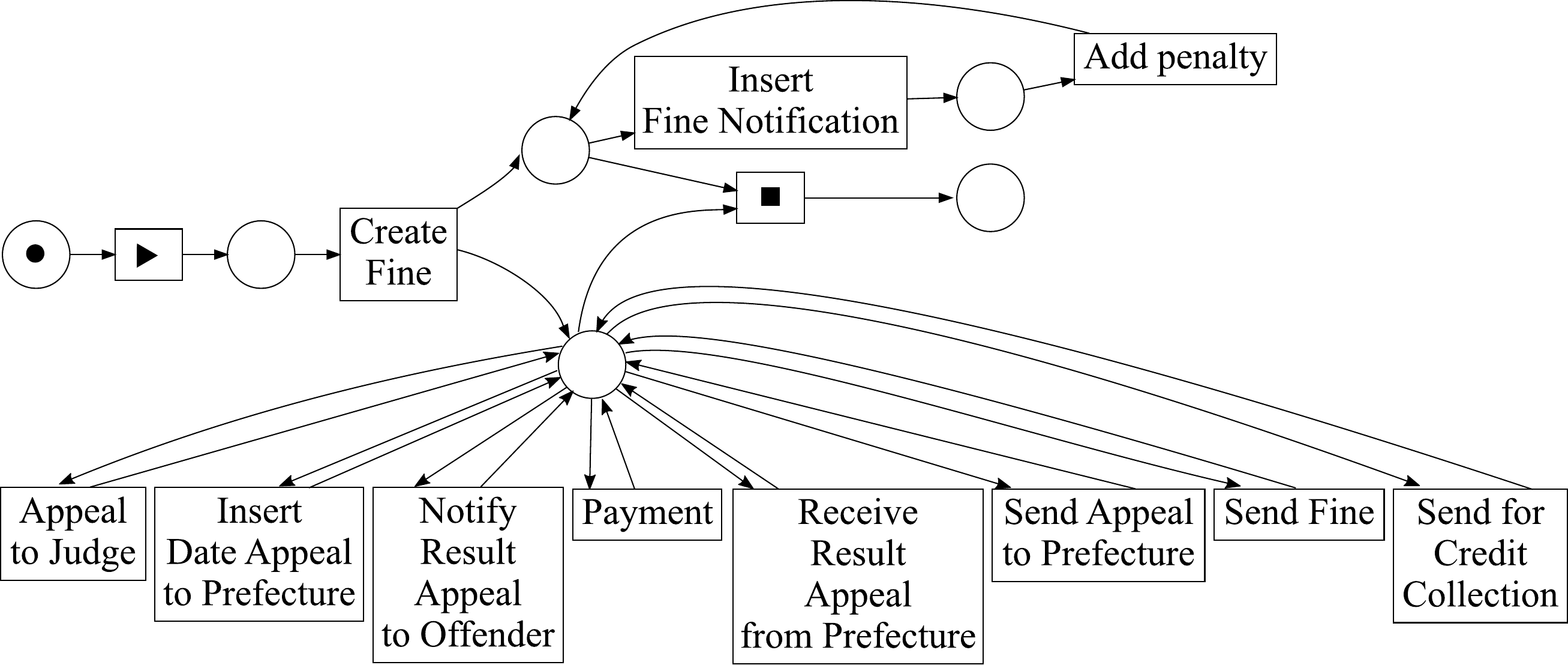}\hfill

\vspace{\figdist}
\footnotesize{Presented Approach: highest \texttt{HM} value out of the complete set of results, as well as the set of models without dead transitions.\vspace{2mm}}

\vspace{\figdist}
\includegraphics[width = \linewidth]{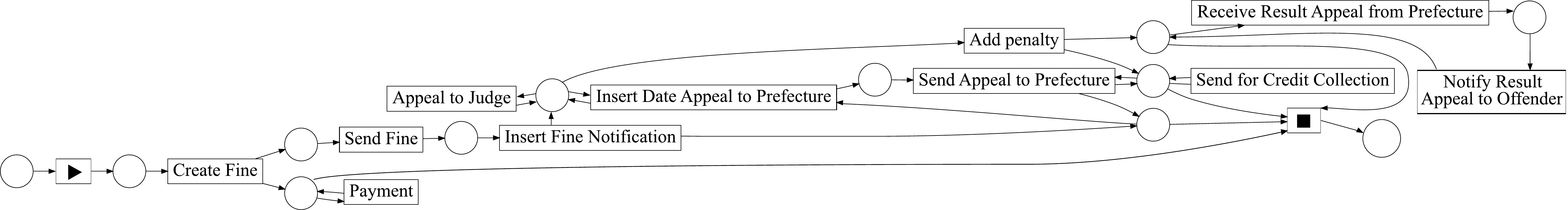}\hfill

\vspace{\figdist}
\footnotesize{Presented Approach: highest $F_1$-score out of the complete set of results.}

\vspace{\figdist}
\includegraphics[width =0.8\linewidth]{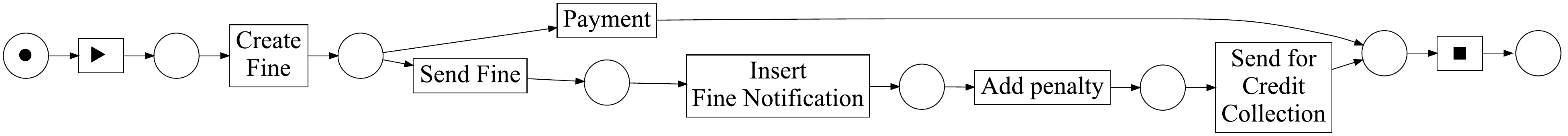}\caption{{The Petri nets discovered based on the \texttt{RTFM} log using the Inductive Miner infrequent (default settings), the eST-Miner with $\tau=1.0$, and a subset of interesting models discovered using the presented approach.}}
\label{fig:rtfmmodels}\vspace*{-3mm}
\end{figure}

\medskip
Figure~\ref{fig:rtfmmodels} shows Petri nets discovered from the \texttt{RTFM} log. Considering the models discovered by IMf and eST-Miner with $\tau=1.0$, we observe the same general tendencies as for the previous logs. For the model with the highest $F_1$-score discovered by our approach, we note that several activities are missing, meaning that they are not part of any replayable trace from the event log. The reason can be found by investigation of this particular event log, which describes two very distinct sub-processes, the more frequent of which consists of the activities still contained in the model. The activities of the infrequent sub-process related to appeals have been removed, allowing to focus on the main process. The model with the highest \texttt{HM} includes all activities from the event log. It includes the main process which is also expressed by the model with highest $F_1$, with additional self-loops that model optionality of those activities. Additionally, the control-flow of the appeal-related subprocess is included. Here, self-loops are not only used to model skippable activities but also to enforce a certain order on the events. Despite the limitations of using only uniquely labeled transitions, the positioning of the infrequent activities within the control-flow is precise and reflects their behavioral patterns in the event log.

\begin{figure}[!h]
\footnotesize{Inductive Miner infrequent (default settings):}

\vspace{\figdist}
\includegraphics[width =1.0\linewidth]{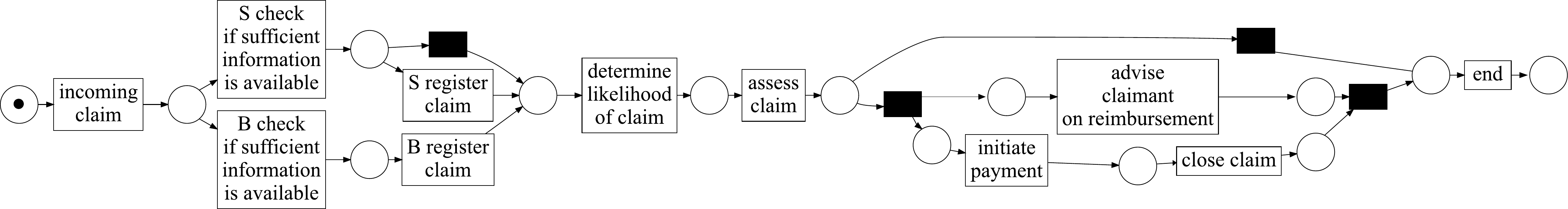}\hfill

\vspace{\figdist}
\footnotesize{eST-Miner (${\tau=1.0}$):}

\includegraphics[width =0.75\linewidth]{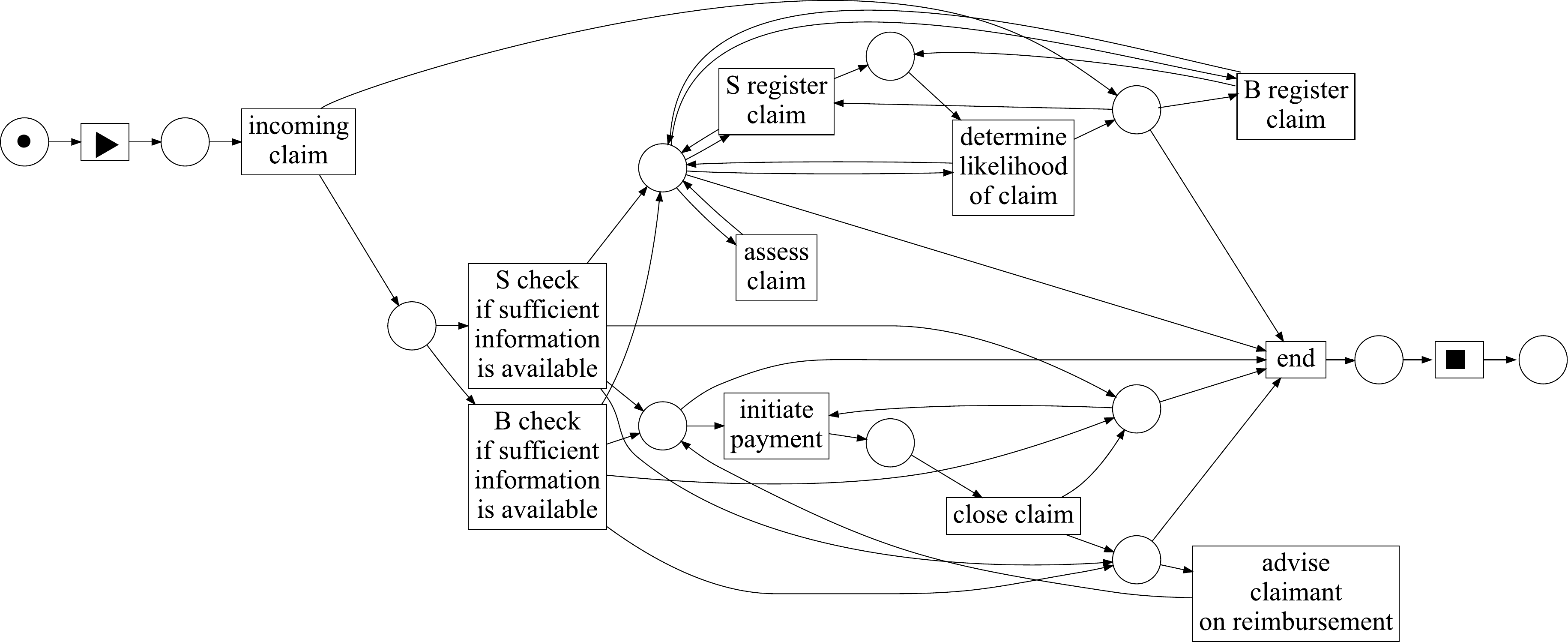}\hfill\vspace*{2mm}

\footnotesize{Presented Approach:  highest \texttt{HM} value and highest $F_1$-score out of the complete set of results, as well as the set of models without dead transitions.}

\vspace{\figdist}
\includegraphics[width =1.0\linewidth]{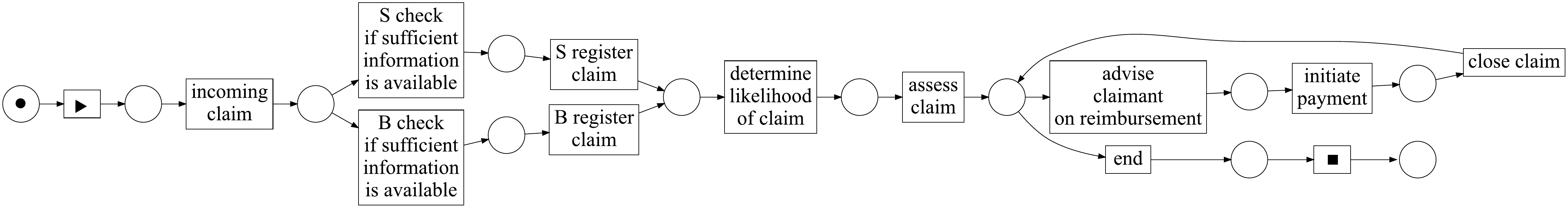}
\caption{{The Petri nets discovered based on the \texttt{Teleclaims} log using the Inductive Miner  infrequent (default settings), the eST-Miner with $\tau=1.0$, and and a subset of interesting models discovered using the presented approach.}}
\label{fig:teleclaimsmodels}\vspace*{-3mm}
\end{figure}

\medskip
A set of process models discovered from the \texttt{Teleclaims} log is presented in Figure~\ref{fig:teleclaimsmodels}. For this event log, the same model scores highest with respect to \texttt{HM} and $F_1$ for the complete set of discovered models as well as for the set of models without dead transitions. This model and the model discovered by IMf express similar behavior, with the main difference being the representation of skippable activities: with all transitions being uniquely labeled,  our approach has to rely on loop constructs rather than silent activities. The eST-Miner with $\tau=1.0$ does not abstract from infrequent behavior, which in this case results in a perfectly fitting but quite complex model.

Our results confirm the expectation that even minor gains in fitness are usually accompanied by a major drop in precision. The models with the best $F_1$-Score are usually those with the highest precision value. From  Figures~\ref{fig:Ordersmodels} to~\ref{fig:teleclaimsmodels} we can observe that these models seem to focus on the main process behavior, giving a clear representation of the control-flow of the main activities. However, they are likely not to incorporate infrequent activities. This is clearly illustrated by the \texttt{Orders} and \texttt{RTFM} event logs.  Here, infrequent but potentially vital paths in the control-flow (e.g. cancellation of an order or appeal against a fine) are revealed when the discovery algorithm abstracts from deviating behavioral patterns without ignoring infrequent activities themselves. Our results show that the presented approach is able to return models anywhere on the scale balancing fitness, precision and activity-coverage, based on the choice of parameters.

\subsection{Impact of parameter choices}

\begin{figure}[!ht]
    \centering
    \includegraphics[width =\linewidth, trim={0 14.2cm 0 0.25cm},clip]{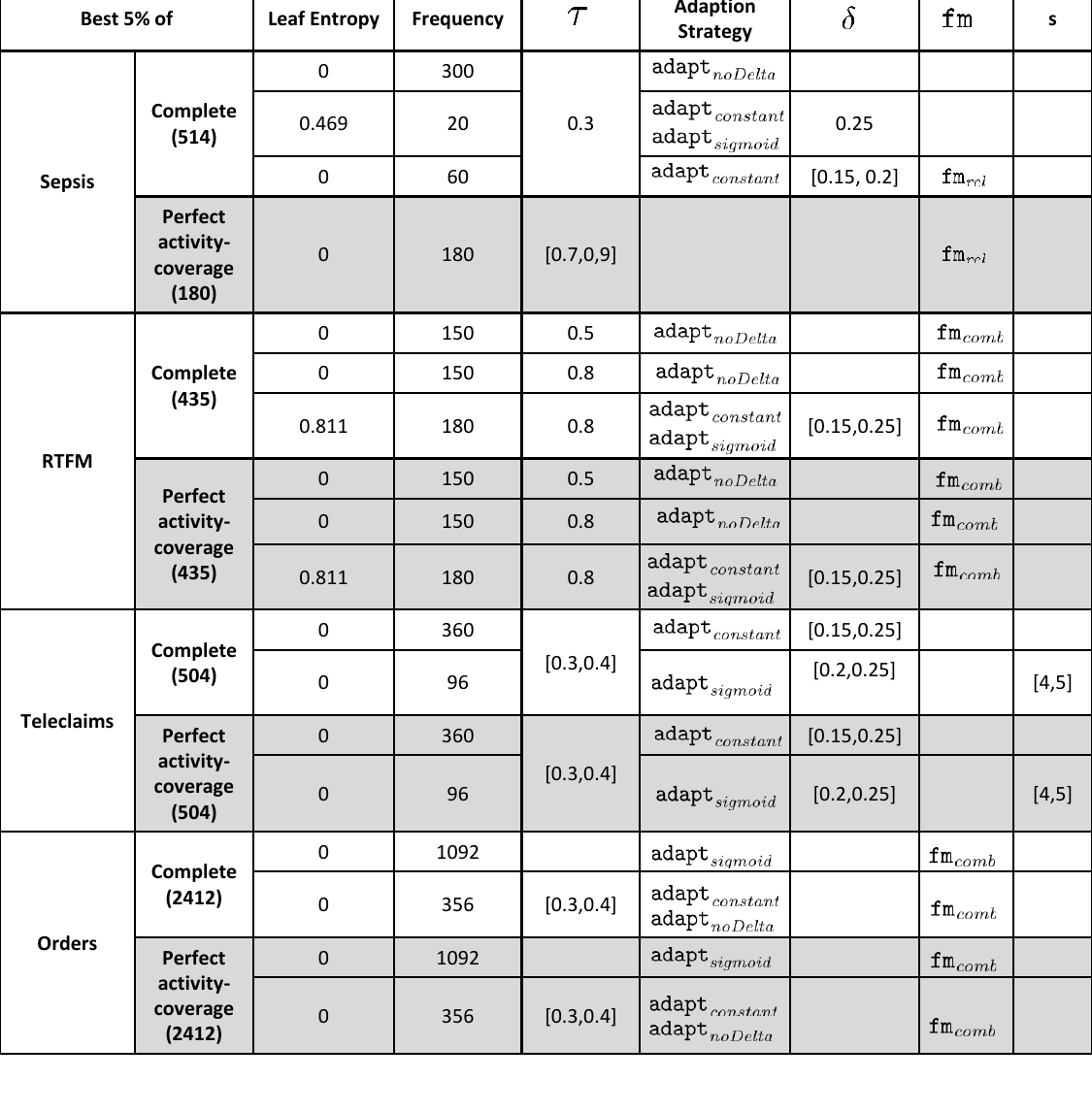}\vspace*{-5mm}
    \caption{{Overview of the parameter choices resulting in the discovery of the models with the top $0.05$ fraction of the \texttt{HM} value, once for all discovered models and once for the models with perfect activity-coverage. For each log, we indicate how often such models have been discovered in our experimentation. Each line refers to a set of parameter combinations, with the frequency and entropy of the corresponding leaf node in the decision tree. For each parameter that our decision tree analysis has revealed to be impactful, the possible values are indicated. Empty cells indicate that the corresponding parameter was insignificant for reaching the leaf. Leaves which did not include a significant number of positive instances are not included, resulting in a slight discrepancy between the number of such models discovered and the sum of leaf frequencies.}}
    \label{fig:bestparams}
\end{figure}

While the quality results discussed earlier clearly indicate that our approach is able to discover models balancing fitness, precision and activity-coverage while maintaining reasonable simplicity, the choice of parameters has a significant impact. Therefore, we investigate this further. We used decision tree analysis to search for parameter settings that would result in the highest quality models as indicated by the \texttt{HM}-Score. Since, for some event logs, the best model according to this metric was discovered very infrequently, we extended the notion to the top $5 \%$ of runs. However, in other logs the best model was discovered in significantly more than $5 \%$ of the runs. This discrepancy resulted in an imbalance of the absolute number of runs resulting in the highest-scoring model(s). Additionally, we performed the same analysis approach for the set of models without dead activities, again looking for the (at least) $5 \%$ of runs which resulted in models with the highest \texttt{HM} value.

\medskip
The results of this analysis are shown in Figure~\ref{fig:bestparams}, where each line represents a set of parameter combinations that leads to the discovery of the best model(s). Information on the frequency and entropy of the corresponding leaf in the decision tree is included as well.

\medskip
For \texttt{Sepsis} the overall best scoring model requires the lowest investigated value of $\tau$, that is $0.3$. Additionally, $\delta$ should either be ignored completely by using the $\texttt{adapt}_\textit{noDelta}$ adaption function or set to the largest value. Slightly lower values of $\delta$ paired with the less constraining relative fitness can be combined with $\texttt{adapt}_\textit{constant}$, i.e., always using the maximum value for $\delta$. All in all, the parameters indicate that rather restricting places with comparatively low fitness must be accepted to discover models with higher quality for the \texttt{Sepsis} log. This makes sense in the context of the large variety of traces without clear main behavior, i.e., all trace variants have similarly low frequency. However, the large \texttt{HM} value of the models discovered with those parameter settings is mostly based on  scoring high precision, which is also due to $5$ activities not being included. Choosing $\tau = 0.7$ or larger together with relative fitness balances fitness and precision such that the highest-scoring model that includes all log activities can be discovered.

\smallskip
For the \texttt{RTFM} event log the use of combined fitness is a prerequisite to discovering the best-scoring models. In contrast to  \texttt{Sepsis}, larger values of $\tau$ seem important: a choice $\tau=0.5$ or $\tau=0.8$ while ignoring $\delta$ ($\texttt{adapt}_\textit{noDelta}$) results in a high-quality model.  Alternatively, $\tau=0.8$ can be combined with large values of $\delta$. In this case, the highest scoring model in general and with full activity-coverage coincide.

\smallskip
For the \texttt{Teleclaims} event log the set of best-scoring models coincides with the set of best-scoring models without dead transitions. A $\tau$ value of at most $0.4$ is a prerequisite. Combining a $\delta \geq 0.15$ with the constant adaption function results in a high-quality model.  Gradually adapting $\delta$ using  sigmoid adaption requires larger values of $\delta$ and high steepness $s$.

\medskip
The set of best-scoring models in general and best-scoring models with perfect activity-coverage is the same for the \texttt{Orders} event log. To discover such models, one can combined  combined fitness either with the sigmoid adaption function or  with any other adaption function and low values for $\tau$.

For the four event logs investigated in this paper, the most important parameter seems to be $\tau$. This is not surprising, since $\tau$ has a direct impact on which places are available for addition to the Petri net. Furthermore, the combined fitness metric plays an important role in finding the best models for all event logs but \texttt{Sepsis}.
The constant and sigmoid adaption functions usually appear in combination with certain choices for $\delta$, which makes sense since $\delta$ is limiting the range of the adaption strategies, which include the use of $s$. Generally, higher values for delta seem to be favored to find good models.

Notably, the artificial tree depth $d^+$ as well as $\Qlength$ seem to have had no major impact on the discovery of any of the examined models (and therefore do not appear in Figure~\ref{fig:bestparams}). For the potential places queue $\Qlength$ this indicates that either the length limit was not reached or no important places were dropped. For $d^+$ we can conclude that in the scope of our experiments it did not result in significant complex places being added.

Some dependencies on log features are expected, and seem to be confirmed by the results in Figure~\ref{fig:bestparams}. For the \texttt{RTFM} log, which has a few very dominant trace variants, we seem to generally achieve good results for rather high values of $\tau$. In contrast, for the \texttt{Sepsis} log, which has a high variety of traces, a low $\tau$-value seems mandatory to achieve high scores. Most likely, the large variety of fitting places allows for obtaining high precision, while our heuristics seems to successfully ensure the focus on the main behavioral patterns.

Indeed, the results from the \texttt{Sepsis} log,  seem to confirm our algorithms ability to discover the main behavior hidden in an event log even in the absence of clear main trace variants: for a low value of $\tau$, e.g. $\tau=0.3$, the fraction of log traces replayable by the returned Petri net is close to $0.3$, however, the alignment-based fitness reliably remains above $0.9$, indicating that most of the traces are close to being replayable. We can conclude that the returned model successfully expresses the core behavior of the process.

\medskip
To summarize, the results clearly show that high-quality models balancing the different quality aspects can be discovered. There is a significant variance in some of the metrics, particularly precision, indicating that the settings of the algorithm have a notable impact. Our preliminary investigation shows that, based on the event log, certain parameter choices are likely to result in high-quality models. Choosing combined fitness significantly increases the likelihood of the discovered model to contain all or most of the log activities. Our experimentation and analysis give a first indication about which parameters have a more notable impact and whether certain settings are more suitable for logs with certain properties. However, we investigated only four event logs and clearly further experimentation needs to be performed to explore to which degree a generalization of our results is possible. Note that the impact of the candidate traversal order has not been investigated yet, and may allow for further improvements.

\subsection{Running time analysis}

In general, the worst-case running time of the eST-Miner is exponential in the number of log activities $A$, since $\mathcal{O}((2^{|\setofactivities}|)^2)$ candidate places may have to be evaluated (compare \cite{PN2019}).
However, limiting the tree depth to a \emph{fixed} value $2 \leq k \leq |\setofactivities|$, as was done in the experiments performed in the context of this work, can improve the worst-case running time.
In the complete candidate tree the depth of a place coincides with the number of activities connected to that place, i.e., for a place \inoutpair{I}{O} at depth $k$ we have that $|I|+|O|=k$.
The number of candidate places at depth $k$ corresponds to the number of all possible subsets of $A$ of size $k$ times all possibilities to split them over the sets of ingoing and outgoing transitions (for simplicity, we omit the insignificant border case of empty sets). For a fixed traversal depth of $k$, in the worst-case we visit all candidates on all levels from $2$ to $k$. Thus, for a fixed traversal depth $k$ the number of place candidates visited in the complete candidate tree is bounded by
\begin{align*}
&\sum_{i=2}^k \left (\binom{|\setofactivities|}{i} \cdot 2^i\right )
\leq k \cdot \left (\binom{|\setofactivities|}{k} \cdot 2^k \right )\\
\in \;&\mathcal{O}(k) \cdot \mathcal{O}(|\setofactivities|^k)\cdot \mathcal{O}(2^k)
\subseteq \; \mathcal{O}(|\setofactivities|^k) \; \; (\textit{ for fixed } k)
\end{align*}
We conclude that for a fixed traversal depth of $k$ the number of place candidates evaluated in the worst-case becomes polynomial in the number of activities.

\medskip
In our experimentation, we chose a fixed tree traversal depth of $5$ to make a large-scale experimentation more feasible. We tracked the running times of the various components of the proposed eST-Miner variant with the goal to verify that our newly introduced place selection subroutine does not significantly decrease the algorithm's performance. In Figure~\ref{fig:overviewTimes}, a summary of the running times of the different subroutines is shown on the left. Clearly, the fitness evaluation, i.e., the replay of the event log on each evaluated place candidate, makes up the bulk of the algorithms running time. The time needed for removing implicit places varies a lot over the runs (low values of $\tau$ tend to result in more places being inserted and thus more time needed for implicit place removal) but remains insignificant in most cases. Most importantly, in the context of this work, the time spend on place selection  is  even less than the time spend on computing the place candidates (tree traversal), as can be seen on the right-hand side of Figure~\ref{fig:overviewTimes} (note that the scale is different).
\begin{figure}[tbh]
    \centering
    \includegraphics[width =0.475\linewidth]{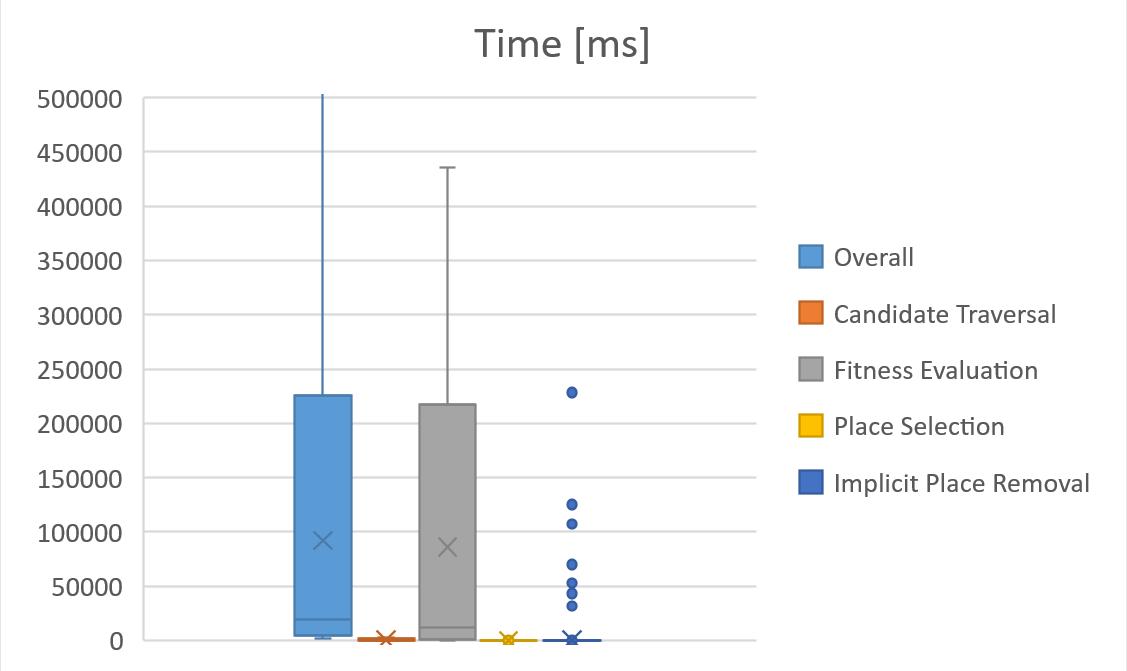}
    \includegraphics[width =0.475\linewidth]{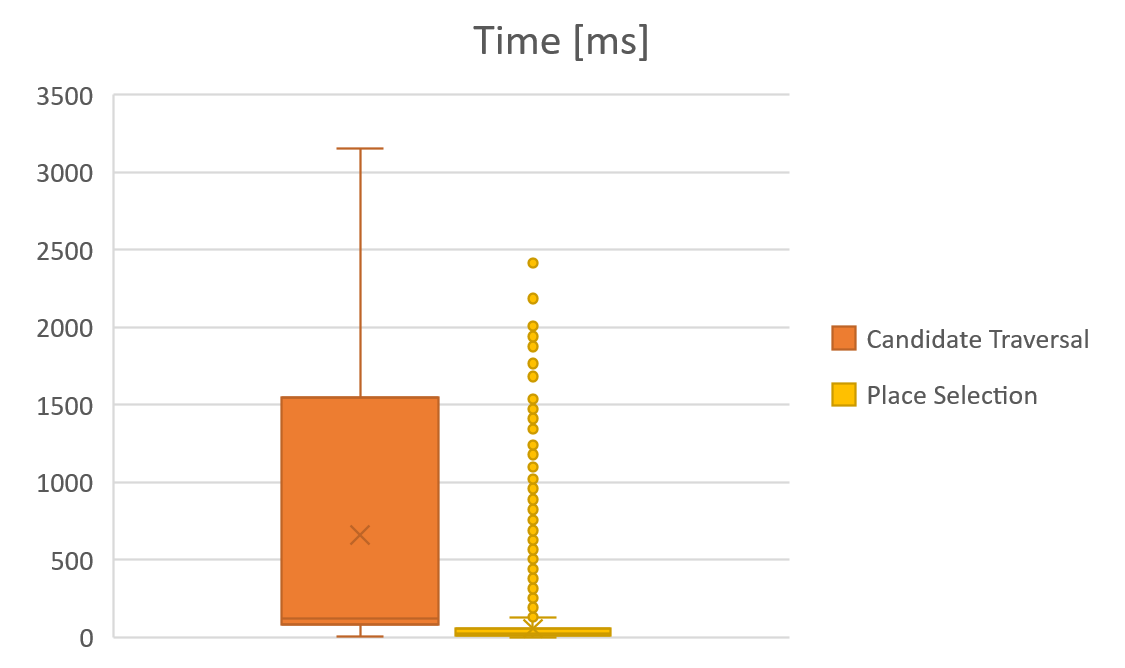}
    \caption{Overview of the running times achieved in our experiments on the left, with a smaller scale visualization of the shorter running times on the right. The newly added place selection subroutine does not add significantly to the overall running time of the algorithm.}
    \label{fig:overviewTimes}\vspace*{-5mm}
\end{figure}

\section{Discussion and future work}\label{sec:fm:discussion}
The algorithm presented in this work is an extension of the eST-Miner. The goal is to provide fitness guarantees on the returned Petri net and to abstract from infrequent behavioral patterns without categorically removing infrequent activities while preserving the advantages of the algorithm, such as its ability to reliably discover complex control-flow structures (e.g. long-term dependencies) and guarantee of returning a maximal set of places. The latter are closely related to achieving high precision.
The algorithmic framework we proposed allows for a lot of flexibility in a variety of components and only a subset of the possible options has been explored in this work. On the one hand we can imagine several extensions and refinements of the strategies proposed so far. On the other hand, a strategy to reduce the load of decision making on the user, such as a recommender system or simplified interface, becomes more pressing the more complex the extensions to the algorithm grow.

In Section~\ref{sec:fm:fitness}, we introduced a new fitness metric usable within the eST-Mining framework with the goal of avoiding the discovery of places that block infrequent activities. Aggregated fitness does not allow for any place to restrict an activity more than what is allowed according to the noise threshold, which guarantees that no activity can be blocked (dead) (if $\tau > 0$). Thus, activities get removed from the Petri net only if all the traces that contain them happen to be no longer replayable.
This strategy has the intended positive effect, i.e., our evaluation clearly shows that we succeed in the sense that the models discovered using combined fitness are much more likely to contain all activities from the event log. However, the decision tree analysis shows that we do not necessarily find the best-scoring models using combined fitness, in particular with respect to precision. While inclusion of infrequent but well-defined behavior has been the goal, achieving it results not only in the representation of infrequent but interesting behavioral patterns by the returned process models, but also activities without a well-defined control-flow remain part of the discovered model. The eST-Miner framework connects these activities using the most constraining places allowed, which in the worst-case result in an activity that is enabled by \sactivity, may loop, and gets disabled by \eactivity. An example for such structures can be seen in Figure~\ref{fig:sepsismodels} for the activities \textit{LacticAcid, Leucocytes} and \textit{CRP}. Obviously, such structures have a strong negative impact on precision, even though they seem reasonable in the context of the goals and constraints. {Less constraining fitness metrics allow for more restrictive places to be added, which may remove such activities or connect them in a more restrictive way, resulting in higher precision values.} Based on the assumption of our user being interested in all activities given in the event log, an investigation of a suitable adaption of aggregated fitness or a postprocessing step for more restrictive re-connection of such activities are promising future work.
Fitness metrics diverging significantly from what has been proposed so far may be useful depending on the quality aspects a potential user is interested in.  However, to utilize the performance optimization enabled by the eST-Miners complete candidate tree, they must satisfy the monotonicity properties discussed in Section~\ref{sec:fm:fitness}. Straightforward relaxations of aggregated fitness such as the average, mean or harmonic mean of the individual fitness of a place's transitions do not satisfy this requirement.

Another interesting topic to investigate is the removal of activities that are no longer part of the replayable event log. In particular, when using combined fitness, this straightforward approach to guarantee deadlock-freeness of the returned model is unnecessarily strict: the removed activity may not be the reason for the trace being unfitting and may not be included in a deadlock at all. Future work includes the investigation of alternative approaches for deadlock detection and prevention that may keep such activities as part of the model.

We proposed and evaluated several delta adaption strategies. Unfortunately, no reliable conclusions can be drawn from the experiments performed so far. In any case, improvements or variations of the adaption strategies are likely possible. In fact, there is a multitude of options for delta adaption functions which do not necessarily need to be based on the parameters of place complexity, but may incorporate any information available at the place level. Examples include all kinds of relations between or constraints on the connected activities, token behavior and/or attributes of related activities or replayable traces.  It would be particularly interesting to investigate to which degree the approach can be used to prioritize non-standard quality aspects, for example related to user interests such as compliance or performance.

There is room for improvement concerning the time performance of the algorithm as well. When an activity is determined to be no longer part of the replayable event log and is therefore removed, all subtrees in the complete candidate tree including this activity may be cut off without impacting the returned model. Furthermore, in this work, we apply the standard ILP-based approach to implicit place removal, which identifies implicit places based on the structure of the Petri net. Unfortunately, this approach becomes very time-consuming for larger sets of places. The replay-based approach introduced in \cite{ATAED2020} is much faster and can identify most implicit places. However, in the standard eST-Mining framework it does not verify whether all places required to make a currently evaluated place implicit are indeed present in the net, which is incompatible with the use of combined fitness in the framework proposed in this work. An adaption of this implicit place removal strategy may contribute to the overall performance of the proposed eST-Miner variant.
Related to performance, we have shown that the number of candidate places becomes polynomial when limiting the depth of a tree. Tighter bounds on the degree of the polynomial would be of interest.

\medskip
Finally, with a working solution to select subsets of fitting places such that high quality models without dead parts can be discovered, the remaining major limitation of the eST-Miner is its current inability to include silent or duplicate transition labels in the discovered Petri nets. With their added expressiveness, even better results with respect to fitness and, in particular, precision could be achieved.

\section{Conclusion} \label{sec:fm:conclusion}
In this paper, we proposed various extensions to the eST-Miner. We introduced a new fitness metric, aggregated fitness, investigated its properties and showed that it can be incorporated into the eST-Mining framework.  The goal of this metric is to improve the eST-Miner's place evaluation to avoid the discovery of places that prevent infrequent activities from being executed categorically, while maintaining its ability to abstract from infrequent behavioral patterns.
Furthermore, we propose a framework for place selection with the goal of guaranteeing that the discovered Petri net is free of deadlocks and satisfies a user-definable minimal fitness constraint. The approach employs heuristics to efficiently select a suitable subset of the discovered fitting places, while aiming towards high precision and simplicity. The algorithm is capable of discovering complex control-flow structures such as non-local dependencies. Furthermore, it is able to abstract from infrequent behavioral patters in the event log without simply filtering out infrequent activities or trace variants and to provide guarantees without over- or underfitting.

\medskip
Our experiments, using four different event logs, clearly show that not only is it possible to discover high-quality models using the introduced approach, but also the heuristics applied have a significant impact on the obtained Petri net.  Based on the parameter settings, models with a very different focus with respect to fitness, precision and the handling of infrequent behavior can be discovered. Some parameters have a stronger effect than others and some parameter choices seem to be more suitable for logs with certain properties, which should be verified by further experimentation.
A theoretical analysis of the running time, as well as an experimental overview of the time needed for the various subroutines of the algorithm, was presented, followed by a discussion of design decisions, open questions and future work.

\medskip
Besides the aspects discussed in detail in Section~\ref{sec:fm:discussion}, future work includes further experimentation to explore the generalization of the presented preliminary results, as well as the impact of the candidate place traversal order and its interaction with the heuristics used.
The dead transitions removed from the model because they are no longer part of the replayable event log give rise to further possible extensions of the eST-Miner. When detected early on, they can be used to identify and cut off candidate subtrees consisting of dead places to improve the running time. Further investigation into the cause of their removal may lead to better noise handling strategies to improve the quality of discovered models.
Finally, it would be interesting to investigate whether the presented place selection strategies can be adapted to improve other algorithms as well.

\flushleft
\begin{minipage}{0.6\columnwidth}
\subsection*{Acknowledgments:}
\noindent

We very much appreciate the time and effort of the reviewers, whose comments and suggestions contributed to improve the quality of this work.

Special thanks go to Tobias Brockhoff for his support with conducting the presented experiments.

The authors gratefully acknowledge the financial support by the Federal Ministry of Education and Research (BMBF) for the joint project Bridging AI (grant no. 16DHBKI023).

We thank the Alexander von Humboldt (AvH) Stiftung for supporting our research.
\end{minipage}~
\begin{minipage}{0.4\columnwidth}
\hspace*{10mm}\includegraphics[width=0.7\columnwidth]{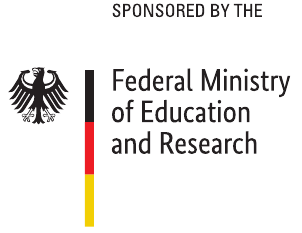}
\end{minipage}



\end{document}